\documentclass[11pt,a4paper,dvipsnames]{article}
\usepackage{cite}

\usepackage{amssymb}
\usepackage{amsmath}
\usepackage{amsthm}
\usepackage{xfrac}
\usepackage{latexsym}
\usepackage[dvips]{epsfig}
\usepackage{mathrsfs}
\usepackage{bm}
\usepackage{authblk}

\usepackage{xcolor}
\usepackage{tikz}
\usepackage{newtxtext}
\usepackage{array}
\usepackage{enumerate}
\usepackage{soul}

\theoremstyle{plain}
\newtheorem{proposition}{Proposition}
\newtheorem{lemma}{Lemma}

\newtheorem{remark}{Remark}
\newtheorem{example}{Example}

\allowdisplaybreaks

\def\bma{{\bm a}}

\def\bme{{\bm e}}

\def\bmg{{\bm g}}

\def\bmo{{\bm o}}

\def\bmzero{{\bm 0}}
\def\bmone{{\bm 1}}

\def\bmA{{\bm A}}
\def\bmB{{\bm B}}
\def\bmC{{\bm C}}
\def\bmD{{\bm D}}

\def\bmK{{\bm K}}

\def\bmQ{{\bm Q}}

\def\bmzeta{{\bm\zeta}}
\def\bmxi{{\bm \xi}}

\def\bmiota{{\bm \iota}}

\def\bmvarphi{{\bm \varphi}}

\def\bmvarepsilon{{\bm \varepsilon}}

\def\bmpartial{{\bm \partial}}

\DeclareMathAlphabet\mathbfcal{OMS}{cmsy}{b}{n}

\def\lAngle{\left\langle\!\left\langle}
\def\rAngle{\right\rangle\!\right\rangle}
\newcommand{\innerbrackets}[2]{\lAngle #1 , #2 \rAngle}

\DeclareMathAlphabet{\mathantt}{OT1}{antt}{li}{it}
\DeclareMathAlphabet{\mathpzc}{OT1}{pzc}{m}{it}

\usepackage{mathtools}
\usepackage{xparse}
\makeatletter
\newcommand{\raisemath}[1]{\mathpalette{\raisem@th{#1}}}
\newcommand{\raisem@th}[3]{\raisebox{#1}{$#2#3$}}
\makeatother
\NewDocumentCommand{\newrbar}{O{0pt} O{0pt}}{
  \ensuremath{\mathrlap{\raisemath{#2}{\hspace*{#1}{\mathchar'26\mkern-9mu}}}r}}

\usepackage{stmaryrd}

\usepackage{amsfonts, amscd}

\usepackage{fullpage}
\newcounter{mnotecount}

\newcommand{\mnotex}[1]
{\protect{\stepcounter{mnotecount}}$^{\mbox{\footnotesize $\bullet$\themnotecount}}$ 
\marginpar{
\raggedright\scriptsize\em
$\!\!\!\!\!\!\,\bullet$\themnotecount: #1} }

\newcounter{mnote}

\usepackage{color}
\usepackage{soul}
\usepackage{hyperref}

\begin{document}

\title{\textbf{The space spinor formalism and estimates for spinor fields}}

\author[1,3]{Mariem Magdy \footnote{E-mail
    address: {\tt mmagdy@perimeterinstitute.ca}}}

\author[2]{Juan A. Valiente Kroon \footnote{E-mail address: {\tt j.a.valiente-kroon@qmul.ac.uk}}}

\affil[1]{Perimeter Institute for Theoretical Physics, 31 Caroline Street North, Waterloo, Ontario, N2L 2Y5, Canada}

\affil[2]{School of Mathematical Sciences, Queen Mary, University of London, Mile End Road, London E1 4NS, United Kingdom.}

\affil[3]{Centro de An\'{a}lise Matem\'{a}tica, Geometria e Sistemas Din\^{a}micos, Instituto Superior T\'{e}cnico IST, Universidade de Lisboa UL, Avenida Rovisco Pais 1, 1049-001 Lisboa, Portugal.}

\maketitle
\begin{abstract}
We show how the space spinor formalism for 2-component spinors can be used to construct estimates for spinor fields satisfying first order equations. We  discuss the connection of the approach presented in this article with other strategies for the construction of estimates. In addition, we recast several concepts related to the notion of hyperbolicity in the context of spinor equations. The approach described in this article can be regarded as an adaptation to first order equations of the method of positive commutators for second order hyperbolic equations.  

\end{abstract}

\section{Introduction}
The 2-spinor formalism is a powerful tool to study the properties of null hypersurfaces in 4-dimensional spacetimes.  In addition, and somehow related to the latter, 2-spinors provide a convenient description of massless fields and of the curvature of spacetimes ---see e.g. \cite{PenRin84}. In particular, this formalism brings to the fore in a very precise manner the structural properties of the equations satisfied by the spinor fields ---foremost among these their \emph{hyperbolicity}.

\medskip
The equations arising 
in the description of  massless fields or in spinor formulations of the Einstein field equations are often of the form
\begin{equation}
\nabla^Q{}_{A'}\varphi_{Q A_2\cdots A_{p} B'_1\cdots B'_q} +F_{A'A_2\cdots A_{p} B'_1\cdots B'_q}{}^{Q_1\cdots Q_p Q'_1\cdots Q'_q} \varphi_{Q_1\cdots Q_p Q'_1\cdots Q_q}= f_{A'A_2\cdots A_p B'_1\cdots B'_q}
\label{PrototypeSpinorialEqn}
\end{equation}
for $p\geq 1$, $q\geq 0$ and 
where the valence $p+q$ spinor $\varphi_{A_1\cdots A_p A'_1\cdots A'_q}$ over the spacetime $(\mathcal{M},\bmg)$ is not assumed, for the time being, to have any specific symmetries. In the previous equation $\nabla_{AA'}$ is the spinor version of the Levi-Civita connection of the metric $\bmg$. We assume that the spinor field $F_{A'A_2\cdots A_{p} B'_1\cdots B'_q}{}^{Q_1\cdots Q_p Q'_1\cdots Q'_q}$ does not depend on the unknown $\varphi_{A_1\cdots A_p A'_1\cdots A'_q}$ so that equation \eqref{PrototypeSpinorialEqn} is, in principle, linear. Of course, this equation can form part of a larger system describing, for example, the frame and connection, in which case the larger system is nonlinear. 

A fundamental property of equation \eqref{PrototypeSpinorialEqn} is that it implies a symmetric hyperbolic system for the components of the spinor field $\varphi_{A_1\cdots A_p A'_1\cdots A'_q}$. The later, in turn, allows to construct estimates for suitable norms of the components of the unknown. These estimates are the key ingredient to establish the basic properties of solutions to the equation like existence, uniqueness and Cauchy stability. Although these observations are often cited facts, there is a lack of a systematic treatment, in the literature, of the structural PDE properties of equations like \eqref{PrototypeSpinorialEqn} which explicitly exploit its spinorial properties. Indeed, the standard strategy to analyse a spinorial equation is to project it with respect to a spin dyad so as to obtain a system of scalar equations. This approach leads to large expressions in which the algebraic properties of the equations are hard to identify and exploit. 

\medskip
This article adopts the point of view that equations like \eqref{PrototypeSpinorialEqn} are best analysed by retaining its spinorial structure. In order to pursue this strategy, one needs to bring additional structure into play. In particular, to have a meaningful definition of the norm of a spinor field, we assume the existence of an \emph{Hermitian structure} on the spacetime $(\mathcal{M},\bmg)$ ---this,  in addition to the implicit assumption of the existence of a spin structure which can be guaranteed if the spacetime is \emph{globally hyperbolic} so that one is dealing with an orientable and time orientable manifold; see e.g.  \cite{CFEBook}, Chapter 3. The existence of a Hermitian structure on the spinors in $(\mathcal{M},\bmg)$ is ensured by the existence of a timelike congruence of curves. In turn, timelike congruences are used to construct gauges to describe the evolution of solutions to field equations  ---see e.g. \cite{Fri95,Fri98c,Fri98b,Fri03b,LueVal09,LueVal13b,Bey07,Bey08,DouFra13,DouFra16,BeyDouFraWha12}.

Central for the present article is the observation that a Hermitian product allow to make use of a \emph{space spinor formalism} in which primed spinor indices are transformed into unprimed ones ---see e.g. \cite{Som80,Ash91,Fra98a} and also \cite{CFEBook}. Working with spinor objects with only one type of indices allows to employ the full machinery of irreducible decompositions and, as a result, one only needs to consider symmetric spinor fields. This approach leads to substantial simplifications in the manipulation of expressions and brings to the fore the key structural properties of the equations. In particular, as pointed out in \cite{FriRen00}, this strategy leads to an almost algorithmic procedure for hyperbolic reductions of systems of geometric equations ---something that can be more of an \emph{art} when working solely with tensors. An alternative approach to study systems of spinor equations by relying only on both symmetric spinor fields and operators has been developed in \cite{AksBac23}\footnote{In fact, the approach pursued in \cite{AksBac23} takes this logic to the ultimate conclusion by dispensing from the use of indices to denote the nature of spinors and the operations being applied to them.}.

\medskip
The present work is motivated by the \emph{positive commutator method} developed in \cite{Vas12,HinVas20} to construct the integrated estimates for weighted Sobolev norms that are needed to run the machinery of Melrose's School of microlocal analysis ---see also \cite{Wun13,HinVas26}. In its original form, this method is applied to tensor fields satisfying a second order hyperbolic equations. In this regard, in the present article, we address the following question: 

\smallskip
\emph{how can one adapt the strategy of the positive commutator method to construct estimates for spinor fields satisfying first order equations?}

\smallskip
As it will be shown in the main text, a natural first attempt is based on the observation that equation \eqref{PrototypeSpinorialEqn} readily implies a wave equation for the field $\varphi_{A_1\cdots A_pB_1'\cdots B'_q}$ which, in turn, leads to scalar wave equations for the independent (scalar) components of the spinor. From the point of view taken in this article, this approach is somehow unnatural ---a substantial part of the structure of the equations is lost by this \emph{reduction procedure}. Moreover, the resulting second order equations tend to be far lengthier than the original first order one.

The key observation in this article is that the formalism of space spinors provides an appropriate toolkit to address the question raised in the previous paragraph. In fact, and perhaps not so unsurprisingly, the construction of estimates for spinor equations put forward in this article is closely related to the hyperbolic reduction of the equation. \emph{The central idea behind this construction is that the hyperbolic reduction must be done with respect to a timelike vector field which is proportional to the vectorfield multiplier used in the construction of estimates.}

\medskip
The main motivation behind the construction described in this article is the analysis of the conformal Einstein field equations in a neighbourhood of spatial infinity by means of the techniques of \emph{Geometric Scattering theory} of Melrose in analogy to what is done in \cite{HinVas20}. A gauge adapted to the geometry of spatial infinity has been introduced in the seminal work \cite{Fri98a}. This gauge makes use of certain conformal invariants (\emph{conformal geodesics}) and a hyperbolic reduction using the space spinor formalism. The resulting conformal evolution system has a very clearly definite hierarchical structure in which the components of the rescaled Weyl spinor satisfy a (coupled) symmetric hyperbolic system while all the other remaining variables satisfy transport equations along the conformal geodesics.  It is expected that the methods developed in the present paper will provide the technical toolkit to analyse these equations. This problem will be analysed elsewhere.

\subsection*{Overview of the article}
This article is structured as follows: in Section \ref{Section:TwoSpinors}, we discuss certain aspects of the 2-spinor formalism which are essential for the analysis in this article and which are, to some extent, less known ---namely a detailed discussion of the irreducible decomposition of 2-spinors of arbitrary valence and the space spinor formalism. In Section \ref{Section:Hyperbolicity}, we discuss various aspects of the notion of hyperbolicity in the context of equations for spinorial fields. In particular, it is shown how the formalism of space spinors can be used to systematically obtain hyperbolic reductions of spinorial equations like \eqref{PrototypeSpinorialEqn}. Section \ref{General-prescription-estimates} is the main one of this article and provides a discussion of the method of positive commutators in the context of first order evolution spinorial equations. Section \ref{Section:SymmetricSpinorFields} provides a more detailed discussion of the construction of estimates for symmetric spinor fields. In Section \ref{Conclusions}, we provide some conclusions and outlook for the methods here presented. Finally, Appendix \ref{Appendix:ProofIrreducible decompositions} provides a brief discussion of the theorem behind the decomposition of spinors of arbitrary valence in terms of irreducible components.

\subsection*{Notations and conventions}

In the following, $(\mathcal{M},\bmg)$ will denote a 4-dimensional spacetime. The signature of the Lorentzian metric $\bmg$ is $(+---)$. We will make use of the abstract index formalism as discussed in \cite{PenRin84,PenRin86}. In particular, low-case Latin indices like $a,\,b,\,c,\ldots$ will be used as abstract spacetime indices while capital Latin letters like $A,\,B,\,C,\ldots$ will be used as abstract spinor indices. Greek low-case indices like $\mu, \,\nu,\, \lambda,\ldots $ will be used as coordinate indices. Finally, boldface indices will be used to denote components with respect to a basis.

Given a spin dyad $\{ \bmo, \bmiota \}$, where 
\begin{equation}
     \bmvarepsilon_{\bmzero}{}^{A}\equiv o^{A}, \qquad  \bmvarepsilon_{\bmone}{}^{A}\equiv \iota^{A}, \qquad  \bmvarepsilon^{\bmone}{}_{A}\equiv o_{A}, \qquad  \bmvarepsilon^{\bmzero}{}_{A}\equiv -\iota_{A},
    \label{Spin-dyad}
\end{equation}
the components of a spinor $\xi_{A}$ are given by
\begin{equation*}
    \xi_{\bmA} = \xi_{A} \bmvarepsilon_{\bmA}{}^{A}.
\end{equation*}

\medskip
The antisymmetric spinor (spinor metric) will be denoted by  $\epsilon_{AB}$. 
Its contravariant version  $\epsilon^{AB}$ is defined through the relation
\begin{equation*}
    \epsilon_{AB} \epsilon^{BC} = \delta_{B}{}^{C}.
\end{equation*}
\section{Two component spinors}
\label{Section:TwoSpinors}

The analysis in this article will make use of two component spinors (2-spinors) in a 4-dimensional spacetime $(\mathcal{M},\bmg)$ ---sometimes also called $\text{SL}(2,\mathbb{C})$ spinors. Throughout, we make use of the conventions in \cite{PenRin84} (see also \cite{Ste91,CFEBook}) and assume familiarity with the basic methods and techniques of this formalism. In this section, we expand on specific aspects of 2-spinors which will be required in our analysis ---namely, irreducible decomposition of 2-spinors and the space spinor formalism. 


\subsection{Irreducible decompositions}
A fundamental tool for the systematic study of equation \eqref{PrototypeSpinorialEqn} is the decomposition of spinors of arbitrary valence  into irreducible components. Central to the latter is the idea that \emph{only symmetric spinors matter} which is implied by the following seminal result ---cf. \cite{PenRin84}, Proposition (3.3.54):

\begin{proposition}
\label{Proposition:IrreducibleDecompositions}
Any spinor $\varphi_{A\cdots F}$ is the sum of a totally symmetric spinor $\varphi_{(A\cdots F)}$ and outer products of the antisymmetric spinor $\epsilon_{AB}$ with totally symmetric spinors of lower valence. 
\end{proposition}

\begin{remark}
{\em For simplicity the above result is only stated for spinors with covariant unprimed indices. An analogous result holds for more general spinors.}
\end{remark}

As some of the computations in  the proof of this result will be used repeatedly in the sequel, we provide a brief discussion of the argument in Appendix \ref{Appendix:ProofIrreducible decompositions}.

\subsubsection{A more detailed look at the decomposition in irreducible terms}
\label{Subsection:IrreducibleDecomposition}

In the subsequent discussion, we will require more information than the one provided by Proposition \ref{Proposition:IrreducibleDecompositions}. To this end we introduce some further notation. Given a nonnegative integer $n$, an arbitrary spinor of valence $n$ with only unprimed covariant indices will be written as $\varphi_{A_1\cdots A_n}$.

\medskip
Now, let
\[
I_n \equiv \{1,\,2,\ldots, n  \}.
\]
Given a further integer $k$, $0\leq k\leq n$, we will be interested in the \emph{partial permutations} of $k$ elements of $I_n$ ---that is, ordered lists of $k$ elements of $I_n$. Let us denote the set of those partial permutations by $S^k_n$. For $k=n$ one recovers the symmetric set ---i.e. the set of  permutations of $I_n$. The elements of $S^k_n$ can be thought of as (ordered) $k$-tuples of elements of $I_n$. In this spirit, we write
\[
(j_1,\,j_2,\ldots ,j_k)\in S^k_n.
\]
Of particular relevance for the irreducible decomposition of $\varphi_{A_1\cdots A_n}$ are those partial permutations where $k$ is even so that one writes $k=2\ell$ 
and which satisfy
\[
j_1<j_2, \quad j_3<j_4, \quad \cdots \quad,  \,j_{2\ell-1}< j_{2\ell}.
\]
We denote the set of this type of partial permutations by $E_n^{2\ell}$ (\emph{pairwise lexicographic partial permutation}) and write
\[
(j_1,\, j_2 \,\ldots, j_{2\ell-1},\, j_{2\ell})\in E_n^{2\ell}.
\]
Finally, given $(j_1,\, j_2 \,\ldots, j_{2\ell-1},\, j_{2\ell})\in E_n^{2\ell}$ let $I_n\setminus\{ j_1,\, j_2 \,\ldots, j_{2\ell-1}, j_{2\ell}\}$ denote  the set of elements in $I_n$ not appearing in the partial permutation under consideration. With the above notation let 
\[
\varphi^{(j_1,j_2,\ldots,j_{2\ell})}_{A_{i_1}\cdots A_{i_{n-2\ell}}}, \qquad \{ i_1,\ldots, i_{n-2\ell}\} = I_n\setminus\{ j_1,\, j_2 \,\ldots, j_{2\ell-1},\, j_{2\ell}\},
\]
be defined as the totally symmetric spinor obtained by contraction of the pairs of indices
\[
\{A_{j_1}, \, A_{j_2}\}, \quad \{ A_{j_3},\, A_{j_4}\}, \quad \cdots \quad \{ A_{2\ell-1},\, A_{2\ell}\}
\]
and then symmetrising the leftover indices $\{A_{i_1}, \ldots A_{i_{n-2\ell}}\}$. We will adopt the convention that contractions are carried out in the Southwest to Northeast direction. That is, we have that 
\[
\varphi^{(j_1,j_2,\ldots,j_{2\ell})}_{A_{i_1}\cdots A_{i_{n-2\ell}}} \equiv \varphi_{(A_{i_1} \cdots |A_{j_1}|\cdots }{}^{A_{j_1}}{}_{\cdots |A_{j_3}|}{}_{\cdots}{}^{A_{j_3}}{}_{\cdots \cdots |A_{2\ell-1}|\cdots}{}^{A_{2\ell-1}}{}_{ \cdots A_{i_{n-2\ell}})}.
\]
The spinors $\varphi^{(j_1,j_2,\ldots,j_{2\ell})}_{A_{i_1}\cdots A_{i_{n-2\ell}}}$ for all possible $(j_1,\, j_2 \,\ldots, j_{2\ell-1},\, j_{2\ell})\in E_n^{2\ell}$ with $2\ell\leq n $ are the \emph{irreducible components} of $\varphi_{A_1\cdots A_n}$ discussed in Proposition \ref{Proposition:IrreducibleDecompositions}. Observe that if $\ell=0$ then $E^0_n =\varnothing$, so that 
\[
\varphi^\varnothing_{A_1\cdots A_n} = \varphi_{(A_1\cdots A_n)}
\]
is the totally symmetric part of $\varphi_{A_1\cdots A_n}$ ---i.e. there are no contractions.

\begin{remark}
{\em If the original spinor $\varphi_{A_1\cdots A_n}$ is symmetric over a certain subset of its indices then some of the irreducible components $\varphi^{(j_1,j_2,\ldots,j_{2\ell})}_{A_{i_1}\cdots A_{i_{n-2\ell}}}$ will vanish.}
\end{remark}

\medskip
With the aid of the above notation, one can obtain a more detailed version of Proposition \ref{Proposition:IrreducibleDecompositions}. Namely, one has the following:
\begin{proposition}
\label{Proposition:IrreducibleDecompositionsALT}
A valence $n$ spinor $\varphi_{A_1\cdots A_n}$ admits the decomposition
\begin{equation}
\varphi_{A_1\cdots A_n} =\sum_{2\ell\leq n } \sum_{(j_1,\ldots,j_{2\ell})\in E^{2\ell}_{n}}  \epsilon_{A_{j_1}A_{j_2}}\cdots \epsilon_{A_{j_{2\ell-1}}A_{j_{2\ell}}}\phi^{[j_1,\cdots,j_{2\ell}]}_{A_{i_1}\cdots A_{i_{n-2\ell}}},
\label{GeneralDecompositionPhi}
\end{equation} 
where
\[
\phi^{[j_1,\cdots,j_{2\ell}]}_{A_{i_1}\cdots A_{i_{n-2\ell}}} \equiv \sum_{(k_1,\ldots,k_{2\ell})\in E^{2\ell}_{n}} \mathfrak{c}^{j_1\cdots j_{2\ell}}_{k_1\cdots k_{2\ell}}\varphi^{(k_1,\cdots,k_{2\ell})}_{A_{i_1}\cdots A_{i_{n-2\ell}}}
\]
with $\mathfrak{c}^{j_1\cdots j_{2\ell}}_{k_1\cdots k_{2\ell}}$ some numerical coefficients.
\end{proposition}

\begin{proof}
The expansion follows from the recursive application of the arguments of the proof of Proposition \ref{Proposition:IrreducibleDecompositions} as discussed in Appendix \ref{Appendix:ProofIrreducible decompositions}.
\end{proof}

\begin{remark}
{\em Observe that the spinors $\phi^{[j_1,\cdots,j_{2\ell}]}_{A_{i_1}\cdots A_{i_{n-2\ell}}}$ are irreducible since they are linear combinations of the irreducible components $\varphi^{(j_1,\cdots,j_{2\ell})}_{A_{i_1}\cdots A_{i_{n-2\ell}}}$. The explicit value of the coefficients $\mathfrak{c}^{j_1\cdots j_{2\ell}}_{k_1\cdots k_{2\ell}}$ will not be required in the sequel. In particular, as it can see from specific examples, many of the coefficients can be zero. It can be verified that knowledge of the spinors $\phi^{[j_1,\cdots,j_{2\ell}]}_{A_{i_1}\cdots A_{i_{n-2\ell}}}$ is equivalent to knowledge of $\varphi^{(k_1,\cdots,k_{2\ell})}_{A_{i_1}\cdots A_{i_{n-2\ell}}}$ ---the latter can be obtained from the former by solving a linear algebraic system of equations. In other words, the two sets of spinors $\phi^{[j_1,\cdots,j_{2\ell}]}_{A_{i_1}\cdots A_{i_{n-2\ell}}}$ and $\varphi^{(k_1,\cdots,k_{2\ell})}_{A_{i_1}\cdots A_{i_{n-2\ell}}}$ correspond to different bases of irreducible components.}
\end{remark}

\begin{remark}
{\em In the case of spinors with given symmetries, specific components $\phi^{[j_1,\cdots,j_{2\ell}]}_{A_{i_1}\cdots A_{i_{n-2\ell}}}$ will vanish. }
\end{remark}

\subsubsection{Some examples}
It is illustrative to provide some concrete examples of the above expressions.

\begin{example}
{\em Given a valence 2 spinor $\varphi_{AB}$ one has the well-known decomposition
\[
\varphi_{AB}= \varphi_{(AB)} + \frac{1}{2}\epsilon_{AB}\varphi_Q{}^Q.
\]
The second term corresponds to the sum over the elements of  $E_2^2=\{(1,2) \}$.
Observe that if the spinor is symmetric then $\varphi_Q{}^Q$. Comparing with the general expression \eqref{GeneralDecompositionPhi} one has that
\[
\phi^\varnothing_{AB}=\varphi_{(AB)}, \qquad \phi^{[1,2]}=\varphi_Q{}^Q.
\]}
\end{example}
Therefore, in this case, $\mathfrak{c}^{12}_{12}= \frac{1}{2}$.

\begin{example}
{\em For an arbitrary valence 3 spinor $\varphi_{ABC}$ one has that
\[
\varphi_{ABC} = \varphi_{(ABC)} + \frac{1}{6}\epsilon_{AB} \big(\varphi_{QC}{}^Q + \varphi_Q{}^Q{}_C) +\frac{1}{6}\epsilon_{AC}\big(\varphi_{QB}{}^Q+ \varphi_Q{}^Q{}_B\big) + \frac{1}{2}\epsilon_{BC}\varphi_{AQ}{}^Q. 
\]
Comparing the above expression with the general expression \eqref{GeneralDecompositionPhi} one has that
\[
E_3^2 =\{ (1,2),\,(1,3),\, (2,3) \}.
\]
Moreover, comparing with \eqref{GeneralDecompositionPhi} one finds that 
\[
\phi^{[1,2]}_C= \frac{1}{6}\big(\varphi_{QC}{}^Q + \varphi_Q{}^Q{}_C), \qquad \phi^{[1,3]}_B = \frac{1}{6}\big(\varphi_{QB}{}^Q+ \varphi_Q{}^Q{}_B\big), \qquad \phi^{[2,3]}_A =  \frac{1}{2}\varphi_{AQ}{}^Q.
\]
}
\end{example}
So, we have
\begin{equation*}
    \mathfrak{c}^{12}_{12} = \mathfrak{c}^{12}_{13} = \mathfrak{c}^{13}_{12} = \mathfrak{c}^{13}_{13}  = \frac{1}{6}, \qquad \mathfrak{c}^{12}_{23} = \mathfrak{c}^{13}_{23} = \mathfrak{c}^{23}_{12} = \mathfrak{c}^{23}_{13} =0, \qquad \mathfrak{c}^{23}_{23} = \frac{1}{2}.
\end{equation*}

\begin{example}
{\em The decomposition in irreducible components of a valence 4 spinor $\varphi_{ABCD}$ is given by
\begin{eqnarray*}
&& \varphi_{ABCD} =  \varphi_{(ABCD)}  \\
&& \hspace{2cm} +  \frac{1}{2} \epsilon_{CD} \varphi_{(AB)}{}_{F}{}^{F} + 
 \frac{1}{6} \epsilon_{BD} \big(\varphi_{(A|F|C)}{}^{F}+\varphi_{(A|F|}{}^{F}{}_{C)}\big) +  \frac{1}{6} \
\epsilon_{BC} \big(\varphi_{(A|F|D)}{}^{F} + \varphi_{(A|F|}{}^{F}{}_{D)}\big)\\
&& \hspace{2cm} +  \frac{1}{12} \epsilon_{AD} \big(
\varphi_{F}{}_{(BC)}{}^{F}+\varphi_{F}{}_{(B}{}^{F}{}_{C)}+\varphi_{F}{}^{F}{}_{(BC)}\big) +  \frac{1}{12} \epsilon_{AC}\big( \varphi_{F}{}_{(BD)}{}^{F}+ \varphi_{F}{}_{(B}{}^{F}{}_{D)}+  \varphi_{F}{}^{F}{}_{(BD)}\big) \\
&& \hspace{2cm}
+  \frac{1}{12} \epsilon_{AB} \big(\varphi_{F}{}_{(CD)}{}^{F}+\varphi_{F}{}_{(C}{}^{F}{}_{D)}+\varphi_{F}{}^{F}{}_{(CD)}\big)\\
&& \hspace{2cm} +\frac{1}{12} \big( \varphi_{FG}{}^{FG} + \varphi_{FG}{}^{GF}\big) \epsilon_{AD} \
\epsilon_{BC} + \frac{1}{12} \big( \varphi_{FG}{}^{FG}+ \varphi_{FG}{}^{GF}\big) \epsilon_{AC} \
\epsilon_{BD} + \frac{1}{4} \varphi_{F}{}^{F}{}_{G}{}^{G} \epsilon_{AB} \
\epsilon_{CD}.
\end{eqnarray*}
In this case, the sums in the general expression \eqref{GeneralDecompositionPhi} are carried over the sets
\begin{eqnarray*}
&& E_4^2 =\{ (1,2), \,(1,3), \,(1,4), \,(2,3), \,(2,4), \,(3,4) \},\\
&& E_4^4 =\{(1,2,3,4),\, (1,3,2,4),\, (1,4,2,3)\}.
\end{eqnarray*}
The identification with the spinor appearing in the formula \eqref{GeneralDecompositionPhi} can be readily done. For example, one has that
\[
\phi^{[2,3]}_{AB} = \frac{1}{2}\varphi_{(AB)Q}{}^Q, \quad \phi^{[1,2,3,4]} = \frac{1}{4}\varphi_F{}^F{}_G{}^G, \quad \mbox{etc.}
\]
}
\end{example}
The coefficients $\mathfrak{c}^{j_1\cdots j_{2\ell}}_{k_1\cdots k_{2\ell}}$ can be deduced easily from the expression of the decomposition. 
\subsection{The space spinor formalism}
\label{Section:Space-spinor-formalism}
In this article, we will make extensive use of the space spinor formalism as described in e.g. \cite{CFEBook}, Chapter 4 ---see also \cite{Som80,Ash91}.

\medskip
Let $\tau^a$ denote a \emph{timelike} vector with normalisation
\begin{equation}
    \tau_a \tau^a =2.
    \label{normalisation-tau}
\end{equation}
Let $\tau^{AA'}$ denote the spinor counterpart of $\tau^a$. By definition, the spinor $\tau^{AA'}$ is Hermitian. In the following, we will consider spin dyads $\{o^A, \iota^A \}$ with $o_A\iota^A =1$ adapted to $\tau^{AA'}$ in the sense that
\[
\tau^{AA'} =o^A \bar{o}^{A'} +\iota^A \bar{\iota}^{A'},
\]
where $\{ \bar{o}^{A'}, \bar{\iota}^{A'} \}$ denote the complex conjugates of $\{o^A, \iota^A \}$. The above expression ensures that 
\begin{equation}
\tau_{AA'}\tau^{BA'}=\delta_A{}^B.
\label{TauTauEpsilon}
\end{equation}

\subsubsection{Space spinor counterparts}  Relation \eqref{TauTauEpsilon}  defines an isomorphism between spinors with primed and unprimed indices. Namely, given $\mu_{A'}$ one defines the spinor $\mu_{A}$, its space spinor counterpart, via the relation 
\[
\mu_{A}\equiv \tau_{A}{}^{A'}\mu_{A'}.
\]
Observe that 
\[
\tau^A{}_{A'}\mu_A = -\mu_{A'},
\]
so that $\mu_{A'}$ and $\mu_A$ contain the same information. The above ideas can be extended in the obvious way to higher valence spinors with arbitrary valence. For example, for the spinor $\varphi_{A_1\cdots A_p A'_1\cdots A'_q}$ in equation \eqref{PrototypeSpinorialEqn}, one has that its space spinor counterpart is given by
\begin{equation}
\varphi_{A_1\cdots A_p B_1\cdots B_q}\equiv \tau_{B_1}{}^{A'_1}\cdots \tau_{B_q}{}^{A'_q}\varphi_{A_1\cdots A_p A'_1\cdots A'_q}.
\label{Definition:SpaceSpinorPhi}
\end{equation}
One of the advantages of working with the space spinor counterpart $\varphi_{A_1\cdots A_p B_1\cdots B_q}$ is that one can make use of the machinery of irreducible decompositions over the whole set of indices $\{A_1,\ldots,A_p,B_1,\ldots B_q \}$  rather than independently on the subsets of unprimed and primed indices, $\{A_1,\ldots,A_p\}$ and $\{A'_1,\ldots A'_q \}$, respectively.

\subsubsection{Hermitian conjugation}
Key for the purposes of this article is the observation that $\tau^{AA'}$ induces a Hermitian structure. More precisely, given a spinor $\kappa_A$, one defines its \emph{Hermitian conjugate} $\widehat{\kappa}_A$ as
\begin{equation}
    \widehat{\kappa}_A \equiv \tau_A{}^{A'}\bar{\kappa}_{A'}.
    \label{Hermitian-conjugate}
\end{equation}
One extends the operation of Hermitian conjugation to higher valence spinors with unprimed indices in the obvious manner. Given (say) a  valence $m$ symmetric spinor, one can reverse the definition of Hermitian conjugation making use of equation \eqref{TauTauEpsilon} to obtain
\begin{equation}
\bar{\varphi}_{A'_1\cdots A'_m} =(-1)^m \tau^{P_1}{}_{A'_1}\cdots \tau^{P_m}{}_{A'_m}\widehat{\varphi}_{P_1\cdots P_m}.
\label{HermitianConjugateToVanillaConjugate}
 \end{equation}

\subsubsection{Space spinor decomposition of the covariant derivative}
In the following, it will be convenient to work with the space spinor counterpart of $\nabla_{AA'}$, which is defined by
\begin{equation}
\nabla_{AB} \equiv \tau_B{}^{A'}\nabla_{AA'}.
\label{Definition:SpaceSpinorCD}
\end{equation}
which can be decomposed to a symmetric and antisymmetric parts
\begin{equation}
\nabla_{AB} = \mathcal{D}_{AB}+ \frac{1}{2}\epsilon_{AB}\mathcal{D}.
\label{Decomposition:SpaceSpinorCD}
\end{equation}
where $\mathcal{\bmD}_{AB}$ is known as the \emph{Sen connection} and $\mathcal{\bmD}$ as the \emph{Fermi derivative}, written explicitly as
\[
    \mathcal{D}_{AB} \equiv \tau_{(B}{}^{A'} \nabla_{A) A'}, \qquad \mathcal{\bmD} \equiv \tau^{AA'} \nabla_{AA'}.
\]

\subsubsection{The Weingarten spinor}
Derivatives of the Hermitian spinor $\tau^{AA'}$ will play an important role in our analysis. The derivatives are encoded in the \emph{Weingarten spinor} $\chi_{AA'BB'}$ which we define as
\begin{equation}
    \chi_{AA'BB'}\equiv \frac{1}{\sqrt{2}}\nabla_{AA'}\tau_{BB'}.
    \label{Definition-chi}
\end{equation}
The space spinor counterpart of $\chi_{AA'BB'}$, denoted by $\chi^{}_{ABCD}$, is then defined by
\begin{eqnarray*}
    && \chi^{}_{ABCD} \equiv \frac{1}{\sqrt{2}} \tau_{D}{}^{C'} \nabla_{AB} \tau_{CC'}, \\
    && \phantom{\chi^{}_{ABCD}} = \chi^{}_{(AB)CD} + \frac{1}{2} \epsilon_{AB} \chi^{}_{CD},
\end{eqnarray*}
where 
\begin{equation*}
    \chi^{}_{(AB)CD} \equiv \frac{1}{\sqrt{2}} \tau_{D}{}^{C'} \mathcal{D}_{AB} \tau_{CC'}, \qquad \chi^{}_{AB} \equiv \frac{1}{\sqrt{2}} \tau_{B}{}^{A'} \mathcal{\bmD}\tau_{AA'}.
\end{equation*}
By definition, $\chi^{}_{(AB)CD}$ and $\chi^{}_{AB}$ can be shown to satisfy the reality properties 
\begin{equation*}
    \widehat{\chi}_{(AB)CD} = \chi^{}_{(AB)CD}, \qquad \widehat{\chi}_{AB} = - \chi^{}_{AB},
\end{equation*}
and the symmetries
\begin{equation*}
    \chi^{}_{(AB)CD} = \chi^{}_{(AB)(CD)}, \qquad \chi^{}_{AB} = \chi^{}_{(AB)}.
\end{equation*}
In terms of $\chi_{(AB)CD}$ and $\chi_{AB}$, we can write the irreducible decomposition for the derivative of $\tau^{AA'}$ as
\begin{equation}
    \nabla_{AA'}\tau_{BB'} = - \frac{1}{2} \chi_{BF} \tau_{AA'} \tau^{F}{}_{B'} + \sqrt{2} \chi_{(AF)BG} \tau^{F}{}_{A'} \tau^{G}{}_{B'}.
    \label{Decomposition-derivative-tau}
\end{equation}

\section{Hyperbolicity and spinorial equations}
\label{Section:Hyperbolicity}

In this section, we discuss the notion of hyperbolicity for spinorial equations. The fundamental observation is that the 2-spinor formalism brings to the fore the structural properties underlying the hyperbolicity of the equations. In particular, while equation \eqref{PrototypeSpinorialEqn} is, strictly speaking not hyperbolic, it readily implies a symmetric hyperbolic system.

\subsection{Hyperbolicity and spinors}
In this section we consider some basic ideas regarding the notion of hyperbolicity in the context of spinorial equations.

\begin{remark}
{\em For simplicity, the discussion in this section is restricted to linear evolution equations. Most of the notions here presented can be extended to the quasilinear setting.}
\end{remark}

\subsubsection{The symbol}
Recall that the \emph{principal part} of a linear differential operator consists of the terms in the operator with the highest order derivatives. In turn, the \emph{symbol} of the operator at a point $p\in\mathcal{M}$, $\sigma_\bmxi$, is obtained by replacing the derivatives $\partial_\mu$ with respect to some local coordinates $x=(x^\mu)$ with a covector $\xi_\mu\in T^*_p\mathcal{M}$. In the following, we identify $T^*_p\mathcal{M}$ with $\mathfrak{H}^*_p\mathcal{M}$ ---the space of Hermitian valence 2 covariant spinors. In this spirit, the symbol of a spinorial equation  can be obtained by formally replacing the spinor covariant derivatives $\nabla_{AA'}$ in the principal part by the components $\xi_{\bmA\bmA'}$ of a Hermitian spinor $\xi_{AA'}\in \mathfrak{H}^*_p\mathcal{M}$ ---i.e. $\xi_{AA'}=\bar{\xi}_{AA'}$. 
Given a fixed $\xi_{AA'}$, the  \emph{Kernel} of $\sigma_\bmxi$ is the vector space of spinors $\varphi$ (at $p$) satisfying the algebraic linear equation
\begin{equation}
\sigma_\bmxi \cdot \varphi =0.
\label{HomogeneousSystem:Symbol}
\end{equation}
The Kernel of $\sigma_\bmxi$ depends on the particular choice of covector $\xi_\mu$. If the Kernel is trivial (that is, it consists solely of the zero spinor) then we write, in a slight abuse of notation, $\mbox{Ker}\, \sigma_\bmxi = 0$. Now, we define the \emph{characteristic subset} $C^*_p \subset T^*_p\mathcal{M}=\mathfrak{H}_p\mathcal{M}$
as 
\[
C^*_p =\{ \xi_{AA'}\in \mathfrak{H}_p\mathcal{M} \;|\; \xi_{AA'}\neq 0, \; \mbox{Ker}\, \sigma_\bmxi \neq 0 \}.
\]

Finally, it is recalled that an operator is said to be \emph{elliptic} if its symbol is bijective. This notion requires a precise specification of the domain and range of the operator.

\medskip
We now particularise the above discussion to the prototype equation \eqref{PrototypeSpinorialEqn}. In this case, the symbol is given by
\begin{equation}
\sigma_\bmxi \cdot \varphi =  \xi^Q{}_{A'}\varphi_{QA_2\cdots A_pB'_1\cdots B'_q}.
\label{Symbol:MasslessEqn}
\end{equation}
Setting the above expression equal to zero, contracting with $\xi_C{}^{A'}$ and using that $\xi_C{}^{A'}\xi^Q_{A'}=\tfrac{1}{2}\epsilon^Q{}_C |\bmxi|^2$ with $|\bmxi|^2 \equiv \xi_{PP'}\xi^{PP'}$, one concludes that
\[
|\bmxi|^2 \varphi_{A_1\cdots A_p B'_1\cdots B'_q}=0. 
\]
So, if $\varphi_{A_1\cdots A_p B'_1\cdots B'_q} \neq 0$ then one necessarily has that $|\bmxi|^2=0$ ---that is, the covector $\xi_\mu$ corresponding to $\xi_{AA'}$ is null. The latter implies the well-known fact that the characteristics of massless spin equations are null hypersurfaces. Accordingly, there exists a spinor $\kappa_A$ such that $\xi_{AA'}=\kappa_A \bar{\kappa}_{A'}$. Consider now a further spinor $\mu_A$ which is not proportional to $\kappa_A$. Without loss of generality, we can assume that $\kappa_A\mu^A=1$. Thus, one has a second singled out null direction described by the Hermitian spinor $\underline{\xi}_{AA'}=\mu_A\bar{\mu}_{A'}$. The covectors $\xi_\mu$
 and $\underline{\xi}_\mu$ associated to $\xi_{AA'}$ and $\underline{\xi}_{AA'}$ span the characteristic subset $C^*_p$. The spinors $\kappa_A$ and $\mu_A$ can be used as a spin dyad (base) which we denote by $\{ o^A,\, \iota^A\}$. Now, the symbol \eqref{Symbol:MasslessEqn} is degenerate. To see this, one observes that from \eqref{HomogeneousSystem:Symbol} it follows that 
 \[
o^Q\varphi_{QA_2\cdots A_p B'_1\cdots B'_q}=0,
 \]
 which, in turn, implies that 
 \begin{equation}
 \varphi_{A_1\cdots A_pB'_1\cdots B'_q } =o_{A_1} \zeta_{A_2\cdots A_pB'_1\cdots B'_q}
 \label{Kernel1}
 \end{equation}
 where $\zeta_{A_2\cdots A_pB'_1\cdots B'_q}$ is a non-zero spinor. In a similar fashion one can find the additional solution
 \begin{equation}
\varphi_{A_1\cdots A_pB'_1\cdots B_q} = \iota_{A_1}\vartheta_{A_2\cdots A_pB'_1\cdots B'_q}.
\label{Kernel2}
\end{equation}
The solutions \eqref{Kernel1} and \eqref{Kernel2} span $\mbox{Ker}\, \sigma_\bmxi$. Observe that it does not span the whole of the space of spinors of the form $(\cdot)_{A_1\cdots A_pB'_1\cdots B'_q}$ at $p$. 

\medskip
It is interesting to contrast the above specific discussion with what happens for the wave operator $\square \equiv \nabla^a \nabla_a$ acting, say, on a spinor $\varphi_{A_1 \cdots A_p B'_1\cdots B'_q}$. In this case the symbol is given by 
\[
\sigma_\bmxi \cdot \varphi = |\bmxi|^2\varphi_{A_1\cdots A_p B'_1\cdots B'_q}.
\]
As before, the characteristic subset at a point $p$ is generated by a pair of null covectors $\xi_\mu$ and $\underline{\xi}_\mu$. Now, however, the Kernel consists of all the space $\mathfrak{S}_{A_1\cdots A_pB'_1\cdots B'_q}$ at $p$.
  
\subsubsection{Symmetric hyperbolic systems}
\label{Section:SHS}

There are several notions of hyperbolicity in the literature ---see e.g. \cite{Chr08,Reu98}. Given that the prototype equation \eqref{PrototypeSpinorialEqn} is first order, in the following, we focus on symmetric hyperbolic systems.

\medskip
Consider a system of partial differential equations of the form
\begin{equation}
\mathbf{A}^\mu(x)\partial_\mu \mathbf{u} = \mathbf{B}(x) 
\label{SHS}
\end{equation}
where $\mathbf{u}$ is a $\mathbb{C}^N$-valued unknown for some positive integer $N$, $\mathbf{A}^\mu(x)$, $\mu=0,\ldots,3$ are $(N\times N)$-matrix-valued functions of the local coordinates $x=(x^\mu)$ and $\mathbf{B}(x)$ is a vector-valued function of $x$. Equation \eqref{SHS} is said to be \emph{symmetric hyperbolic} if:

\begin{enumerate}[i.]
\item there exists a timelike covector $\tau_\mu$ such that $\mathbf{A}^\mu(x)\tau_\mu$ is a positive-definite matrix;
\item the matrices $\mathbf{A}^\mu(x)$ are Hermitian ---that is, $(\mathbf{A}^\mu)^* = \mathbf{A}^\mu$.
\end{enumerate}

From the above definition, using basic results from matrix theory, it follows that $\mathbf{A}^\mu\tau_\mu$ is an invertible matrix. 

\begin{remark}
{\em The above classic definition of symmetric hyperbolic systems is not naturally suited to spinor equations like \eqref{PrototypeSpinorialEqn} as it requires writing the unknown $\varphi_{A_1\cdots A_p B'_1\cdots B'_q}$ as a column vector $\mathbf{u}$. In the following, we will consider at alternative manners of describing symmetric hyperbolic systems in the context of spinor equations and fields. }
\end{remark}

\subsection{Spinorial equations}
\label{Section:IrreducibleDecompositionSpinorialEqns}
Equation \eqref{PrototypeSpinorialEqn} is an abstract spinor equation in the sense that it makes no reference to any coordinate system and spin basis. In order to further discuss its properties, consider a normalised spin frame $\{ \bmvarepsilon_\bmA{}^A \}$. Contracting equation \eqref{PrototypeSpinorialEqn} with respect to the basis, one obtains the set of scalar equations described by
\[
\nabla^\bmQ{}_{\bmA'}\varphi_{\bmQ \bmA_2\cdots \bmA_{p} \bmB'_1\cdots \bmB'_q} +F_{\bmA'\bmA_1\cdots \bmA_{p} \bmB'_1\cdots \bmB'_q}{}^{\bmQ_1\cdots \bmQ_p \bmQ'_1\cdots \bmQ'_q} \varphi_{\bmQ_1\cdots \bmQ_p \bmQ'_1\cdots \bmQ_q}= f_{\bmA' \bmA_2\cdots \bmA_p \bmB'_1\cdots \bmB'_q}.
\]
The components $\varphi_{\bmA_1\cdots \bmA_p \bmA'_1\cdots \bmA'_q}$ are scalars depending on some local coordinates $ x =(x^\mu)$. 

\medskip
Now, an important observation is that the combination 
\begin{subequations}
\begin{eqnarray}
&& \nabla^\bmQ{}_{\bmone'}\varphi_{\bmQ \bmA_2\cdots \bmA_{p} \bmB'_1\cdots \bmB'_q} \nonumber\\
&& \hspace{2cm}+F_{\bmone'\bmA_2\cdots \bmA_{p} \bmB'_1\cdots \bmB'_q}{}^{\bmQ_1\cdots \bmQ_p \bmQ'_1\cdots \bmQ'_q} \varphi_{\bmQ_1\cdots \bmQ_p \bmQ'_1\cdots \bmQ_q}= f_{\bmone'\bmA_2\cdots \bmA_p \bmB'_1\cdots \bmB'_q}, \label{PrototypeSHS1}\\
&& -\nabla^\bmQ{}_{\bmzero'} \varphi_{\bmQ \bmA_2\cdots \bmA_{p} \bmB'_1\cdots \bmB'_q} \nonumber\\
&& \hspace{2cm} -F_{\bmzero'\bmA_2\cdots \bmA_{p} \bmB'_1\cdots \bmB'_q}{}^{\bmQ_1\cdots \bmQ_p \bmQ'_1\cdots \bmQ'_q} \varphi_{\bmQ_1\cdots \bmQ_p \bmQ'_1\cdots \bmQ_q}= -f_{\bmzero'\bmA_2\cdots \bmA_p \bmB'_1\cdots \bmB'_q}
\label{PrototypeSHS2}
\end{eqnarray}
\end{subequations}
is a \emph{symmetric hyperbolic system} ---see e.g. \cite{CFEBook}. Indeed, using that
\[
\nabla^\bmQ{}_{\bmA'}\varphi_{\bmQ \bmA_2\cdots \bmA_{p} \bmB'_1\cdots \bmB'_q} = \nabla_{\bmone\bmA'} \varphi_{\bmzero\bmA_2\cdots \bmA_{p} \bmB'_1\cdots \bmB'_q}-\nabla_{\bmzero\bmA'} \varphi_{\bmone\bmA_2\cdots \bmA_{p} \bmB'_1\cdots \bmB'_q}
\]
one has that the principal part of the system \eqref{PrototypeSHS1}-\eqref{PrototypeSHS2} can be written in matricial form as
\[
\mathbf{A}^\mu \partial_\mu \bmvarphi \equiv 
\left(
\begin{array}{cc}
 \sigma^{\mu}{}_{\bmone\bmone'}   & -\sigma^{\mu}{}_{\bmzero\bmone'}{} \\
-\sigma^{\mu}{}_{\bmone\bmzero'}{}      & \sigma^{\mu}{}_{\bmzero\bmzero'}
\end{array}
\right)\partial_\mu
\left(
\begin{array}{c}
     \varphi_{\bmzero\bmA_2\cdots \bmA_{p} \bmB'_1\cdots \bmB'_q}  \\     \varphi_{\bmone\bmA_2\cdots \bmA_{p} \bmB'_1\cdots \bmB'_q}
\end{array}
\right).
\]
The matrices $\mathbf{A}^\mu$ are Hermitian as the vectors $\bme_{\bmzero\bmzero'}\equiv \sigma^{\mu}{}_{\bmzero\bmzero'} \bmpartial_\mu$ and $\bme_{\bmone\bmone'}\equiv \sigma^{\mu}{}_{\bmone\bmone'} \bmpartial_\mu$ are real and while for $\bme_{\bmzero\bmone'}\equiv \sigma^{\mu}{}_{\bmzero\bmone'}\bmpartial_\mu$ and $\bme_{\bmone\bmzero'}\equiv \sigma^{\mu}{}_{\bmone\bmzero'}\bmpartial_\mu$, one has that $\bme_{\bmzero\bmone'} =\overline{\bme_{\bmone\bmzero'}}$. Letting 
\[
\xi_\mu \equiv \sigma_{\mu}{}^{\bmzero\bmzero'}{} + \sigma_{\mu}{}^{\bmone\bmone'}
\]
and using that $\sigma^{\mu}{}_{\bmA\bmA'} \sigma_{\mu}{}^{\bmB\bmB'} =\delta_\bmA{}^\bmB\delta_{\bmA'}{}^{\bmB'}$, it follows that 
\[
\mathbf{A}^\mu \xi_\mu =
\left(
\begin{array}{cc}
  1   & 0 \\
  0   & 1
\end{array}
\right).
\]
The latter is clearly a positive definite matrix.

\subsubsection{A wave equation for the spinor field}
An alternative approach to the one discussed above consists of deriving a wave equation for the spinor field $\varphi_{A_1\cdots A_pB'_1\cdots B_q}$. We briefly discuss this approach here for completeness.

\medskip
Applying the derivative $\nabla_{A_1}{}^{A'}$ to equation \eqref{PrototypeSpinorialEqn} and using the identity
\begin{equation}
\nabla_{A_1}{}^{A'}\nabla^Q{}_{A'} =-\frac{1}{2}\epsilon_{A_1}{}^Q +\square_{A_1}{}^Q
\label{Commutator:Reduced}
\end{equation}
where $\square_{AB}\equiv \nabla_{(A|A'|}\nabla_{B)}{}^{A'}$ is the so-called \emph{Penrose box} encoding the information of the commutator of the spinorial covariant derivative $\nabla_{AA'}$, one readily obtains that
\begin{eqnarray}
&& \square \varphi_{A_1\cdots A_pB'_1\cdots B'_q} - 2 F_{Q'A_2\cdots A_p B'_1\cdots B'_q}{}^{Q_1\cdots Q_pQ'_1\cdots Q'_q} \nabla_{A_1}{}^{Q'}\varphi_{Q_1\cdots Q_p Q'_1\cdots Q'_q} \nonumber \\
&& \hspace{2cm}- 2 \nabla_{A_1}{}^{Q'} F_{Q'A_2\cdots A_p B'_1\cdots B'_q}{}^{Q_1\cdots Q_pQ'_1\cdots Q'_q}\varphi_{Q_1\cdots Q_p Q'_1\cdots Q'_q} - 2 \square_{A_1}{}^Q\varphi_{QA_2\cdots A_p B'_1\cdots B'_q}\nonumber \\
&& \hspace{4cm}=- 2\nabla_{A_1}{}^{Q'} f_{Q'A_2\cdots A_p B'_1\cdots B'_q}.
\label{WaveEqnPhi}
\end{eqnarray}
Now, recall that the action of the Penrose box $\square_{AB}$ on valence 1 spinors $\mu_C$ and $\nu_{C'}$ is given by
\[
\square_{AB} \mu_C = \Psi_{ABCDQ}\mu^Q - 2 \Lambda \mu_{(A}\epsilon_{B)C}, \qquad \nabla_{AB} \nu_{C'} = \Phi_{ABC'Q'}\nu^{Q'},
\]
where $\Psi_{ABCD}$, $\Lambda$ and $\Phi_{ABA'B'}$ denote, respectively, the Weyl spinor, Ricci Scalar and tracefree Ricci spinor ---collectively, these spinor fields encode the curvature of the spacetime $(\mathcal{M},\bmg)$. It thus follows from the previous discussion that equation \eqref{WaveEqnPhi} constitutes a system of wave equations for the various components of the spinor $\varphi_{A_1\cdots A_p B'_1\cdots B'_q}$.  

\begin{remark}
{\em Observe that in the case that certain contractions of the spinor $\varphi_{A_1\cdots A_p B'_1\cdots B'_q}$ vanish, then equation \eqref{WaveEqnPhi} can give rise to algebraic constraints similar in nature to the \emph{Buchdahl constraint} for massless fields. \emph{In the reminder of this article we ignore these constraints and assume them to be satisfied.}}
\end{remark}

\begin{remark}
{\em While any solution of equation \eqref{PrototypeSpinorialEqn} is a solution to the wave equation \eqref{WaveEqnPhi}, the converse is not necessarily true. The conditions under which the latter may be true need to be elucidated via a \emph{propagation of constraints}-type of argument. The natural way of carrying out such an argument is to define a \emph{zero-quantity} $Z_{A'A_2\cdots A_p B'_1\cdots B_q}$ and then compute $\square Z_{A'A_2\cdots A_p B'_1\cdots B_q}$.
If it is possible, making use of the commutator of covariant derivatives and the definition of $Z_{A'A_2\cdots A_p B'_1\cdots B_q}$, to obtain an equation of the form
\begin{equation}
\square Z = H(Z,\nabla Z),
\label{BoxZ}
\end{equation}
where $H$ is an homogeneous expression of $Z_{A'A_2\cdots A_p B'_1\cdots B_q}$ and its first order derivatives, then from the uniqueness of solutions to wave equations, it follows that the criteria for propagation of equation  \eqref{PrototypeSpinorialEqn}, from say, a Cauchy initial value problem, is that the initial data satisfy the conditions
\[
Z_{A'A_2\cdots A_p B'_1\cdots B_q}=0, \qquad \nabla_{CC'}Z_{A'A_2\cdots A_p B'_1\cdots B_q}=0, \qquad \mbox{on the initial hypersurface.}
\]
Depending on the specific form of equation \eqref{PrototypeSpinorialEqn}, it may be possible to rewrite the above conditions in terms of conditions on the initial value of $\varphi_{A_1\cdots A_p B'_1\cdots B'_q}$. Finally, observe that whether it is possible to obtain an expression of the form \eqref{BoxZ} or not very much depends on the specific structural properties of equation \eqref{PrototypeSpinorialEqn} and of the background geometry $(\mathcal{M},\bmg)$.
}   
\end{remark}

\subsection{Decomposition of the equation in irreducible terms}
\label{Subsection:IrreducibleDecompositionEquation}

Rather than working with the spinor $\varphi_{\bmA_1\cdots \bmA_{p} \bmB'_1\cdots \bmB'_q}$, it will prove convenient to consider its space spinor version as given by equation \eqref{Definition:SpaceSpinorPhi}. This allows to make full use of the machinery of irreducible decomposition of spinors as given by Proposition \ref{Proposition:IrreducibleDecompositions}.   Contracting equation \eqref{PrototypeSpinorialEqn} with $ \tau_{A_{p+1}}{}^{B_1'}\cdots\tau_{A_{p+q}}{}^{B'_q}$ and using that
\begin{eqnarray*}
&&  \nabla^Q{}_{A'}\varphi_{QA_2\cdots A_p A_{p+1}\cdots A_{p+q}} = \nabla^Q{}_{A'}\big( \tau_{A_{p+1}}{}^{B_1'}\cdots\tau_{A_{p+q}}{}^{B'_q} \varphi_{A_1\cdots A_{p} B'_1\cdots B'_q} \big)    \\
&& \phantom{\nabla^Q{}_{A'}\varphi_{QA_2\cdots A_p A_{p+1}\cdots A_{p+q}}} = \tau_{A_{p+1}}{}^{B_1'}\cdots\tau_{A_{p+q}}{}^{B'_q} \nabla^Q{}_{A'} \varphi_{A_1\cdots A_{p} B'_1\cdots B'_q}\\
&& \hspace{5cm} + \sqrt{2} \varphi_{Q A_2 \cdots A_p B'_1 \cdots B'_q} \tau_{A_{p+2}}{}^{B'_2}\cdots\tau_{A_{p+q}}{}^{B'_q} \chi^Q{}_{A'}{}_{A_{p+1}}{}^{B'_1}\\
&& \hspace{6cm} \vdots \\
&& \hspace{5cm} + \sqrt{2} \varphi_{Q A_2 \cdots A_p B'_1 \cdots B'_q} \tau_{A_{p+1}}{}^{B'_1}\cdots\tau_{A_{p+q-1}}{}^{B'_{q-1}} \chi^Q{}_{A'}{}_{A_{p+q}}{}^{B'_q},
\end{eqnarray*}
where we have used \eqref{Definition-chi}, one obtains the equation 
\begin{eqnarray}
&& \nabla^Q{}_{A'} \varphi_{QA_2\cdots A_p A_{p+1} \cdots A_{p+q}} \nonumber\\
&& \hspace{1cm}+ G_{A'A_2\cdots A_p A_{p+1}\cdots A_{p+q}}{}^{Q_1\cdots Q_p Q_{p+1} \cdots Q_{p+q}} \varphi_{Q_1\cdots Q_p Q_{p+1} \cdots  Q_{p+q}} = f_{A'A_2 \cdots A_p A_{p+1} \cdots A_{p+q}}, \quad
\label{PrototypeSpinorEquationHalfSpaceSpinor}
\end{eqnarray}
where 
\begin{eqnarray*}
 && G_{A'A_2\cdots A_p A_{p+1} \cdots A_{p+q}}{}^{Q_1\cdots Q_p Q_{p+1}\cdots  Q_{p+q}} \\
 && \hspace{2cm}= (-1)^{q} F_{A' A_2\cdots A_{p} B'_1\cdots B'_q}{}^{Q_1\cdots Q_p Q'_1\cdots Q'_q} \tau_{A_{p}}{}^{B'_1}\cdots \tau_{A_{p+q}}{}^{B'_q}\tau^{Q_{p+1}}{}_{Q'_{1}}\cdots \tau^{Q_{p+q}}{}_{Q'_q} \\
 && \hspace{3cm}- \sqrt{2} (-1)^{q} \chi^{Q_1}{}_{A'A_{p+1}}{}^{B_1'}\tau^{Q_{p+1}}{}_{B'_1} \epsilon_{A_2}{}^{Q_2}\cdots \epsilon_{A_p}{}^{Q_p} \epsilon_{A_{p+2}}{}^{Q_{p+2}}\cdots \epsilon_{A_{p+q}}{}^{Q_{p+q}}\\
 && \hspace{4cm} \vdots \\
 && \hspace{3cm} - \sqrt{2} (-1)^{q} \chi^{Q_1}{}_{A'A_{p+q}}{}^{B_q'}\tau^{Q_{p+q}}{}_{B'_q} \epsilon_{A_2}{}^{Q_2}\cdots \epsilon_{A_p}{}^{Q_p} \epsilon_{A_{p+1}}{}^{Q_{p+1}}\cdots \epsilon_{A_{p+q-1}}{}^{Q_{p+q-1}}.
\end{eqnarray*}

\begin{remark}
    {\em In brief, the spinor 
\[    
    G_{A'A_2\cdots A_p A_{p+1} A_{p+q}}{}^{Q_1\cdots Q_p Q_{p+1}\cdots  Q_{p+q}}
 \]   
    is the space spinor version of
\[    
F_{A' A_2\cdots A_{p} B'_1\cdots B'_q}{}^{Q_1\cdots Q_p Q'_1\cdots Q'_q}
\]
corrected by terms involving the derivative of the Hermitian spinor $\tau^{AA'}$ ---i.e. the Weingarten spinor.}
\end{remark}

In order to extract the content of equation \eqref{PrototypeSpinorEquationHalfSpaceSpinor}, we make use of the irreducible decomposition of spinors. To this end define
\[
\Omega_{A'A_2\cdots A_{p+q}} \equiv \nabla^Q{}_{A'}\varphi_{QA_2\cdots A_{p+q}}.
\] 
It then follows that the above spinor admits the decomposition
\[
\Omega_{A'A_2\cdots A_{p+q}}=\sum_{0\leq 2\ell\leq p+q-1 } \sum_{(j_1,\ldots,j_{2\ell})\in E^{2\ell}_{p+q-1}}  \epsilon_{A_{j_1}A_{j_2}}\cdots \epsilon_{A_{j_{2\ell-1}}A_{j_{2\ell}}}\omega^{[j_1,\cdots,j_{2\ell}]}_{A'A_{i_1}\cdots A_{i_{p+q-1-2\ell}}},
\]
with the understanding that
\[
\{ A_{j_1},\ldots, A_{j_{2\ell}} \} \cup \{ A_{i_1},\, \ldots, A_{i_{p+q-1-2\ell}}\} = \{A_2,\ldots,A_{p+q} \}.
\]
The spinors $\omega^{[j_1,\cdots,j_{2\ell}]}_{A'A_{i_1}\cdots A_{i_{p+q-1-2\ell}}}$ are linear combinations of the irreducible components of $\Omega_{A'A_2\cdots A_{p+q}}$ ---namely, one has that 
\[
\omega^{[j_1,\cdots,j_{2\ell}]}_{A'A_{i_1}\cdots A_{i_{p+q-1-2\ell}}} \equiv \sum_{(k_1,\ldots,k_{2\ell})\in E^{2\ell}_{p+q-1}}\mathfrak{w}^{j_1\cdots j_{2\ell}}_{k_1\cdots k_{2\ell}}\Omega_{A'A_{i_1}\cdots A_{i_{p+q-1-2\ell}}}^{(k_1\cdots k_{2\ell})}.
\]
The valence $p+q-1-2\ell$ symmetric spinors $\Omega^{(j_1,\cdots,j_{2\ell})}_{A'A_{i_1}\cdots A_{i_{p+q-1-2\ell}}}$ for $0\leq 2\ell\leq p+q-1$ and $(j_1,\ldots,j_{2\ell})\in E^{2\ell}_{p+q-1}$ encode the essential content of the principal part of equation \eqref{PrototypeSpinorEquationHalfSpaceSpinor}. This decomposition implies, in turn, a coupled system of equations for the irreducible components $\varphi^{(j_1,\cdots,j_{2\ell})}_{A_{i_1}\cdots A_{i_{p+q -2\ell}}}$ of $\varphi_{A_1\cdots A_{p+q}}$. 

\medskip
By definition one has that 
\begin{eqnarray*}
&& \Omega^{(j_1,j_2,\ldots,j_{2\ell})}_{A'A_{i_1}\cdots A_{i_{p+q -1-2\ell}}} \equiv \Omega_{A'(A_{i_1} \cdots |A_{j_1}|\cdots }{}^{A_{j_1}}{}_{\cdots |A_{j_3}|}{}_{\cdots}{}^{A_{j_3}}{}_{\cdots \cdots |A_{j_{2\ell-1}}|\cdots}{}^{A_{j_{2\ell-1}}}{}_{ \cdots A_{i_{p+q -1-2\ell}})}\\
&& \phantom{\Omega^{(j_1,j_2,\ldots,j_{2\ell})}_{A'A_{i_1}\cdots A_{i_{p+q -1-2\ell}}}} = \nabla_{A'}{}^Q\varphi_{Q(A_{i_1}\cdots |A_{j_1}|\cdots}{}^{A_{j_1}}{}_{\cdots |A_{j_{2\ell-1}}|\cdots}{}^{A_{j_{2\ell-1}}}{}_{\cdots A_{i_{p+q -1-2\ell}})}.
\end{eqnarray*}
Making use of the classic argument behind Proposition \ref{Proposition:IrreducibleDecompositions} one readily finds that 
\begin{eqnarray*}
&& \varphi_{Q(A_{i_1}\cdots |A_{j_1}|\cdots}{}^{A_{j_1}}{}_{\cdots |A_{j_{2\ell-1}}|\cdots}{}^{A_{j_{2\ell-1}}}{}_{\cdots A_{i_{p+q -1-2\ell}})} \\
&& \hspace{1cm} = \varphi_{(QA_{i_1}\cdots |A_{j_1}|\cdots}{}^{A_{j_1}}{}_{\cdots |A_{j_{2\ell-1}}|\cdots}{}^{A_{j_{2\ell-1}}}{}_{\cdots A_{i_{p+q -1-2\ell}})}\\
&& \hspace{1cm}+ \frac{1}{p+q -2\ell}\sum_{k=1}^{p+q-1 -2\ell} \epsilon_{A_{i_k}Q}
\varphi^P{}_{(PA_{i_2}\cdots A_{i_{k-1}} A_{i_{k+1}} \cdots  |A_{j_1}|\cdots}{}^{A_{j_1}}{}_{\cdots |A_{j_{2\ell-1}}|\cdots}{}^{A_{j_{2\ell-1}}}{}_{\cdots A_{i_{p+q -1-2\ell}})}.
\end{eqnarray*}
The terms in the sum require further manipulations. Again, using the classic argument leading to Proposition \ref{Proposition:IrreducibleDecompositions}, one has that
\begin{eqnarray}
&& \varphi^P{}_{(PA_{i_2} \cdots A_{i_{k-1}} A_{i_{k+1}} \cdots   |A_{j_1}|\cdots}{}^{A_{j_1}}{}_{\cdots |A_{j_{2\ell-1}}|\cdots}{}^{A_{j_{2\ell-1}}}{}_{\cdots A_{i_{p+q -1-2\ell}})} \nonumber \\
&& \hspace{1cm}=\frac{1}{p+q -1 -2\ell} \left( \varphi^{P}{}_{P( A_{i_2} \cdots A_{i_{k-1}} A_{i_{k+1}}\cdots   |A_{j_1}|\cdots}{}^{A_{j_1}}{}_{\cdots |A_{j_{2\ell-1}}|\cdots}{}^{A_{j_{2\ell-1}}}{}_{\cdots A_{i_{p+q -1-2\ell}})}  \right. \nonumber \\
&& \hspace{1cm} + \varphi^{P}{}_{( A_{i_2} |P| \cdots A_{i_{k-1}} A_{i_{k+1}}\cdots   |A_{j_1}|\cdots}{}^{A_{j_1}}{}_{\cdots |A_{j_{2\ell-1}}|\cdots}{}^{A_{j_{2\ell-1}}}{}_{\cdots A_{i_{p+q -1-2\ell}})} \nonumber \\
&& \hspace{3cm} \vdots \nonumber \\
&& \hspace{1cm} + \left. \varphi^{P}{}_{( A_{i_2} \cdots A_{i_{k-1}} A_{i_{k+1}}\cdots   |A_{j_1}|\cdots}{}^{A_{j_1}}{}_{\cdots |A_{j_{2\ell-1}}|\cdots}{}^{A_{j_{2\ell-1}}}{}_{\cdots A_{i_{p+q -1-2\ell}})P} \right). \label{Expansion-phi-pp}
\end{eqnarray}
Using the above, we can write 
\begin{eqnarray*}
&& \Omega^{(j_1,j_2,\ldots,j_{2\ell})}_{A'A_{i_1}\cdots A_{i_{p+q -1-2\ell}}} = \nabla^{Q}{}_{A'}\varphi_{QA_{i_1}\cdots A_{i_{p+q -1-2\ell}}}^{(j_1+1,j_2+1,\cdots j_{2\ell}, j_{2\ell}+1)}\\
&& \hspace{1cm}- \frac{1}{p+q -2\ell}\sum_{k=1}^{p+q-1 -2\ell} 
\nabla_{A_{i_k}A'}\varphi^P{}_{(PA_{i_2}\cdots A_{i_{k-1}} A_{i_{k+1}} \cdots  |A_{j_1}|\cdots}{}^{A_{j_1}}{}_{\cdots |A_{j_{2\ell-1}}|\cdots}{}^{A_{j_{2\ell-1}}}{}_{\cdots A_{i_{p+q -1-2\ell}})}.
\end{eqnarray*}
where we have used the fact that 
\begin{equation*}
    \nabla^{Q}{}_{A'}\varphi_{QA_{i_1}\cdots A_{i_{p+q -1-2\ell}}}^{(j_1+1,j_2+1,\cdots j_{2\ell}, j_{2\ell}+1)}=\nabla^{Q}{}_{A'}\varphi_{(QA_{i_1}\cdots |A_{j_1}|\cdots}{}^{A_{j_1}}{}_{\cdots |A_{j_{2\ell-1}}|\cdots}{}^{A_{j_{2\ell-1}}}{}_{\cdots A_{i_{p+q -1-2\ell}})}.
\end{equation*}
Then, using equation \eqref{Expansion-phi-pp}, we can write 
\begin{eqnarray*}
    && \nabla_{A_{i_k}A'} \varphi^P{}_{(PA_{i_2}\cdots A_{i_{k-1}} A_{i_{k+1}} \cdots  |A_{j_1}|\cdots}{}^{A_{j_1}}{}_{\cdots |A_{j_{2\ell-1}}|\cdots}{}^{A_{j_{2\ell-1}}}{}_{\cdots A_{i_{p+q -1-2\ell}})} \\
    && \hspace{1cm}= -\frac{1}{p+q -1 -2\ell} \left(\nabla_{A_{i_k}A'} \varphi_{P}{}^{P}{}_{( A_{i_2} \cdots A_{i_{k-1}} A_{i_{k+1}}\cdots   |A_{j_1}|\cdots}{}^{A_{j_1}}{}_{\cdots |A_{j_{2\ell-1}}|\cdots}{}^{A_{j_{2\ell-1}}}{}_{\cdots A_{i_{p+q -1-2\ell}})}  \right. \nonumber \\
    && \hspace{1cm} + \nabla_{A_{i_k}A'} \varphi_{P( A_{i_2}}{}^{P}{}_{\cdots A_{i_{k-1}} A_{i_{k+1}}\cdots   |A_{j_1}|\cdots}{}^{A_{j_1}}{}_{\cdots |A_{j_{2\ell-1}}|\cdots}{}^{A_{j_{2\ell-1}}}{}_{\cdots A_{i_{p+q -1-2\ell}})} \nonumber \\
    && \hspace{3cm} \vdots \nonumber \\
    && \hspace{1cm} + \left. \nabla_{A_{i_k}A'} \varphi_{P( A_{i_2} \cdots A_{i_{k-1}} A_{i_{k+1}}\cdots   |A_{j_1}|\cdots}{}^{A_{j_1}}{}_{\cdots |A_{j_{2\ell-1}}|\cdots}{}^{A_{j_{2\ell-1}}}{}_{\cdots A_{i_{p+q -1-2\ell}})}{}^{P} \right),
\end{eqnarray*}
which by definition can be written as 
\begin{eqnarray*}
    && \nabla_{A_{i_k}A'} \varphi^P{}_{(PA_{i_2}\cdots A_{i_{k-1}} A_{i_{k+1}} \cdots  |A_{j_1}|\cdots}{}^{A_{j_1}}{}_{\cdots |A_{j_{2\ell-1}}|\cdots}{}^{A_{j_{2\ell-1}}}{}_{\cdots A_{i_{p+q -1-2\ell}})} \\
    && \hspace{1cm}= -\frac{1}{p+q -1 -2\ell} \sum_{m\in\{ 2,\ldots,p+q -1-2\ell\}} \nabla_{A'(A_{i_k}} \varphi_{A_{i_{2}}\cdots A_{i_{p+q -1-2\ell}})}^{(1,m,j_{1}+1,j_{2}+1,\ldots,j_{2\ell}+1)}.
\end{eqnarray*}
Making use of the above expression, one then concludes that
\begin{eqnarray*}
&& \Omega^{(j_1,j_2,\ldots,j_{2\ell})}_{A'A_{i_1}\cdots A_{i_{p+q -1-2\ell}}} = \nabla^{Q}{}_{A'}\varphi_{QA_{i_1}\cdots A_{i_{p+q -1-2\ell}}}^{(j_1+1,j_2+1,\cdots j_{2\ell}, j_{2\ell}+1)}\\
&& \hspace{1cm}+ \mathfrak{C}_{p,q,\ell}\sum_{k=1}^{p+q-1 -2\ell} 
\sum_{m\in\{ 2,\ldots,p+q -1-2\ell\}} \nabla_{A'(A_{i_k}} \varphi_{A_{i_{2}}\cdots A_{i_{p+q -1-2\ell}})}^{(1,m,j_{1}+1,j_{2}+1,\ldots,j_{2\ell}+1)},
\end{eqnarray*}
where
\[
\mathfrak{C}_{p,q,\ell}\equiv \frac{1}{(p+q -1 -2\ell)(p+q -2\ell)}.
\]
In turn, the coefficients $\omega^{[j_1,\cdots,j_{2\ell}]}_{A'A_{i_1}\cdots A_{i_{p+q-1-2\ell}}}$ are linear combinations of the above expressions. 

\begin{remark}
{\em Observe that  for given admissible $\ell$ the spinor $\varphi_{QA_{i_1}\cdots A_{i_{p+q -1-2\ell}}}^{(j_1+1,j_2+1,\cdots j_{2\ell}, j_{2\ell +1})}$ is of valence $p+q-2\ell$ while $\varphi_{A_{i_1}\cdots A_{i_{p+q-2\ell}}}^{(1,m,j_1+1,j_2+1,\ldots,j_{{2\ell}-1},j_{2\ell})}$ is of valence $p+q-2-2\ell$, consistent with the fact that $\Omega^{(j_1,j_2,\ldots,j_{2\ell})}_{A'A_{i_1}\cdots A_{i_{p+q-1-2\ell}}}$ is of valence $p+q-2\ell$.}
\end{remark}

Taking linear combinations of the coefficients 
$\Omega^{(j_1,j_2,\ldots,j_{2\ell})}_{A'A_{i_1}\cdots A_{i_{p+q -1-2\ell}}}$ for the allowed range of the indices one finds an expression of the form
\begin{eqnarray*}
    && \omega^{[j_1\cdots j_{2\ell}]}_{A'A_{i_1}\cdots A_{i_{p+q-1-2\ell}}} \sim \nabla^Q{}_{A'} \phi_{QA_{i_1}\cdots A_{i_{p+q-1-2\ell}}}^{(j_{1}\cdots j_{2\ell})}+ \mathfrak{c}_{p,q,\ell} \nabla_{A'(A_{i_1}}\psi^{(j_1,\cdots,j_{2\ell})}_{A_{i_2}\cdots A_{i_{p+q-1-2\ell}})},
\end{eqnarray*}
where 
\[
\phi_{A_{i_1}\cdots A_{i_{p+q-2\ell}}}, \qquad \psi_{A_{i_2}\cdots A_{i_{p+q-1-2\ell}}} 
\]
are, respectively, linear combinations of the irreducible components 
\[
\varphi_{A_{i_1}\cdots A_{i_{p+q-2\ell}}}^{(j_1+1,j_2+1,\cdots j_{2\ell}, j_{2\ell}+1)}, \qquad \varphi_{A_{i_{2}}\cdots A_{i_{p+q -1-2\ell}}}^{(1,m,j_{1}+1,j_{2}+1,\ldots,j_{2\ell}+1)}
\]
for a fixed choice of $(j_1\cdots, j_{p+q-2\ell})\in E^{2\ell}_{p+q-1}$ with $2\ell \leq p+q-1$ and where $\mathfrak{c}_{p,q,\ell}$ is a constant.

\medskip
One can also apply the decomposition in terms of irreducible components to the lower order terms and the source term of equation \eqref{PrototypeSpinorEquationHalfSpaceSpinor}. On this spirit, one can write 
\begin{eqnarray*}
&& G_{A'A_{1}\cdots A_{p+q-1}}{}^{Q_1\cdots Q_{p+q}}\varphi_{Q_1\cdots Q_{p+q}} = \sum_{2\ell \leq p+q-1} \sum_{(j_1,\ldots, j_{2\ell})\in E^{2\ell}_{p+q -1}} \epsilon_{A_{j_1}A_{j_2}}\cdots \epsilon_{A_{j_{2\ell-1}}A_{j_{2\ell}}} \\
&& \hspace{9cm} \times \mathrm{G}_{A'A_{i_1}\cdots A_{i_{p+q -1-2\ell}}}^{[j_1,\ldots j_{2\ell}]} {}^{Q_1\cdots Q_{p+q}} \varphi_{Q_1\cdots Q_{p+q}}, \\
&& f_{A'A_{1}\cdots A_{p+q-1}} = \sum_{2\ell \leq p+q-1} \sum_{(j_1,\ldots, j_{2\ell})\in E^{2\ell}_{p+q-1}} \epsilon_{A_{j_1}A_{j_2}}\cdots \epsilon_{A_{j_{2\ell-1}}A_{j_{2\ell}}} \mathrm{f}_{A'A_{i_1}\cdots A_{i_{p+q-1-2\ell}}}^{[j_1,\ldots,j_{2\ell}]},
\end{eqnarray*}
where $\mathrm{G}_{A'A_{i_1}\cdots A_{i_{p+q -1-2\ell}}}^{[j_1,\ldots j_{2\ell}]} {}^{Q_1\cdots Q_{p+q}}$ and $\mathrm{f}_{A'A_{i_1}\cdots A_{i_{p+q-1-2\ell}}}^{[j_1,\ldots,j_{2\ell}]}$ are linear combinations of the irreducible components of $G_{A'A_2\cdots A_p A_{p+1} A_{p+q}}{}^{Q_1\cdots Q_p Q_{p+1}\cdots  Q_{p+q}}$ and $f_{A'A_2\cdots A_p A_{p+1} A_{p+q}}$ with respect to the indices $\{ A_{1}, \cdots A_{p+q-1} \}$. So, they are of the form
\begin{eqnarray*}
    && \mathrm{G}_{A'A_{i_1}\cdots A_{i_{p+q -1-2\ell}}}^{[j_1,\ldots j_{2\ell}]} {}^{Q_1\cdots Q_{p+q}} \sim G_{A'A_{i_1}\cdots A_{i_{p+q -1-2\ell}}}^{(j_1,\ldots j_{2\ell})}{}^{Q_1\cdots Q_{p+q}}, \\ 
    && \mathrm{f}_{A'A_{i_1}\cdots A_{i_{p+q -1-2\ell}}}^{[j_1,\ldots j_{2\ell}]} \sim f_{A'A_{i_1}\cdots A_{i_{p+q -1-2\ell}}}^{(j_1,\ldots j_{2\ell})}.
\end{eqnarray*}

\medskip
Combining all the above expressions one obtains a \emph{hierarchy of equations} for the (independent) irreducible components of the (space spinorialised) spinor field $\varphi_{A_1\cdots A_{p+q}}$ of the form  
\begin{eqnarray*}
&& \nabla^{Q}{}_{A'}\phi_{QA_{i_1}\cdots A_{i_{p+q -1-2\ell}}}^{(j_1,j_2,\cdots j_{2\ell}, j_{2\ell})}+ \mathfrak{c}_{p.q.\ell}\nabla_{A'(A_{i_1}}\psi_{A_{i2}\cdots A_{i_{p+q-1-2\ell}}} \\
&& \hspace{2cm} + G_{A'A_{i_1}\cdots A_{i_{p+q -1-2\ell}}}^{(j_1,\ldots j_{2\ell})} {}^{Q_1\cdots Q_{p+q}} \varphi_{Q_1\cdots Q_{p+q}} = f_{A'A_{i_1}\cdots A_{i_{p+q-1-2\ell}}}^{(j_1,\ldots,j_{2\ell})},
\end{eqnarray*}
for 
\[
0\leq 2\ell \leq p+q-1, \qquad (j_1,j_2,\ldots, j_{2\ell-1},j_{2\ell})\in E^{2\ell}_{p+q-1},
\]
with the understanding that the field $\varphi_{Q_1\cdots Q_{p+q}}$ is replaced by the expansion \eqref{GeneralDecompositionPhi} in terms of irreducible components. 

\subsection{Space spinor decomposition of the equation}
Following the discussion of the preceding subsection, we restrict our attention, without loss of generality, to equations of the form
\begin{equation}
\nabla^Q{}_{A'}\phi_{QA_2\cdots A_m} + \mathfrak{c} \nabla_{(A_2|A'|}\psi_{A_3\cdots A_m)} + G_{A'A_2\cdots A_m}{}^{Q_1\cdots Q_{m}}\varphi_{Q_1\cdots Q_{m}} = f_{A'A_2\cdots A_m},
\label{SimplifiedPrototypeEqn}
\end{equation}
where $\mathfrak{c}$ a constant and $\phi_{A_1\cdots A_m}$ and $\psi_{A_1\cdots A_{m-2}}$ are symmetric spinors of valence $m$ and $m-2$, respectively. In what follows, equation \eqref{SimplifiedPrototypeEqn} will be known as the $\phi$-$\psi$ system. Now, contracting with $\tau_{A_1}{}^{A'}$ to remove the remaining primed indices from the equation and taking into account the definition of the operator $\nabla_{AB}$ as given by equation \eqref{Definition:SpaceSpinorCD} one obtains
\[
\nabla^Q{}_{A_1}\phi_{QA_2\cdots A_m} + \mathfrak{c} \nabla_{(A_2|A_1|}\psi_{A_3\cdots A_m)} + G_{A_1A_2\cdots A_m}{}^{Q_1\cdots Q_{m}}\varphi_{Q_1\cdots Q_{m}} = f_{A_1A_2\cdots A_m}.
\]
Further insights into the structure of this equation can be obtained making use of decomposition \eqref{Decomposition:SpaceSpinorCD} of the operator $\nabla_{AB}$ to obtain
\begin{eqnarray}
&& \mathcal{D}\phi_{A_1\cdots A_m} -2 \mathcal{D}^Q{}_{A_1}\phi_{QA_2\cdots A_m}\nonumber \\
&& \hspace{2cm} -\mathfrak{c} \epsilon_{(A_2| A_1|}\mathcal{D} \psi_{A_3\cdots A_m)} -2 \mathfrak{c}\mathcal{D}_{(A_2|A_1|}\psi_{A_3\cdots A_m)}\nonumber\\
&& \hspace{3cm}-2 G_{A_1A_2\cdots A_m}{}^{Q_1\cdots Q_{p+q}}\varphi_{Q_1\cdots Q_{p+q}} = -2 f_{A_1A_2\cdots A_m}.\label{SimplifiedPrototypeEqnSpaceSpinor}
\end{eqnarray}
Independent evolution equations for the components $\phi_{A_1\cdots A_m}$ and $\psi_{A_1\cdots A_{m-2}}$ are obtained from symmetrising and taking traces on the above equation. 

\medskip
From the totally symmetric part of Equation \eqref{SimplifiedPrototypeEqnSpaceSpinor} one obtains
\begin{eqnarray}
&& \mathcal{D}\phi_{A_1\cdots A_m} -2 \mathcal{D}^Q{}_{(A_1}\phi_{A_2\cdots A_m)Q}  -2 \mathfrak{c}\mathcal{D}_{(A_1A_2}\psi_{A_3\cdots A_m)}\nonumber  \\
&& \hspace{5cm} -2 G_{(A_1A_2\cdots A_m)}{}^{Q_1\cdots Q_{m}}\varphi_{Q_1\cdots Q_{m}} = -2 f_{(A_1A_2\cdots A_m)}. \label{SimplifiedPrototypeEqnSpaceSpinor:ComponentPhi}
\end{eqnarray}
The trace of Equation \eqref{SimplifiedPrototypeEqnSpaceSpinor} over $A_1$ and $A_2$ yields
\begin{eqnarray*}
   && \mathfrak{c} \epsilon_{(A_3}{}^P\mathcal{D}\psi_{P\cdots A_m)}+2\mathfrak{c} \mathcal{D}_{(P}{}^P\psi_{A_3\cdots A_m)} +2\mathcal{D}^{PQ}\phi_{PQA_3\cdots A_m}\\
   && \hspace{4cm}+ 2G^P{}_{PA_3\cdots A_m}{}^{Q_1\cdots Q_{m}}\varphi_{Q_1\cdots Q_{m}}=2f^P{}_{PA_3\cdots A_n}.
\end{eqnarray*}
The latter expression can be simplified using the identities
\begin{eqnarray*}
&& \epsilon_{(A_3}{}^P\mathcal{D}\psi_{P\cdots A_m)} = \frac{m}{m-1}\mathcal{D}\phi_{A_3\cdots A_m}, \\
&& \mathcal{D}_{(P}{}^P\psi_{A_3\cdots A_m)} =\frac{m-2}{m-1}\mathcal{D}^Q{}_{(A_3}\psi_{A_4\cdots A_m)Q},
\end{eqnarray*}
so as to obtain the evolution equation
\begin{eqnarray}
&& \mathcal{D} \psi_{A_3\cdots A_m } +\frac{(m-2)}{m}\mathcal{D}^P{}_{(A_3} \psi_{A_4\cdots A_m)P} +\frac{2(m-1)}{\mathfrak{c}m} \mathcal{D}^{PQ}\phi_{PQA_3\cdots A_m} \nonumber \\
&& \hspace{3cm}+ \frac{2(m-1)}{\mathfrak{c}m}G^P{}_{PA_3\cdots A_m}{}^{Q_1\cdots Q_{m}}\varphi_{Q_1\cdots Q_{m}}=\frac{2(m-1)}{\mathfrak{c}m}f^P{}_{PA_3\cdots A_n}. \label{SimplifiedPrototypeEqnSpaceSpinor:ComponentPsi}
\end{eqnarray}

\begin{remark}
{\em Equations \eqref{SimplifiedPrototypeEqnSpaceSpinor:ComponentPhi} and \eqref{SimplifiedPrototypeEqnSpaceSpinor:ComponentPsi} constitute an evolution system for the irreducible components $\phi_{A_1\cdots A_m}$ and $\psi_{A_1\cdots A_{m-2}}$. Observe that there is also a potential  coupling with other irreducible components via the lower order term of the equations.}
\end{remark}

\begin{remark}
    {\em If $\psi_{A_1\cdots A_{m-2}}=0$, then equation \eqref{SimplifiedPrototypeEqnSpaceSpinor:ComponentPsi}becomes a constraint for the spinor $\phi_{A_1\cdots A_m}$. This can be observed, for example, in the massless spin equations.} 
\end{remark}

\begin{remark}
{\em In the construction of estimates it will be more convenient to express the evolution equation in terms of the derivative $\nabla_{AB}$ rather than in terms of the operators $\mathcal{D}$ and $\mathcal{D}_{AB}$. It can be readily verified that equations \eqref{SimplifiedPrototypeEqnSpaceSpinor:ComponentPhi} and \eqref{SimplifiedPrototypeEqnSpaceSpinor:ComponentPsi} are equivalent to the pair
\begin{subequations}
\begin{eqnarray}
&& \nabla^Q{}_{(A_1}\phi_{A_2\cdots A_m)Q} + \mathfrak{c} \nabla_{(A_1A_2}\psi_{A_3\cdots A_m)} + G_{(A_1A_2\cdots A_m)}{}^{Q_1\cdots Q_{m}}\varphi_{Q_1\cdots Q_{m}} = f_{(A_1A_2\cdots A_m)},  \label{SpaceSpinorEvolutionSystem1}\qquad \\
&& \nabla^{PQ}\phi_{PQA_3\cdots A_m} + \mathfrak{c} \nabla_{(P}{}^P\psi_{A_3\cdots A_m)} + G^P{}_{PA_3\cdots A_m}{}^{Q_1\cdots Q_{m}}\varphi_{Q_1\cdots Q_{m}} = f^P{}_{PA_3\cdots A_m}.\label{SpaceSpinorEvolutionSystem2}
\end{eqnarray}
\end{subequations}
It is also noticed that
\[
\nabla_{(P}{}^P \psi_{A_3\cdots A_m)} =\frac{1}{m-1}\nabla_P{}^P\psi_{A_3\cdots A_m}+\frac{(m-2)}{(m-1)}\nabla_{(A_3}{}^P \psi_{A_4\cdots A_m)P}.
\]}
\end{remark}

\subsubsection{Hyperbolicity of the space spinor form of the equations}
It is of both conceptual and practical interest to understand how the notion of symmetric hyperbolicity discussed in Section \ref{Section:SHS} can be translated into the language of space spinors. 

\medskip
For further reference, it is recalled that the principal part of the evolution equations \eqref{SimplifiedPrototypeEqnSpaceSpinor:ComponentPhi} and \eqref{SimplifiedPrototypeEqnSpaceSpinor:ComponentPsi} is given by the expressions
\begin{subequations}
\begin{eqnarray}
&& \mathcal{D}\phi_{A_1\cdots A_m} -2 \mathcal{D}^Q{}_{(A_1}\phi_{A_2\cdots A_m)Q}  -2 \mathfrak{c}\mathcal{D}_{(A_1A_2}\psi_{A_3\cdots A_m)}, \label{PrincipalPart1}\\
&& \mathcal{D} \psi_{A_3\cdots A_m } +\frac{(m-2)}{m}\mathcal{D}^P{}_{(A_3} \psi_{A_4\cdots A_m)P} +\frac{2(m-1)}{\mathfrak{c}m} \mathcal{D}^{PQ}\phi_{PQA_3\cdots A_m}.
\label{PrincipalPart2}
\end{eqnarray}
\end{subequations}
In order to write the symbol associated to the above principal part it is observed that from the identity \eqref{Decomposition:SpaceSpinorCD}, it follows that
\[
\nabla_{AA'} =\frac{1}{2}\tau_{AA'}-\tau^Q{}_{A'}\mathcal{D}_{AQ}.
\]
In a similar vein, one can write, for a Hermitian spinor $\xi_{AA'}$, the decomposition
\[
\xi_{AA'}=\frac{1}{2}\xi\, \tau_{AA'}-\tau^Q{}_{A'}\xi_{AQ}
\]
where $\xi_{AB}=\xi_{(AB)}$. Moreover, one has the reality conditions
\begin{equation}
\xi=\bar{\xi}, \qquad \hat{\xi}_{AB} =-\xi_{AB}.
\label{RealityConditions:Xi}
\end{equation}
The above conditions ensure that $\xi_{AB}$
 is the spinor counterpart of a 3-dimensional covector. It then follows that the symbol of equations \eqref{SimplifiedPrototypeEqnSpaceSpinor:ComponentPhi} and \eqref{SimplifiedPrototypeEqnSpaceSpinor:ComponentPsi} can be written as
\[
 \sigma_\bmxi \cdot 
 \left(
\begin{array}{c}
\phi_{A_1\cdots A_m} \\
\psi_{A_3\cdots A_m}
\end{array}
 \right)
 =
 \left(
 \begin{array}{l}
 \xi \phi_{A_1\cdots A_m} - 2 \xi^Q{}_{(A_1}\phi_{A_2\cdots A_m)Q}-2 \mathfrak{c} \xi_{(A_1A_2}\psi_{A_3\cdots A_m)}\\
 \xi \psi_{A_3\cdots A_m} + \displaystyle \frac{(m-2)}{m}\xi^P{}_{(A_3}\psi_{A_4\cdots A_m)P}+\frac{2(m-1)}{\mathfrak{c}m}\xi^{PQ}\phi_{PQA_3\cdots A_m}
 \end{array}
 \right).
 \]
 The later can be recast as 
 \begin{equation}
 \sigma_\bmxi \cdot 
 \left(
\begin{array}{c}
\phi_{A_1\cdots A_m} \\
\psi_{A_3\cdots A_m}
\end{array}
 \right)
 \equiv
 \left(
\begin{array}{c}
S_{A_1\cdots A_m}{}^{P_1\cdots P_m}\phi_{P_1\cdots P_m} -2\mathfrak{c} P_{A_1\cdots A_m}{}^{Q_3\cdots Q_m}\psi_{Q_3\cdots Q_m}\\
\displaystyle \frac{2(m-1)}{\mathfrak{c}m} Q_{A_3\cdots A_m}{}{}^{P_1\cdots P_m}\phi_{P_1\cdots P_m} +R_{A_3\cdots A_m}{}^{Q_3\cdots Q_m}\psi_{Q_3\cdots Q_m}
\end{array}
 \right)
 \label{Symbol:PhiPsi1}
 \end{equation}
 where 
 \begin{eqnarray*}
&& S_{A_1\cdots A_m}{}^{P_1\cdots P_m} \equiv \xi \delta_{(A_1}{}^{(P_1}\cdots \delta_{A_m)}{}^{P_m)} -2 \xi^{(P_1}{}_{(A_1}\delta_{A_2}{}^{P_2}\cdots \delta_{A_m)}{}^{P_m)},\\
&& P_{A_1\cdots A_m}{}^{Q_3\cdots Q_m} \equiv \xi_{(A_1A_2}\delta_{A_3}{}^{(Q_3}\cdots \delta_{A_m}{}^{Q_m)},\\
&&  Q_{A_3\cdots A_m}{}{}^{P_1\cdots P_m} \equiv \xi^{P_1P_2}\delta_{(A_3}{}^{(P_3}\cdots \delta_{A_m)}{}^{P_m)}, \\
&& R_{A_3\cdots A_m}{}^{Q_3\cdots Q_m} \equiv \xi \delta_{(A_3}{}^{(Q_3}\cdots \delta_{A_m)}{}^{Q_m)} + \frac{(m-2)}{m}\xi^{Q_3}{}_{(A_3}\delta_{A_4}{}^{Q_4}\cdots \delta_{A_m)}{}^{Q_m)}.
 \end{eqnarray*}

 In order to make connection with the standard definition of symmetric hyperbolic systems as discussed in Section \ref{Section:SHS} we consider a bases
 \[
 \{ \sigma^i_{A_1\cdots A_{m}}\}_{i=0}^{m}, \qquad  \{ \sigma^k_{A_1\cdots A_{m-2}}\}_{k=0}^{m-2}
 \]
 for valence $m$ and $m-2$ symmetric spinors, respectively. Given a normalised dyad $\{ o^A, \; \iota^A \}$ the spinors $\sigma^i_{A_1\cdots A_{m}}$ consist of symmetrised combinations of $o_A$ and $\iota_A$ with the index $i$ indicating the number of $\iota$'s in the spinor. The normalisation of the basis is chosen so that
 \[
\sigma_j{}^{A_1\cdots A_{m}}\sigma_{A_1\cdots A_{m}}^j =\delta_i{}^j.
\]
In terms of these bases one can expand the spinors $\phi_{A_1\cdots A_{m}}$ and $\psi_{A_3\cdots A_{m}}$ as 
\[
\phi_{A_1\cdots A_{m}} =\phi_j \sigma^j_{A_1\cdots A_{m}}, \qquad \psi_{A_3\cdots A_m} =\psi_k \sigma^k_{A_3\cdots A_m}.
\]
With help of the above notation one can rewrite the symbol \eqref{Symbol:PhiPsi1}
as
\[
\sigma_\bmxi \cdot 
 \left(
\begin{array}{c}
\phi_{A_1\cdots A_m} \\
\psi_{A_3\cdots A_m}
\end{array}
 \right)
 \equiv
 \left(
\begin{array}{c}
S_{A_1\cdots A_m}{}^{P_1\cdots P_m}\sigma^j_{P_1\cdots P_m}\phi_j -2\mathfrak{c} P_{A_1\cdots A_m}{}^{Q_3\cdots Q_m}\sigma^k_{Q_3\cdots Q_m}\psi_l\\
\displaystyle \frac{2(m-1)}{\mathfrak{c}m} Q_{A_3\cdots A_m}{}{}^{P_1\cdots P_m}\sigma^j_{P_1\cdots P_m}\phi_j +R_{A_3\cdots A_m}{}^{Q_3\cdots Q_m}\sigma^l_{Q_3\cdots Q_m}\psi_l
\end{array}
 \right)
\]
so that, contracting the first entry with $\sigma_i{}^{A_1\cdots A_m}$ and the second with $\sigma_k{}^{A_3\cdots A_m}$ 
one obtains the \emph{matrix form} of the symbol ---namely
\[
\sigma_\bmxi \cdot 
 \left(
\begin{array}{c}
\phi_i \\
\psi_k
\end{array}
 \right)
 \equiv
 \left(
\begin{array}{cc}
S_i{}^j & -2\mathfrak{c} P_i{}^l\\
\displaystyle \frac{2(m-1)}{\mathfrak{c}m} Q_k{}^j &  R_k{}^l
\end{array}
 \right)
\left(
\begin{array}{c}
\phi_j\\
\psi_l
\end{array}
\right)
\]
where the various blocks in the above expression given by 
\begin{eqnarray*}
    && S_i{}^j \equiv \sigma_i{}^{A_1\cdots A_m} S_{A_1\cdots A_m}{}^{P_1\cdots P_m}\sigma^j_{P_1\cdots P_m}, \\
    && P_i{}^l \equiv \sigma_i{}^{A_1\cdots A_m} P_{A_1\cdots A_m}{}^{Q_3\cdots Q_m}\sigma^k_{Q_3\cdots Q_m}, \\
    && Q_k{}^i \equiv \sigma_k{}^{A_3\cdots A_m} Q_{A_3\cdots A_m}{}^{P_1\cdots P_m}\sigma^j_{P_1\cdots P_m}, \\
    && R_k{}^l \equiv \sigma_k{}^{A_3\cdots A_m} R_{A_3\cdots A_m}{}^{Q_3\cdots Q_m}\sigma^l_{Q_3\cdots Q_m}.
\end{eqnarray*}
In order to discuss the properties of the above symbol in relation to \emph{symmetric hyperbolicity}, it is more convenient to consider the manifestly more symmetric matrix
\[
\mathbf{A}[\bmxi] \equiv
\left(
\begin{array}{cc}
S_i{}^j & -2\mathfrak{c} P_i{}^l\\
 2\mathfrak{c} Q_k{}^j &  \displaystyle\frac{\mathfrak{c}^2 m}{m-1} R_k{}^l
\end{array}
\right)
 .
\]

\medskip
Now, to verify that the matrix $\mathbf{A}[\bmxi]$ satisfies property (i) of symmetric hyperbolicity in Section \ref{Section:SHS} set, for simplicity $\xi=2$ and $\xi_{AB}=0$ so that $\xi_{AA'}=\tau_{AA'}$. In this case, it readily follows that 
\[
\mathbf{A}[\tau] \equiv
\left(
\begin{array}{cc}
\delta_i{}^j & 0\\
0 &  \displaystyle\frac{\mathfrak{c}^2 m}{m-1} \delta_k{}^l
\end{array}
 \right)
\]
which is manifestly a positive definite matrix. Property (ii) (Hermiticity) can be expressed in terms of properties of the various blocks ---more precisely, one requires that
\[
(S_i{}^j)^*=(S_i{}^j), \quad (R_k{}^l)^*=R_k{}^l, \quad (P_i{}^l)^*= -(Q_l{}^i)
\]
where $i,\,j=0,\ldots, m$ and $k,\,l=0,\cdots m-2$. These properties can be readily verified recalling \eqref{RealityConditions:Xi} and that $\hat\epsilon_{AB}=\epsilon_{AB}$. 

\begin{remark}
    {\em Ultimately, the symmetric hyperbolicity of the $\phi$-$\psi$ system can be traced back to the specific combination of operators appearing in the principal part  \eqref{PrincipalPart1}-\eqref{PrincipalPart2}. More precisely, in \eqref{PrincipalPart1} the field $\phi_{A_1\cdots A_m}$ is being acted upon by the Fermi derivative $\mathcal{D}$ and the curl operator ---the later being a self-adjoint elliptic operator in the space of symmetric spinors. The field $\psi_{A_3\cdots A_m}$ is being acted upon by the symmetric tracefree derivative which is an elliptic underdetermined operator. Similarly, in \eqref{PrincipalPart2} the field $\psi_{A_3\cdots A_m}$ is being acted upon by $\mathcal{D}$ and the curl operator while $\phi_{A_1\cdots A_m}$ is being acted upon by the divergence operator. It is important to observe that the divergence and the symmetric tracefree derivative are formal adjoints of each other in the space of symmetric spinors. This structural observation allows to recognise by mere inspection potential symmetric hyperbolic systems.  
    }
    \end{remark}

\subsubsection{Further examples}
We illustrate the previous general discussion with a number of additional examples. 

\medskip
\noindent
\textbf{The massless spin-$s$ equation.} Arguably, the simplest example of an equation fitting the template of the model equation \eqref{PrototypeSpinorialEqn} is the massless spin-$s$ equation satisfied by a totally symmetric valence $2s$ spinor $\phi_{A_1\cdots A_{2s}}$ with $s\in \tfrac{1}{2}\mathbb{N}$ ---namely,
\begin{equation}
\nabla^Q{}_{A'}\phi_{QA_2\cdots A_{2s}}=0.
\label{MasslessEqn}
\end{equation}
Observe that this equation corresponds to the spacial case of the $\phi$-$\psi$ system where the spinors $\psi_{A_3\cdots A_m}, G_{A'A_2\cdots A_m}{}^{Q_1\cdots Q_{m}}$ and $f_{A'A_2\cdots A_m}$ vanish.
A direct application of the space spinor formalism allows to decompose this equation as
\begin{subequations}
\begin{eqnarray}
&& \mathcal{D}\phi_{A_1\cdots A_{2s}} -2 \mathcal{D}^Q{}_{(A_1}\phi_{A_2\cdots A_{2s})Q}=0, \label{Massless:EvolutionEqn}\\
&& \mathcal{D}^{PQ}\phi_{PQA_3\cdots A_{2s}}=0. \label{Massless:ConstraintEqn}
\end{eqnarray}
\end{subequations}
Thus, one sees that equation \eqref{Massless:EvolutionEqn} plays the role of an evolution equation while \eqref{Massless:ConstraintEqn} is a constraint. Finally, observe that applying $\nabla_B{}^{A'}$ to equation \eqref{MasslessEqn} and using the identity \eqref{Commutator:Reduced} one obtains the wave equation
\[
\square \phi_{A_1\cdots A_{2s}} +2(2s-1)\ \Psi^{QR}{}_{(A_1 A_2}\phi_{A_3\cdots A_{2s})QR} -(2s-1) \Lambda \phi_{A_1\cdots A_{2s}} =0.
\]
If one contracts any two of the free indices  in the above equation, the first and third terms vanish by the symmetry of $\phi_{A_1 \ldots A_{2s}}$, implying that the second term must also vanish. When $s>1$, one therefore recovers the well-known \emph{Buchdahl constraint} \cite{PenRin84,Ste91} ---namely that
\[
\phi_{ABM(C\ldots K} \Psi_{L)}{}^{ABM} = 0.
\]
These restrictions do not arise in the cases $s=\frac{1}{2}$ and $s=1$ ---i.e. for the Dirac and Maxwell fields.

\medskip
\noindent
\textbf{The wave equation for a scalar field.} The scalar wave equation
\[
\square \phi =0
\]
can be recast in first order form so that it fits the scheme of the prototype equation \eqref{PrototypeSpinorialEqn}. To this end, introduce the auxiliary variable $\phi_{AA'}\equiv \nabla_{AA'}\phi$. Now, the definition of $\phi_{AA'}$ together with the torsion-freeness of  $\nabla_{AA'}$ implies that 
\[
\nabla_{(A}{}^{Q'}\phi_{B)Q'}=0.
\]
Thus, from the identity 
\[
\nabla_A{}^{Q'}\phi_{BQ'} = \nabla_{(A}{}^{Q'}\phi_{B)Q'} -\frac{1}{2}\nabla^{QQ'}\phi_{QQ'},
\]
together with $\square\phi =\nabla^{QQ'}\phi_{QQ'}$ one concludes that 
\[
\nabla_A{}^{Q'}\phi_{BQ'}=0.
\]
This equation is similar to the prototype equation \eqref{SimplifiedPrototypeEqn} except for the fact that the contraction between the covariant derivative and the spinor is made on primed indices. 

 Defining $\varphi\equiv \tau^{AA'}\phi_{AA'}$ and $\varphi_{AB}\equiv \tau_{(B}{}^{A'}\phi_{A)A'}$, one has the decomposition
\[
\phi_{AA'}=\frac{1}{2}\varphi \tau_{AA'} -\tau^Q{}_{A'}\varphi_{AQ}.
\]
In the following, for simplicity of the presentation, it is assumed that the Weingarten spinor, $\chi_{ABCD}$, of $\tau_{AA'}$   vanishes. Under these assumptions, a space spinor split readily yields the following system of evolution equations:
\begin{eqnarray*}
&& \mathcal{D}\varphi + 2 \mathcal{D}^{AB}\varphi_{AB}=0,\\
&& \mathcal{D}\varphi_{AB} -\mathcal{D}_{AB}\varphi + 2 \mathcal{D}_{(A}{}^Q \varphi_{B)Q}=0.
\end{eqnarray*}

\medskip
\noindent
\textbf{The Bianchi identity for the tracefree Ricci tensor.} One of the constituents of the conformal Einstein field equations is the Bianchi equation for the Schouten tensor ---see e.g. \cite{CFEBook}. In its spinorial form this equation is given by
\begin{equation}
\nabla_A{}^{Q'}L_{BQ'CC'}+ \Sigma^Q{}_{C'}\phi_{ABCQ} +\frac{1}{12}\epsilon_{AB}\nabla_{CC'} R(x)
\label{BianchiRicci}
\end{equation}
where $L_{BQ'CC'}$ is the spinor counterpart of the Schouten tensor, $\phi_{ABCQ}$ is the Rescaled Weyl spinor, $\Sigma_{QC'}$ encodes the derivative of a conformal factor and $R(x)$ is the Ricci scalar ---the latter plays the role of a \emph{conformal gauge source function} which can be freely specified. Now, the Schouten tensor admits the decomposition
\[
L_{AA'CC'}=\Phi_{AA'CC'}+\frac{1}{24}\epsilon_{AC}\epsilon_{A'C'}R(x),
\]
where $\Phi_{AA'CC'}$ is the spinorial counterpart of the tracefree part of the Ricci tensor. The space spinor version of $L_{AA'CC'}$ is defined as
\begin{eqnarray*}
&& L_{ABCD} \equiv \tau_B{}^{A'}\tau_D{}^{C'} L_{AA'CC'}\\
&& \phantom{L_{ABCD}}=\Phi_{ABCD}
+\frac{1}{24}\epsilon_{AC}\epsilon_{BD}R(x),
\end{eqnarray*}
where $\Phi_{ABCD}\equiv \tau_B{}^{A'}\tau_D{}^{C'}\Phi_{AA'CC'}$ so that one has the symmetries
\[
\Phi_{ABCD}=\Phi_{CBAD}=\Phi_{ADCB}.
\]
A spinor with these symmetries admits the decomposition
\[
\Phi_{ABCD} = \Phi_{(ABCD)}+ \frac{1}{2}\big( \epsilon_{A(B}\Phi_{D)C} + \epsilon_{C(B}\Phi_{D)A} \big) + \frac{1}{3}\Phi h_{ABCD},
\]
where $h_{ABCD}\equiv \epsilon_{A(C}\epsilon_{D)B}$ and
\[
\Phi_{AB}\equiv \Phi_{(AB)Q}{}^Q, \qquad \Phi\equiv \Phi_{ABCD}h^{ABCD}.
\]
Making use of the above decomposition in equation \eqref{BianchiRicci}, one obtains a system for the fields $\Phi_{(ABCD)}$, $\Phi_{AB}$ and $\Phi$ whose principal part is given by
\begin{eqnarray*}
&& \mathcal{D}\Phi_{(ABCD)}-\mathcal{D}_{(AB}\Phi_{CD)}, \\
&& \mathcal{D}\Phi_{AB} + 2 \mathcal{D}^{PQ}\Phi_{(PQAB)}
-\frac{1}{3}\mathcal{D}_{AB}\Phi,\\
&& \mathcal{D}\Phi +\mathcal{D}^{PQ}\Phi_{PQ}.
\end{eqnarray*}
More details of the above calculation can be found in \cite{CFEBook}, subsection 13.2.3.

\section{Strategies for the 
construction of estimates for spinor fields}
\label{General-prescription-estimates}

In this section, we review the basic strategy for the construction of estimates using the so-called \emph{positive commutator method} and discuss how this strategy can be adapted to the analysis of spinor fields and equation \eqref{PrototypeSpinorialEqn}. 

\subsection{Basic notions}
Before providing an overview of the positive commutator method, we briefly review some basic concepts and definitions which will be used throughout.

\subsubsection{Inner product}
\label{Section:InnerProduct}
Let $(\mathcal{M}, \bmg)$ denote a $4$-dimensional Lorentzian manifold with metric $\bmg$ and let $\mathcal{U}\subset \mathcal{M}$ be a subset with boundary $\partial \mathcal{U}$.  If $\mathrm{d} \mu$ denotes a measure on $\mathcal{U}$, the inner product of two symmetric spinor fields $\psi_{A \ldots F}$ and $\gamma_{A \ldots F}$ will be defined as
\begin{equation}
     \innerbrackets{\psi}{\gamma}\equiv \int_{\mathcal{U}} \psi_{A_{1} \ldots A_{p}} \widehat{\gamma}^{ A_{1} \ldots A_{p}} \mathrm{d}\mu,
     \label{Inner-product}
\end{equation}
where $\widehat{\gamma}^{A \ldots F}$ is the Hermitian conjugate of $\gamma^{A \ldots F}$ defined  earlier ---see equation\eqref{Hermitian-conjugate}.

\begin{remark}
{\em The measure $d\mu$ is not necessarily the volume form induced by the metric $\bmg$. When this is the case, we make use of the symbol $d\mu_\bmg$. }
\end{remark}

Let $\mathfrak{S}(\mathcal{U})$ denote the $\text{SL}(2,\mathbb{C})$-spinor bundle on $\mathcal{U}$, then the inner product satisfies the  properties:
\begin{enumerate}[i.]
    \item $\overline{\innerbrackets{\psi}{\gamma}} = \innerbrackets{\gamma}{\psi} \quad \psi,\gamma \in \mathfrak{S}(\mathcal{U})$;
    \item $\innerbrackets{a \psi_{1}+b \psi_{2}}{\gamma} = a \innerbrackets{\psi_{1}}{\gamma} + b \innerbrackets{\psi_{2}}{\gamma} \quad a,b \in \mathbb{C} \; \quad  \;  \psi_{1}, \psi_{2}, \gamma \in \mathfrak{S}(\mathcal{U})$;
    \item $\innerbrackets{\psi}{a \gamma_{1}+b \gamma_{2}} = \bar{a} \innerbrackets{\psi}{\gamma_{1}} + \bar{b} \innerbrackets{\psi}{\gamma_{2}} \quad a,b \in \mathbb{C} \; \quad \;  \psi, \gamma_{1}, \gamma_{2} \in \mathfrak{S}(\mathcal{U})$;
    \item for $\psi \in \mathfrak{S}(\mathcal{U})$, $\innerbrackets{\psi}{\psi} >0 \text{ if and only if } \psi \neq 0$.
\end{enumerate}

For any linear operator $\bmA$, we define its  adjoint (with respect to the measure $d\mu$) by
\begin{equation*}
    \innerbrackets{\bmA \psi}{ \gamma} = \innerbrackets{\psi}{\bmA^{*} \gamma} \quad  \psi,\gamma \in \mathfrak{S}(\mathcal{U}).
\end{equation*}
For scalar fields $\phi,\psi$ on $\mathcal{U}$, the inner product $\innerbrackets{\cdot}{\cdot}$ as defined by equation \eqref{Inner-product} reduces to
\begin{equation*}
    \innerbrackets{\phi}{\psi} = \int_{\mathcal{U}} \phi \bar{\psi} \mathrm{d}\mu,
\end{equation*}
where $\bar{\psi}$ is the complex conjugate of $\psi$.

\subsubsection{The domain of integration}
\label{Subsection:DomainOfIntegration}

\begin{figure}[t]
\begin{center}
\includegraphics[width=0.4\textwidth]{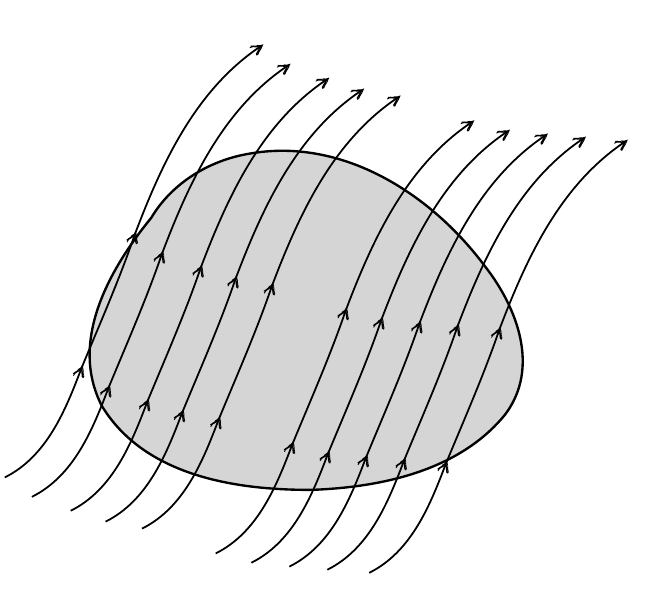}
\put(-100,70){$\mathcal{U}$}
\put(-33,60){$\partial\mathcal{U}$}
\put(-26,100){$\bmzeta$}
\end{center}
\label{DomainOfIntegration}
\caption{Schematic representation of the domain  on which estimates are computed. The domain $\mathcal{U}$ is assumed, in general, to have a boundary $\partial\mathcal{U}$. The causal nature of the boundary is, in principle, arbitrary. Crucially, the domain $\mathcal{U}$ is assumed to be covered by a nonsingular congruence generated by the integral curves of a timelike vector field $\bmzeta$. }
\end{figure}

The domain $\mathcal{U}\subset \mathcal{M}$ over which the inner product $\innerbrackets{\cdot}{\cdot}$ defined is assumed to be compact with boundary $\partial \mathcal{U}$. The particular assumptions on $\mathcal{U}$ very much depend on the specific applications under consideration. Here, we only assumed that $\mathcal{U}$ is covered by a non-intersecting congruence generated by the integral curves of a timelike vector field $\bmzeta$ ---see Figure \ref{DomainOfIntegration}.  

\subsubsection{Some ancillary definitions}
The positive commutator method make use of a carefully chosen vector field multiplier $\zeta^a$ ($\bmzeta$). In the following, let  $\zeta^{AA'}$ denote its spinorial counterpart. By analogy to the tensor case, we define the \emph{deformation spinor} of $\zeta^{AA'}$ as
\begin{equation*}
    \Pi_{AA'BB'} \equiv \frac{1}{2} (\nabla_{AA'} \zeta_{BB'} + \nabla_{BB'}\zeta_{AA'}),
\end{equation*}
where $\nabla_{AA'}$ is the spinor version of the Levi-Civita connection of the metric $\bmg$. The deformation spinor $\Pi_{AA'BB'}$ can be written in terms of the trace $\Sigma$ and trace-free part $\Sigma_{AA'BB'}$ as
\begin{equation*}
    \Pi_{AA'BB'} = \Sigma_{AA'BB'} + \frac{1}{4} \Sigma \epsilon_{AB} \bar{\epsilon}_{A'B'},
\end{equation*}
where
\begin{subequations}
    \begin{eqnarray}
        && \Sigma_{AA'BB'} \equiv  \frac{1}{4} (\nabla_{AA'} \zeta_{BB'} + \nabla_{AB'} \zeta_{BA'} + \nabla_{BA'} \zeta_{AB'} + \nabla_{BB'} \zeta_{AA'}), \\
        && \Sigma \equiv \nabla^{AA'} \zeta_{AA'}.
    \end{eqnarray}
    \label{Definitions-Sigma}
\end{subequations}
Then, the irreducible decomposition of the derivative of $\zeta^{AA'}$ is given by
\begin{equation*}
    \nabla_{AA'} \zeta_{BB'} = \Sigma_{AA'BB'} - \frac{1}{2} \Xi_{AB} \bar{\epsilon}_{A'B'} - \frac{1}{2} \bar{\Xi}_{A'B'} \epsilon_{AB} + \frac{1}{4} \Sigma \epsilon_{AB} \bar{\epsilon}_{A'B'},
\end{equation*}
where 
\begin{subequations}
    \begin{eqnarray}
    && \Xi_{AB} \equiv \frac{1}{2} (\nabla_{A}{}^{Q'} \zeta_{BQ'} + \nabla_{B}{}^{Q'} \zeta_{AQ'}), \\
    && \bar{\Xi}_{A'B'} \equiv \frac{1}{2} (\nabla^{Q}{}_{A'} \zeta_{QB'} + \nabla^{Q}{}_{B'} \zeta_{QA'}). 
    \end{eqnarray}
    \label{Definitions-Xi}
\end{subequations}
We now define the \emph{K-current} of $\zeta^{AA'}$ as
\begin{equation*}
    K_{AA'BB'}\equiv  \frac{1}{2} (\nabla_{AA'}\zeta_{BB'}+ \nabla_{BB'}\zeta_{AA'} - \epsilon_{AB} \bar{\epsilon}_{A'B'} \nabla_{QQ'} \zeta^{QQ'}). 
\end{equation*}
In terms of $\Pi_{AA'BB'}$, the K-current can be written as
\begin{equation*}
    K_{AA'BB'} = \Pi_{AA'BB'} - \frac{1}{2} \Pi^{QQ'}{}_{QQ'} \epsilon_{AB} \bar{\epsilon}_{A'B'}.
\end{equation*}
Making use of the K-current one obtains the following alternative expression for the decomposition of
 the derivative of $\zeta^{AA'}$:
\begin{equation*}
    \nabla_{AA'} \zeta_{BB'} = K_{AA'BB'} - \frac{1}{2} \Xi_{AB} \bar{\epsilon}_{A'B'} - \frac{1}{2} \bar{\Xi}_{A'B'} \epsilon_{AB} + \frac{1}{2} \Sigma \epsilon_{AB} \bar{\epsilon}_{A'B'}.
\end{equation*}
In later calculations, it will be convenient to specify the choice of $\bmzeta$ as 
\begin{equation}
    \zeta^{AA'} =\varpi^2 \tau^{AA'}.
    \label{Definition-zeta-spinorial}
\end{equation}
for some suitable weight $\varpi$ and $\tau^{AA'}$ as described in Section \ref{Section:Space-spinor-formalism} so that, in particular, $\tau_{AA'}\tau^{AA'}=2$.
\subsection{A model problem: the scalar wave equation}
\label{Section:TheScalarWaveEqn}

The construction of estimates for the scalar wave equation using the \emph{positive commutator method} will be the fundamental point of reference in our subsequent discussion. As such, we provide a brief review of it.

\subsubsection{Basic identities}
In the following, we will consider the non-homogeneous wave equation
\begin{equation}
\square \phi  =f,
\label{WaveEquation}
\end{equation}
where $\square \equiv \nabla^{PP'}\nabla_{PP'}$ is the usual D'Alambertian operator ---expressed in terms of the spinor covariant derivative. The scalar $f$ is a suitable (known) source term. The scalar field $\phi$ will be assumed to be complex. 

\medskip
Given a domain $\mathcal{U}\subset \mathcal{M}$  with boundary $\partial \mathcal{U}$, we will denote by $\dot{C}^\infty(\mathcal{U})$ the set of functions vanishing to infinite order at $\partial \mathcal{U}$. Let $\bmzeta$ denote a real vector field and define its action on $\phi \in \dot{C}(\mathcal{U})$ by
\begin{equation}
    \bmzeta \phi \equiv \zeta^{a} \nabla_{a} \phi.
    \label{action-zeta}
\end{equation}
Then, its adjoint $\bmzeta^*$ is defined by
\begin{equation*}
    \innerbrackets{\bmzeta\phi}{\psi} =\innerbrackets{\phi}{\bmzeta^*\psi}, \qquad \phi,\,\psi\in \dot{C}(\mathcal{U}).
\end{equation*}
Given \eqref{action-zeta}, one can show that
\begin{equation}
\bmzeta^* = -\bmzeta - \nabla_a \zeta^a.
\label{IdentityAdjointZeta}
\end{equation}
In terms of the above, the \emph{self-adjoint commutator} of $\square$ and $\bmzeta$ is defined as
\[
\bmA \equiv \bmzeta^* \square + \square \bmzeta.
\]
It readily follows from this definition that $\bmA =\bmA^*$. A further computation then shows that 
\begin{eqnarray*}
&&\bmA \phi = \Pi^{AA'BB'} \nabla_{AA'} \nabla_{BB'}\phi - \frac{1}{2}\Pi^{PP'}{}_{PP'} \square\phi + \left(\nabla_{BB'} \Pi^{BB'AA'} -\frac{1}{2}\nabla^{AA'}\Pi^{BB'}{}_{BB'}   \right)\nabla_{AA'}\phi\\
&& \phantom{\bmA \phi} = T_{CC'DD'}{}^{AA'BB'} \Pi^{CC'DD'} \nabla_{AA'}\nabla_{BB'}\phi + \nabla_{BB'}\left( T_{CC'DD'}{}^{AA'BB'} \Pi^{CC'DD'}\right)\nabla_{AA'}\phi\\
&& \phantom{\bmA \phi} = 2 \nabla_{AA'}\left( T_{CC'DD'}{}^{AA'BB'}\Pi^{CC'DD'}\nabla_{AA'}\phi \right)\\
&& \phantom{\bmA \phi} = 2 \nabla_{BB'}\left( K^{AA'BB'}\nabla_{AA'}\phi \right),
\end{eqnarray*}
where 
\begin{subequations}
\begin{eqnarray}
    && T_{AA'BB'}{}^{CC'DD'}\equiv \frac{1}{2} \left( \delta_{A}{}^{C} \delta_{A'}{}^{C'} \delta_{B}{}^{D} \delta_{B'}{}^{D'} + \delta_{A}{}^{D} \delta_{A'}{}^{D'} \delta_{B}{}^{C} \delta_{B'}{}^{C'} + \epsilon_{AB} \bar{\epsilon}_{A'B'} \epsilon^{CD} \epsilon^{C'D'} \right), \label{Definition:AbstractEnergyMomentumTensor} \qquad  \\
    && K^{AA'BB'} \equiv T_{CC'DD'}{}^{AA'BB'} \Pi^{CC'DD'} \nonumber \\
    && \phantom{K^{AA'BB'}} = \Pi^{AA'BB'} - \frac{1}{2} \Pi^{QQ'}{}_{QQ'} \epsilon^{AB} \epsilon^{A'B'}. \label{Definition:KCurrent}
\end{eqnarray}
\end{subequations}

\medskip
Integrating by parts, one finds that 
\begin{eqnarray*}
&& \innerbrackets{\bmA \phi}{\phi} = \int_{\mathcal{U}}\bar{\phi}\bmA \phi \mathrm{d}\mu_\bmg \\
&& \phantom{\innerbrackets{\bmA \phi}{\phi}} = -2 \int_{\mathcal{U}}K^{AA'BB'}\nabla_{AA'}\phi \nabla_{BB'}\bar\phi \, \mathrm{d}\mu_\bmg, \qquad \phi \in \dot{C}^\infty(\mathcal{U}).
\end{eqnarray*}

One can summarise the above discussion in the following:

\begin{lemma}
\label{Lemma:IdentitySelfadjointCommutatorWaveEquation}
Given $\phi\in \dot{C}^\infty (\mathcal{U})$, one 
has that 
\[
\innerbrackets{\bmA\phi}{\phi}=-2 \innerbrackets{K^{AA'BB'}\nabla_{AA'}\phi\nabla_{BB'}\bar{\phi}}{1}.
\]
\end{lemma}

\subsubsection{Construction of estimates}
In the remainder of this subsection, we show how the above result can be used to construct estimates for the scalar field. From equation \eqref{WaveEquation} and the definition of the self-adjoint commutator $\bmA$, it follows the identity
\[
\innerbrackets{\bmA\phi}{\phi}= 2 \mbox{Re}\, \innerbrackets{\bmzeta \phi}{f},
\]
and we no longer assume that $\phi\in\dot{C}^\infty (\mathcal{U})$. 
From the above, using Lemma \ref{Lemma:IdentitySelfadjointCommutatorWaveEquation}, one can then write
\[
-2 \innerbrackets{K^{AA'BB'}\nabla_{AA'}\phi\nabla_{BB'}\bar{\phi}}{1}\approx  2 \mbox{Re}\, \innerbrackets{\bmzeta \phi}{f},
\]
where in the above, and in the following, the symbol $\approx$ is used to denote equality up to boundary terms. Now, it is observed that the term $K^{AA'BB'}\nabla_{AA'}\phi\nabla_{BB'}\phi$  is a quadratic form on the components of the gradient of $\phi$, $\mathbf{d}\phi$. If this quadratic form has some special properties, one can then, in turn, use them to control a suitable norm of $\mathbf{d}\phi$. A particular case of interest arises when $\bmzeta$ is chosen so that $\bmK$ is negative definite ---say, there exists a constant $\mathfrak{K}>0$ such that 
\begin{equation}
    \bmK (\mathbf{d}\phi,\mathbf{d}\bar\phi) \leq - \mathfrak{K} | \mathbf{d}\phi|^2,
    \label{Control-K-dphi-dphi}
\end{equation}
where $| \mathbf{d}\phi|^2$ is the standard norm defined by
\begin{equation*}
    |f|^2 \equiv f \cdot \bar{f}, 
\end{equation*}
for any complex function $f$. In addition to the above, it is also assumed that there is a constant $\mathfrak{Z}>0$ such that
\begin{equation}
||\bmzeta \phi||^2 \leq \mathfrak{Z} ||\mathbf{d}\phi||^2.
\label{BoundZetaPhi}
\end{equation}
In the above, $||\bmzeta \phi||^2 \equiv \innerbrackets{\bmzeta \phi}{\bmzeta \phi}$ and similarly for $\mathbf{d}\phi$. Under the above circumstances and using Cauchy–Schwarz and Young's inequalities to show that 
\[
2 \mbox{Re}\, \innerbrackets{\bmzeta \phi}{f} \leq ||f||^2+ ||\bmzeta \phi ||^2 ,
\]
one concludes that 
\begin{equation}
2(\mathfrak{K}-\mathfrak{Z}) || \mathbf{d}\phi||^2 \preccurlyeq ||f||^2 + ||\bmzeta \phi||^2,
\label{IntermmediateEstimateWaveEquation}
\end{equation}
where $\preccurlyeq$ denotes $\leq$ up to boundary terms. The above inequality provides a non-trivial bound on $||\mathbf{d}\phi||$ in terms of $||f||$ and boundary terms.

\medskip
Now, the construction of estimates for Sobolev norms requires control also on the norm $||\phi||$. In order to obtain the required control, start by observing that
\begin{eqnarray*}
&& 2 \mbox{Re}\, \innerbrackets{\bmzeta \phi}{\phi} = \innerbrackets{\bmzeta\phi}{\phi} +\overline{\innerbrackets{\bmzeta\phi}{\phi}}\\
&& \phantom{2 \mbox{Re}\, \innerbrackets{\bmzeta \phi}{\phi}} = \innerbrackets{\bmzeta\phi}{\phi} + \innerbrackets{\phi}{\bmzeta \phi}\\
&& \phantom{2 \mbox{Re}\, \innerbrackets{\bmzeta \phi}{\phi}} \approx \innerbrackets{(\bmzeta+\bmzeta^*)\phi}{\phi}.
\end{eqnarray*}
Using equation \eqref{IdentityAdjointZeta}, one then obtains
\begin{equation}
    2 \mbox{Re}\, \innerbrackets{\bmzeta\phi}{\phi} \approx -\innerbrackets{(\mbox{div}_\bmg \bmzeta)\phi}{\phi}.
    \label{Bounds-divzeta}
\end{equation}
To exploit this identity, it is further assumed that there is a constant $\mathfrak{D}>0$ such that
\[
\mbox{div}_\bmg \bmzeta \leq -\mathfrak{D}.
\]
If this is the case, using the bound \eqref{BoundZetaPhi}, the Cauchy–Schwarz and Young's inequalities to show that $ -\innerbrackets{(\mbox{div}_\bmg \bmzeta)\phi}{\phi} \approx 2 \mbox{Re}\, \innerbrackets{\bmzeta \phi}{\phi} \leq  \mathfrak{Z} ||\mathbf{d}\phi||^2 + ||\phi||^2$, one obtains
\begin{equation}
(\mathfrak{D}-1)||\phi||^2 \preccurlyeq \mathfrak{Z}||\mathbf{d}\phi||^2.
\label{OtherCommutatorScalarField}
\end{equation}
Accordingly, one obtains a nontrivial inequality if
\[
\mathfrak{D}>1. 
\]
To conclude the argument, one adds inequalities \eqref{IntermmediateEstimateWaveEquation} and \eqref{BoundZetaPhi} so as to obtain
\[
(\mathfrak{D}-1)||\phi||^2 + (2\mathfrak{K}-4\mathfrak{Z}) ||\mathbf{d}\phi||^2 \preccurlyeq ||f||^2.
\]
This last inequality gives suitable control  over $||\phi||$ and $||\mathbf{d}\phi||$ in terms of $||f||$ if the constants $\mathfrak{D}$, $\mathfrak{Z}$ and $\mathfrak{K}$ satisfy
\[
\mathfrak{D}>1, \qquad 2\mathfrak{K}> 4\mathfrak{Z}.
\]
Observe that the above conditions are, ultimately, conditions on the vector field multiplier $\bmzeta$.

\begin{proposition}
Assume that there exist constants $\mathfrak{D},\, \mathfrak{K}, \, \mathfrak{Z}>0$ satisfying
\[
\mathfrak{D}>1, \qquad 2\mathfrak{K}> 4\mathfrak{Z}.
\]
and such that the bounds \eqref{Control-K-dphi-dphi}, \eqref{BoundZetaPhi} and \eqref{Bounds-divzeta} hold. Then the solutions to equation \eqref{WaveEquation} are bounded by
\[
(\mathfrak{D}-1)||\phi||^2 + (2\mathfrak{K}-4\mathfrak{Z}) ||\mathbf{d}\phi||^2 \preccurlyeq ||f||^2.
\]
\end{proposition}

\subsection{The $\phi$-$\psi$ system}
\label{Subsection:EstimatesPhiPsi}

In this subsection, we discuss how the strategy of constructing estimates by means of the \emph{positive commutator method} can be adapted to first order spinor equations. Motivated by the definition of inner product in Section \ref{Section:InnerProduct}, we consider first a general discussion of the $\phi$-$\psi$ system \eqref{SimplifiedPrototypeEqnSpaceSpinor:ComponentPhi}-\eqref{SimplifiedPrototypeEqnSpaceSpinor:ComponentPsi} and later consider specific features of a particular subcase.

\subsubsection{The main identity}
\label{Subsubsection:TheMainIdentity}

Given a vector field $\bmzeta= \zeta^\mu \partial_\mu$ as in  Section \eqref{Section:TheScalarWaveEqn} ---in abstract index notation, we denote this vector by $\zeta^a$ with spinorial counterpart $\zeta^{AA'}$. In the following, we assume that the Hermitian spinor $\zeta^{AA'}$ is of the form \eqref{Definition-zeta-spinorial} and that $\tau^{AA'}$ is defined as in Section \ref{Section:Space-spinor-formalism}. 

\medskip
Then, the starting point of the subsequent discussion will be equations \eqref{SpaceSpinorEvolutionSystem1}-\eqref{SpaceSpinorEvolutionSystem2} written in terms of the standard spinorial covariant derivative $\nabla_{AA'}$ and with the replacement $\tau^{AA'}\mapsto \zeta^{AA'}$ ---that is:
\begin{eqnarray*}
&& \zeta_{(A_1}{}^{A'}\nabla^Q{}_{|A'|}\phi_{A_2\cdots A_m)Q} + \mathfrak{c} \zeta_{(A_1}{}^{A'}\nabla_{A_2|A'|}\psi_{A_3\cdots A_m)}\\
&& \hspace{4cm}+ \varpi^{2} G_{(A_1A_2\cdots A_m)}{}^{Q_1\cdots Q_{m}}\varphi_{Q_1\cdots Q_{m}} = \varpi^{2} f_{(A_1A_2\cdots A_m)},  \\
&& \zeta^{QA'}\nabla^{P}{}_{A'}\phi_{PQA_3\cdots A_m} + \mathfrak{c} \zeta^{PA'}\nabla_{(P|A'|}\psi_{A_3\cdots A_m)}\\
&& \hspace{4cm}+ \varpi^{2} G^P{}_{PA_3\cdots A_m}{}^{Q_1\cdots Q_{m}}\varphi_{Q_1\cdots Q_{m}} = \varpi^{2} f^P{}_{PA_3\cdots A_m}.
\end{eqnarray*}

For convenience, define
\begin{eqnarray*}
&& \mathcal{E}_{A_1\cdots A_m} \equiv \zeta_{(A_1}{}^{A'}\nabla^Q{}_{|A'|}\phi_{A_2\cdots A_m)Q} + \mathfrak{c} \zeta_{(A_1}{}^{A'}\nabla_{A_2|A'|}\psi_{A_3\cdots A_m)}, \\
&& \mathcal{F}_{A_3\cdots A_m} \equiv \zeta^{QA'}\nabla^{P}{}_{A'}\phi_{PQA_3\cdots A_m} + \mathfrak{c} \zeta^{PA'}\nabla_{(P|A'|}\psi_{A_3\cdots A_m)},
\end{eqnarray*}
and, in terms of the latter 
\begin{eqnarray*}
&& \mathbf{E}(\phi,\psi)\equiv \widehat{\phi}^{A_1\cdots A_m} \mathcal{E}_{A_1\cdots A_m} + (-1)^m\phi^{A_1\cdots A_m}\widehat{\mathcal{E}}_{A_1\cdots A_m}, \\
&& \mathbf{F}(\phi,\psi) \equiv \widehat{\psi}^{A_3\cdots A_m} \mathcal{F}_{A_3\cdots A_m} + (-1)^m\psi^{A_3\cdots A_m}\widehat{\mathcal{F}}_{A_3\cdots A_m}.
\end{eqnarray*}

It is also noted that 
\begin{eqnarray*}
    && \widehat{\mathcal{E}}_{A_1\cdots A_m} = \zeta^Q{}_{Q'} \nabla_{(A_1}{}^{Q'} \widehat{\phi}{}_{A_2\cdots A_m)Q} -\mathfrak{c}\zeta_{(A_1}{}^{Q'}\nabla_{A_2|Q'|}\widehat{\psi}{}_{A_3\cdots A_m)} \\
    && \hspace{3cm}+ \tilde{W}_{A_1\cdots A_m}{}^{Q_1\cdots Q_m}\widehat{\phi}_{Q_1\cdots Q_m}+ \tilde{X}_{A_1\cdots A_m}{}^{Q_3\cdots Q_m}\widehat{\psi}_{Q_3\cdots Q_m},\\
    && \widehat{\mathcal{F}}_{A_3\cdots A_m}= \zeta^P{}_{P'}\nabla^{QP'}\widehat{\phi}_{PQA_3\cdots A_m} +\mathfrak{c} \zeta^{AP'}\nabla_{AP'}\widehat{\psi}_{A_3\cdots A_m}\\
    && \hspace{3cm}+ \tilde{Y}_{A_3\cdots A_m}{}^{Q_1\cdots Q_m}\widehat{\phi}_{Q_1\cdots Q_m}+ \tilde{Z}_{A_3\cdots A_m}{}^{Q_3\cdots Q_m}\widehat{\psi}_{Q_3\cdots Q_m},
\end{eqnarray*}
where the spinors
\[
\tilde{W}_{A_1\cdots A_m}{}^{Q_1\cdots Q_m}, \quad \tilde{X}_{A_1\cdots A_m}{}^{Q_3\cdots Q_m}, \quad \tilde{Y}_{A_3\cdots A_m}{}^{Q_1\cdots Q_m}, \quad \tilde{Z}_{A_3\cdots A_m}{}^{Q_3\cdots Q_m},
\]
are build up from the Weingarten spinor of $\tau^{AA'}$, 
$\chi_{ABCD}$, and the weight $\varpi$. It then follows that
\begin{eqnarray*}
&& \mathbf{E}(\phi,\psi) =  \widehat{\phi}^{A_1\cdots A_m} \zeta_{A_1}{}^{Q'}\nabla^Q{}_{Q'}\phi_{A_2\cdots A_m Q} + (-1)^m \phi^{A_1\cdots A_m}\zeta^Q{}_{Q'}\nabla_{A_1}{}^{Q'}\widehat{\phi}_{A_2\cdots A_m Q}\\
&& \hspace{2cm} + \mathfrak{c} \widehat{\phi}^{A_1\cdots A_m} \zeta_{A_1}{}^{Q'}\nabla_{A_2 Q'}\psi_{A_3\cdots A_m} -\mathfrak{c}(-1)^m \phi^{A_1\cdots A_m}\zeta_{A_1}{}^{Q'}\nabla_{A_2Q'}\widehat{\psi}_{A_3\cdots A_m} \\
&& \hspace{2cm} + (-1)^m\big(\phi^{P_1\cdots P_m} \tilde{W}_{P_1\cdots P_m}{}^{Q_1\cdots Q_m}\widehat{\phi}_{Q_1\cdots Q_m} + \phi^{P_1\cdots P_m} \tilde{X}_{P_1\cdots P_m}{}^{Q_3\cdots Q_m}\widehat{\psi}_{Q_3\cdots Q_m}\big),\\
&& \mathbf{F}(\phi,\psi) = \widehat{\psi}^{A_3\cdots A_m} \zeta^{QP'}\nabla^P{}_{P'}\phi_{PQA_3\cdots A_m} + (-1)^m\psi^{A_3\cdots A_m} \zeta^{P}{}_{P'}\nabla^{QP'}\widehat{\phi}_{PQA_3\cdots A_m} \\
&& \hspace{2cm} +\mathfrak{c}\widehat{\psi}^{A_3\cdots A_m} \zeta^{PQ'}\nabla_{PQ'}\psi_{A_3\cdots A_m} + \mathfrak{c}(-1)^m\psi^{A_3\cdots A_m}\zeta^{PP'}\nabla_{PP'}\widehat{\psi}_{A_3\cdots A_m}, \\
&& \hspace{2cm} +(-1)^m\big(\psi^{P_3\cdots P_m} \tilde{Y}_{P_3\cdots P_m}{}^{Q_1\cdots Q_m}\widehat{\phi}_{Q_1\cdots Q_m} + \psi^{P_3\cdots P_m}\tilde{Z}_{P_3\cdots P_m}{}^{Q_3\cdots Q_m}\widehat{\psi}_{Q_3\cdots Q_m}\big).
\end{eqnarray*}
Now, in order to simplify the above expressions, it is observed that 
\begin{eqnarray*}
&& \nabla^Q{}_{Q'}\big( \widehat{\phi}^{A_1\cdots A_m}\zeta_{A_1}{}^{Q'}\phi_{A_2\cdots A_mQ} \big) = \widehat{\phi}^{A_1\cdots A_m} \zeta_{A_1}{}^{Q'}\nabla^Q{}_{Q'} \phi_{A_2\cdots A_m Q} \\
&& \phantom{\nabla^Q{}_{Q'}\big( \widehat{\phi}^{A_1\cdots A_m}\zeta_{A_1}{}^{Q'}\phi_{A_2\cdots A_mQ} \big) } + \phi_{A_2\cdots A_m Q}\zeta_{A_1}{}^{Q'}\nabla^Q{}_{Q'}\widehat{\phi}^{A_1\cdots A_m}\\
&& \phantom{\nabla^Q{}_{Q'}\big( \widehat{\phi}^{A_1\cdots A_m}\zeta_{A_1}{}^{Q'}\phi_{A_2\cdots A_mQ} \big) }+ \Big( \nabla^Q{}_{Q'}\zeta_{A_1}{}^{Q'}\Big)\widehat{\phi}^{A_1\cdots A_m } \phi_{A_2\cdots A_m Q}, \\
&& \nabla_{PP'}\big( \widehat{\psi}^{A_3\cdots A_m}\zeta^{PP'} \psi_{A_3\cdots A_m}  \big)= \widehat{\psi}^{A_3\cdots A_m}\zeta^{PP'}\nabla_{PP'}\psi_{A_3\cdots A_m} \\
&& \phantom{\nabla_{PP'}\big( \widehat{\psi}^{A_3\cdots A_m}\zeta^{PP'} \psi_{A_3\cdots A_m}  \big)} +(-1)^m \psi^{A_3\cdots A_m} \zeta^{PP'} \nabla_{PP'}\widehat{\psi}_{A_3\cdots A_m},\\
&& \phantom{\nabla_{PP'}\big( \widehat{\psi}^{A_3\cdots A_m}\zeta^{PP'} \psi_{A_3\cdots A_m}  \big)} + \Big(\nabla_{PP'}\zeta^{PP'}\Big) \widehat{\psi}{}^{A_3\cdots A_m}\psi_{A_3\cdots A_m}, \\
&& \nabla_{QQ'}\big( \widehat{\phi}{}^{PQ A_3\cdots A_m}\zeta_P{}^{Q'}\psi_{A_3\cdots A_m} \big) = \widehat{\phi}^{PQA_3\cdots A_m}\zeta_P{}^{Q'}\nabla_{QQ'}\psi_{A_3\cdots A_m} \\
 && \phantom{\nabla_{QQ'}\big( \widehat{\phi}{}^{PQ A_3\cdots A_m}\zeta_P{}^{Q'}\psi_{A_3\cdots A_m} \big)}-(-1)^m \psi^{A_3\cdots A_m}\zeta^P{}_{Q'}\nabla^{QQ'}\widehat{\phi}_{PQA_3\cdots A_m} \\
 && \phantom{\nabla_{QQ'}\big( \widehat{\phi}{}^{PQ A_3\cdots A_m}\zeta_P{}^{Q'}\psi_{A_3\cdots A_m} \big)} + \Big(\nabla_{QQ'}\zeta_P{}^{Q'} \Big) \widehat{\phi}^{PQ A_3\cdots A_m}\psi_{A_3\cdots A_m},\\
&& \nabla^P{}_{P'}\big(\widehat{\psi}^{A_3\cdots A_m} \zeta^{QP'} \phi_{PQA_3\cdots A_m}
\big) = \widehat{\psi}{}^{A_3\cdots A_m}\zeta^{QP'}\nabla^P{}_{P'}\phi_{PQA_3\cdots A_m} \\
&& \phantom{\nabla^P{}_{P'}\big(\widehat{\psi}^{A_3\cdots A_m} \zeta^{QP'} \phi_{PQA_3\cdots A_m}
\big)} +(-1)^m\phi^{PQA_3\cdots A_m}\zeta_Q{}^{P'}\nabla_{PP'}\widehat{\psi}_{A_3\cdots A_m} \\
&& \nabla^P{}_{P'}\big(\widehat{\psi}^{A_3\cdots A_m} \zeta^{QP'} \phi_{PQA_3\cdots A_m}
\big)+\Big(\nabla^P{}_{P'}\zeta^{QP'} \Big) \widehat{\psi}{}^{A_3\cdots A_m}\phi_{PQA_3\cdots A_m}.
\end{eqnarray*}
Making use of the above identities, one obtains the following:
\begin{lemma}
\label{Lemma:MainIdentityPhiPsiSystem}
One has that 
\begin{eqnarray*}
&& \mathbf{E}(\phi,\psi)-\mathfrak{c}\mathbf{F}(\phi,\psi) = \nabla^Q{}_{Q'}\big( \widehat{\phi}^{A_1\cdots A_m}\zeta_{A_1}{}^{Q'}\phi_{A_2\cdots A_mQ} \big) + \nabla_{PP'}\big( \widehat{\psi}^{A_3\cdots A_m}\zeta^{PP'} \psi_{A_3\cdots A_m}  \big)\\
&& \hspace{2cm} + \mathfrak{c} \nabla_{QQ'}\big( \widehat{\phi}{}^{PQ A_3\cdots A_m}\zeta_P{}^{Q'}\psi_{A_3\cdots A_m} \big) -\mathfrak{c}^2 \nabla^P{}_{P'}\big(\widehat{\psi}^{A_3\cdots A_m} \zeta^{QP'} \phi_{PQA_3\cdots A_m}
\big)\\
&& \hspace{2cm}  + \phi^{P_1\cdots P_m} W_{P_1\cdots P_m}{}^{Q_1\cdots Q_m}\widehat{\phi}_{Q_1\cdots Q_m} + \phi^{P_1\cdots P_m} X_{P_1\cdots P_m}{}^{Q_3\cdots Q_m}\widehat{\psi}_{Q_3\cdots Q_m}\\
&& \hspace{2cm} +\mathfrak{c} \psi^{P_3\cdots P_m} Y_{P_3\cdots P_m}{}^{Q_1\cdots Q_m}\widehat{\phi}_{Q_1\cdots Q_m} + \mathfrak{c}\psi^{P_3\cdots P_m}Z_{P_3\cdots P_m}{}^{Q_3\cdots Q_m}\widehat{\psi}_{Q_3\cdots Q_m},
\end{eqnarray*}
with
\[
W_{P_1\cdots P_m}{}^{Q_1\cdots Q_m}, \quad X_{P_1\cdots P_m}{}^{Q_3\cdots Q_m}, \quad Y_{P_3\cdots P_m}{}^{Q_1\cdots Q_m}, \quad Z_{P_3\cdots P_m}{}^{Q_3\cdots Q_m},
\]
depending solely on $\chi_{ABCD}$ and $\varpi$.
\end{lemma}

\begin{remark}
{\em In other words, the principal part in the quadratic form 
\[
\mathbf{G}(\phi,\psi)\equiv \mathbf{E}(\phi,\psi)-\mathfrak{c}\mathbf{F}(\phi,\psi)
\]
can be expressed as a total divergence. This is, in some sense, the main observation of this article. The divergence identity in Lemma \ref{Lemma:MainIdentityPhiPsiSystem} provides the connection with the \emph{positive commutator method} described in Subsection \ref{Section:TheScalarWaveEqn}. More precisely, letting $\mathbf{S}$ and $\mathbf{R}$ denote, respectively, the operators associated to the principal part of equations \eqref{SpaceSpinorEvolutionSystem1}-\eqref{SpaceSpinorEvolutionSystem2}, a naive adaptation of the positive commutator method to first order systems would suggest to consider the commutators $[\mathbf{S},\bmzeta]$ and $[\mathbf{R},\bmzeta]$. However, these commutators are, generically, differential operators of first order so that they do not provide any obvious control of a quadratic form involving the pair of spinor fields $(\phi,\psi)$. Instead, the right approach is to consider a suitable contraction of the commutators $[\mathbf{S},\zeta^{\bmA\bmA'}]$ and $[\mathbf{R},\zeta^{\bmA\bmA'}]$ ---that is, the commutator of a first order  and a zero order differential operator. It should be noticed that, in this approach, the construction of estimates is closely related to the hyperbolic reduction procedure using the space spinor formalism.} 
\end{remark}


For future use, define
\begin{eqnarray*}
&& J_{QQ'}\equiv \widehat{\phi}{}^{A_1\cdots A_m}\zeta_{A_1Q'}\phi_{A_2\cdots A_m Q},\\
&& I_{PP'}\equiv \widehat{\psi}^{A_3\cdots A_m}\zeta_{PP'} \psi_{A_3\cdots A_m},\\
&& K^{QQ'}\equiv \widehat{\phi}^{PQA_3\cdots A_m}\zeta_P{}^{Q'}\psi_{A_3\cdots A_m}, \\
&& L_{PP'}\equiv \widehat{\psi}^{A_3\cdots A_m}\zeta^Q{}_{P'}\phi_{PQA_3\cdots A_m}. 
\end{eqnarray*}

One then has the following:

\begin{lemma}
One has that
\begin{eqnarray*}
&& \bar{J}_{A'A} = J_{AA'}, \quad \bar{I}_{A'A} = I_{AA'} \qquad \mbox{for}\quad m \quad \mbox{even},\\
&& \bar{J}_{A'A} = -J_{AA'}, \quad \bar{I}_{A'A} = -I_{AA'} \qquad \mbox{for}\quad m \quad \mbox{odd}.
\end{eqnarray*}
\end{lemma}

\subsubsection{Construction of estimates}
\label{Subsection:GeneralConstructionEstimates}

The starting point for the construction of estimates is the identity in Lemma \ref{Lemma:MainIdentityPhiPsiSystem} written in the form
\begin{eqnarray*}
&& \mathbf{G}(\phi,\psi) = \nabla_{QQ'}I^{QQ'} -\nabla_{QQ'} J^{QQ'} + \mathfrak{c} \nabla_{QQ'} K^{QQ'} +  \mathfrak{c}^2 \nabla_{QQ'} L^{QQ'} \\
&& \hspace{2cm}  + \phi^{P_1\cdots P_m} W_{P_1\cdots P_m}{}^{Q_1\cdots Q_m}\widehat{\phi}_{Q_1\cdots Q_m} + \phi^{P_1\cdots P_m} X_{P_1\cdots P_m}{}^{Q_3\cdots Q_m}\widehat{\psi}_{Q_3\cdots Q_m}\\
&& \hspace{2cm} +\mathfrak{c} \psi^{P_3\cdots P_m} Y_{P_3\cdots P_m}{}^{Q_1\cdots Q_m}\widehat{\phi}_{Q_1\cdots Q_m} + \mathfrak{c}\psi^{P_3\cdots P_m}Z_{P_3\cdots P_m}{}^{Q_3\cdots Q_m}\widehat{\psi}_{Q_3\cdots Q_m}. 
\end{eqnarray*}
Now, let 
\begin{eqnarray*}
    && \mathbf{W}(\phi,\widehat{\phi}) \equiv \phi^{P_1\cdots P_m} W_{P_1\cdots P_m}{}^{Q_1\cdots Q_m}\widehat{\phi}_{Q_1\cdots Q_m}, \quad \mathbf{X}(\phi, \widehat{\psi}) \equiv \phi^{P_1\cdots P_m} X_{P_1\cdots P_m}{}^{Q_3\cdots Q_m}\widehat{\psi}_{Q_3\cdots Q_m}, \\
    && \mathbf{Y}(\psi, \widehat{\phi}) \equiv \psi^{P_3\cdots P_m} Y_{P_3\cdots P_m}{}^{Q_1\cdots Q_m}\widehat{\phi}_{Q_1\cdots Q_m}, \quad \mathbf{Z}(\psi, \widehat{\psi}) \equiv \psi^{P_3\cdots P_m}Z_{P_3\cdots P_m}{}^{Q_3\cdots Q_m}\widehat{\psi}_{Q_3\cdots Q_m},
\end{eqnarray*}
so we can write 
\begin{eqnarray}
&& \mathbf{G}(\phi,\psi)= \nabla \cdot \mathbf{I} -\nabla \cdot \mathbf{J} + \mathfrak{c} \nabla \cdot \mathbf{K} +  \mathfrak{c}^2 \nabla \cdot \mathbf{L} \nonumber  \\
&& \hspace{2cm}  + \mathbf{W}(\phi,\widehat{\phi}) + \mathbf{X}(\phi, \widehat{\psi}) +\mathfrak{c} \mathbf{Y}(\psi, \widehat{\phi}) + \mathfrak{c}\mathbf{Z}(\psi, \widehat{\psi}). 
\label{PhiPsi:DivergenceIdentity}
\end{eqnarray}
Note that the original system can be written as 
\begin{eqnarray*}
&& \mathcal{E}_{A_1\cdots A_m}  + \varpi^{2} G_{(A_1\cdots A_m)}{}^{Q_1\cdots Q_{m}}\varphi_{Q_1\cdots Q_{m}} = \varpi^{2} f_{(A_1\cdots A_m)}, \qquad \\
&& \mathcal{F}_{A_3\cdots A_m} + \varpi^{2} G^P{}_{PA_3\cdots A_m}{}^{Q_1\cdots Q_{m}}\varphi_{Q_1\cdots Q_{m}} = \varpi^{2} f^P{}_{PA_3\cdots A_m}.
\end{eqnarray*}
Multiply the first equation by $\widehat{\phi}^{A_{1} \cdots A_{m}}$ and the second by $\widehat{\psi}^{A_{3} \cdots A_{m}}$ and rearrange to get
\begin{subequations}
    \begin{eqnarray}
        && \widehat{\phi}^{A_{1} \cdots A_{m}}  \mathcal{E}_{A_1\cdots A_m} = \varpi^{2} \widehat{\phi}^{A_{1} \cdots A_{m}} f_{A_1\cdots A_m}- \varpi^{2}\widehat{\phi}^{A_{1} \cdots A_{m}} G_{A_1\cdots A_m}{}^{Q_1\cdots Q_{m}}\varphi_{Q_1\cdots Q_{m}}, \qquad \\
        && \widehat{\psi}^{A_{3} \cdots A_{m}} \mathcal{F}_{A_3\cdots A_m} = \varpi^{2} \widehat{\psi}^{A_{3} \cdots A_{m}} g_{A_3\cdots A_m}- \varpi^{2} \widehat{\psi}^{A_{3} \cdots A_{m}} H_{A_3\cdots A_m}{}^{Q_1\cdots Q_{m}}\varphi_{Q_1\cdots Q_{m}},
    \end{eqnarray}
    \label{System-phi-psi-Aux1}
\end{subequations}
where we have defined 
\begin{equation*}
    H_{A_3\cdots A_m}{}^{Q_1\cdots Q_{m}} \equiv G^P{}_{PA_3\cdots A_m}{}^{Q_1\cdots Q_{m}}, \qquad g_{A_3\cdots A_m} \equiv f^P{}_{PA_3\cdots A_m}.
\end{equation*}
If we further define 
\begin{eqnarray*}
    && \mathcal{G}_{A_1\cdots A_m}{}^{Q_1\cdots Q_{m}} \equiv \varpi^{2} G_{A_1\cdots A_m}{}^{Q_1\cdots Q_{m}}, \qquad \mathcal{H}_{A_3\cdots A_m}{}^{Q_1\cdots Q_{m}} \equiv \varpi^{2} H_{A_3\cdots A_m}{}^{Q_1\cdots Q_{m}}, \\
    && \mathpzc{f}_{A_{1} \cdots A_{m}} \equiv \varpi^{2} f_{A_{1} \cdots A_{m}} , \qquad \mathpzc{g}_{A_{3} \cdots A_{m}}  \equiv \varpi^{2} g_{A_{3} \cdots A_{m}},
\end{eqnarray*}
then, using \eqref{System-phi-psi-Aux1}, we can write 
\begin{eqnarray*}
    && \mathbf{G}(\phi,\psi)+ \mathbfcal{G}(\widehat{\phi},\varphi) + (-1)^{m} \widehat{\mathbfcal{G}}(\phi,\widehat{\varphi})- \mathfrak{c} \mathbfcal{H}(\widehat{\psi},\varphi) - \mathfrak{c} (-1)^{m} \widehat{\mathbfcal{H}}(\psi,\widehat{\varphi})\\
    && \hspace{5cm}  =  \widehat{\phi}\cdot  \mathpzc{f} + (-1)^{m} \phi\cdot \widehat{\mathpzc{f}}  - \mathfrak{c} \widehat{\psi}\cdot \mathpzc{g}  - \mathfrak{c} (-1)^{m} \psi\cdot \widehat{\mathpzc{g}},
\end{eqnarray*}
where we defined
\begin{eqnarray*}
    && {\mathbfcal{G}}(\widehat{\phi},\varphi)\equiv \widehat{\phi}^{A_{1} \cdots A_{m}} \mathcal{G}_{A_1\cdots A_m}{}^{Q_1\cdots Q_{m}}\varphi_{Q_1\cdots Q_{m}}, \qquad \widehat{\mathbfcal{G}}(\phi,\widehat{\varphi}) \equiv \phi^{A_{1} \cdots A_{m}} \widehat{\mathcal{G}}_{A_1\cdots A_m}{}^{Q_1\cdots Q_{m}} \widehat{\varphi}_{Q_1\cdots Q_{m}}, \\
    && \mathbfcal{H}(\widehat{\psi},\varphi) \equiv \widehat{\psi}^{A_{3} \cdots A_{m}} \mathcal{H}_{A_3\cdots A_m}{}^{Q_1\cdots Q_{m}}\varphi_{Q_1\cdots Q_{m}}, \qquad \widehat{\mathbfcal{H}}(\psi,\widehat{\varphi}) \equiv \psi^{A_{3} \cdots A_{m}} \widehat{\mathcal{H}}_{A_3\cdots A_m}{}^{Q_1\cdots Q_{m}} \widehat{\varphi}_{Q_1\cdots Q_{m}},\\
    &&\widehat{\phi}\cdot \mathpzc{f} \equiv  \widehat{\phi}^{A_{1} \cdots A_{m}} \mathpzc{f}_{A_1\cdots A_m}, \qquad \phi\cdot \widehat{\mathpzc{f}} \equiv \phi^{A_{1} \cdots A_{m}} \widehat{\mathpzc{f}}_{A_1\cdots A_m}, \\
    && \widehat{\psi}\cdot \mathpzc{g} \equiv  \widehat{\psi}^{A_{3} \cdots A_{m}} \mathpzc{g}_{A_3\cdots A_m}, \qquad \psi\cdot \widehat{\mathpzc{g}} \equiv \psi^{A_{3} \cdots A_{m}} \widehat{\mathpzc{g}}_{A_3\cdots A_m}. 
\end{eqnarray*}
As in the discussion of Section \ref{Section:TheScalarWaveEqn}, one integrates the above expression over a domain $\mathcal{U}$ so that: 
\begin{eqnarray*}
    && \int_{\mathcal{U}}\mathbf{G}(\phi,\psi)\mathrm{d}\mu_{\bmg} \\
    && \hspace{2cm}+  \int_{\mathcal{U}}\Big(\mathbfcal{G}(\widehat{\phi},\varphi) + (-1)^{m} \widehat{\mathbfcal{G}}(\phi,\widehat{\varphi})\Big) \mathrm{d}\mu_\bmg - \mathfrak{c} \int_{\mathcal{U}}\Big(\mathbfcal{H}(\widehat{\psi},\varphi) + (-1)^{m}  \widehat{\mathbfcal{H}}(\psi,\widehat{\varphi})\Big)\mathrm{d}\mu_\bmg \\
    && \hspace{4cm} =\int_{\mathcal{U}}\Big( \widehat{\phi}\cdot  \mathpzc{f} + (-1)^{m} \phi\cdot \widehat{\mathpzc{f}}\Big) \mathrm{d}\mu_\bmg   - \mathfrak{c} \int_{\mathcal{U}}\Big( \widehat{\psi}\cdot \mathpzc{g}  +  (-1)^{m} \psi\cdot \widehat{\mathpzc{g}}\Big)\mathrm{d}\mu_\bmg.
\end{eqnarray*}

Observing identity \eqref{PhiPsi:DivergenceIdentity}, integration by parts leads to integrals over $\partial \mathcal{U}$ of the normal components of the \emph{currents} $\mathbf{I}$, $\mathbf{J}$, $\mathbf{K}$ and $\mathbf{L}$. Accordingly, one can write 
\begin{eqnarray*}
&& \int_{\mathcal{U}}\Big( \mathbf{W}(\phi,\widehat{\phi}) + \mathbf{X}(\phi,\widehat{\psi}) +\mathfrak{c} \mathbf{Y}(\psi,\widehat{\phi}) +\mathfrak{c} \mathbf{Z}(\psi,\widehat{\psi})  \Big) \mathrm{d}\mu_\bmg \\
 && \hspace{2cm}+  \int_{\mathcal{U}}\Big(\mathbfcal{G}(\widehat{\phi},\varphi) + (-1)^{m} \widehat{\mathbfcal{G}}(\phi,\widehat{\varphi})\Big) \mathrm{d}\mu_\bmg - \mathfrak{c} \int_{\mathcal{U}}\Big(\mathbfcal{H}(\widehat{\psi},\varphi) + (-1)^{m}  \widehat{\mathbfcal{H}}(\psi,\widehat{\varphi})\Big)\mathrm{d}\mu_\bmg \\
    && \hspace{4cm} \approx \int_{\mathcal{U}}\Big( \widehat{\phi}\cdot  \mathpzc{f} + (-1)^{m} \phi\cdot \widehat{\mathpzc{f}}\Big) \mathrm{d}\mu_\bmg   - \mathfrak{c} \int_{\mathcal{U}}\Big( \widehat{\psi}\cdot \mathpzc{g}  +  (-1)^{m} \psi\cdot \widehat{\mathpzc{g}}\Big)\mathrm{d}\mu_\bmg.
\end{eqnarray*}
The first integral on the left-hand side of the above equation is a bilinear form on the spinor fields $\phi_{A_1\cdots A_m}$ and $\psi_{A_3\cdots A_m}$. It provides the basic control over the norms
\[
|| \phi ||^{2} \equiv \int_{\mathcal{U}} \widehat{\phi}^{A_1\cdots A_m} \phi_{A_1\cdots A_m}\mathrm{d}\mu_\bmg, \qquad || \psi ||^{2} \equiv \int_{\mathcal{U}} \widehat{\psi}^{A_3\cdots A_m} \psi_{A_3\cdots A_m}\mathrm{d}\mu_\bmg.
\]
Accordingly, in what follows it is assumed there exist constants $\mathfrak{W},\, \mathfrak{Z}>0$ and further constants $\mathfrak{X},\, \mathfrak{Y}\in\mathbb{R}$ (i.e. not necessarily positive) such that 
\begin{eqnarray*}
&& \mathfrak{W} |\phi|^2  \leq \mathbf{W}(\phi,\widehat{\phi}), \qquad \mathfrak{Z}|\psi|^2  \leq \mathbf{Z}(\psi,\widehat{\psi}), \\
&& \mathfrak{X}\Big( |\phi|^2 + |\psi|^2  \Big)\leq \mathbf{X}(\phi,\widehat{\psi}), \\
&& \mathfrak{Y} \Big( |\phi|^2 +|\psi|^2   \Big)\leq \mathbf{Y}(\psi,\widehat{\phi}).
\end{eqnarray*}
It then follows that 
\begin{eqnarray*}
&& \Big(\mathfrak{W} +\mathfrak{X}+\mathfrak{Y}\Big) ||\phi||^2 + \mathfrak{c} \Big( \mathfrak{Z}+\mathfrak{X} +\mathfrak{Y} \Big)||\psi||^2  \\
 && \hspace{2cm}+  \int_{\mathcal{U}}\Big(\mathbfcal{G}(\widehat{\phi},\varphi) + (-1)^{m} \widehat{\mathbfcal{G}}(\phi,\widehat{\varphi})\Big) \mathrm{d}\mu_\bmg - \mathfrak{c} \int_{\mathcal{U}}\Big(\mathbfcal{H}(\widehat{\psi},\varphi) + (-1)^{m}  \widehat{\mathbfcal{H}}(\psi,\widehat{\varphi})\Big)\mathrm{d}\mu_\bmg \\
    && \hspace{4cm} \preccurlyeq \int_{\mathcal{U}}\Big( \widehat{\phi}\cdot  \mathpzc{f} + (-1)^{m} \phi\cdot \widehat{\mathpzc{f}}\Big) \mathrm{d}\mu_\bmg   - \mathfrak{c} \int_{\mathcal{U}}\Big( \widehat{\psi}\cdot \mathpzc{g}  +  (-1)^{m} \psi\cdot \widehat{\mathpzc{g}}\Big)\mathrm{d}\mu_\bmg.  \end{eqnarray*}
In order to further develop the above inequality, it is observed that 
\begin{eqnarray*}
&& \int_{\mathcal{U}}\Big( \widehat{\phi}\cdot  \mathpzc{f} + (-1)^{m} \phi\cdot \widehat{\mathpzc{f}}\Big) \mathrm{d}\mu_\bmg  = 2 \mbox{Re} \innerbrackets{\phi}{\mathpzc{f}}, \\
&& \int_{\mathcal{U}}\Big( \widehat{\psi}\cdot \mathpzc{g}  +  (-1)^{m} \psi\cdot \widehat{\mathpzc{g}}\Big)\mathrm{d}\mu_\bmg = 2 \mbox{Re}\, \innerbrackets{\psi}{\mathpzc{g}},
\end{eqnarray*}
so that, using Cauchy–Schwarz, Young's inequality and moving terms, one obtains
\begin{eqnarray*}
&& \Big(\mathfrak{W}+\mathfrak{X} +\mathfrak{Y}-1\Big)||\phi||^2 +\mathfrak{c}\Big(\mathfrak{Z}+\mathfrak{X} +\mathfrak{Y}-1\Big)||\psi||^2 \\
&& \hspace{2cm} \preccurlyeq \mathfrak{c} \int_{\mathcal{U}}\Big(\mathbfcal{H}(\widehat{\psi},\varphi) + (-1)^{m}  \widehat{\mathbfcal{H}}(\psi,\widehat{\varphi})\Big)\mathrm{d}\mu_\bmg - \int_{\mathcal{U}}\Big(\mathbfcal{G}(\widehat{\phi},\varphi) + (-1)^{m} \widehat{\mathbfcal{G}}(\phi,\widehat{\varphi})\Big) \mathrm{d}\mu_\bmg\\
&& \hspace{4cm} + ||\mathpzc{f}||^2 + ||\mathpzc{g}||^2.
\end{eqnarray*}
In order to proceed any further, one needs further assumptions on the form of the quadratic forms  $\mathbfcal{H}$ and $\mathbfcal{G}$. Recall that the spinors $\phi$ and $\psi$ are elements of the irreducible decomposition of the spinor $\varphi$. Thus, the estimation of the 
\[
\mathbfcal{H}(\widehat{\psi},\varphi), \qquad \widehat{\mathbfcal{H}}(\psi,\widehat{\varphi}), \qquad  \mathbfcal{G}(\widehat{\phi},\varphi), \qquad \widehat{\mathbfcal{G}}(\phi,\widehat{\varphi})
\]
requires, in principle, knowledge/control of all the components of the irreducible decomposition of $\varphi$ other than $\phi$ and $\psi$. In the following we denote this set by $\mho(\varphi)$. The type of control depends on the structural properties of the particular equation under consideration.
For simplicity of presentation, in the following, it is assumed that there exist positive constants $\mathfrak{H}$, $\mathfrak{H}'$, $\mathfrak{H}''$, $\mathfrak{G}$, $\mathfrak{G}'$ and $\mathfrak{G}''$ such that 
\begin{eqnarray*}
&& \int_{\mathcal{U}}\Big(\mathbfcal{H}(\widehat{\psi},\varphi) + (-1)^{m}  \widehat{\mathbfcal{H}}(\psi,\widehat{\varphi})\Big)\mathrm{d}\mu_\bmg \leq \mathfrak{H}||\psi||^2 +\mathfrak{H}' ||\phi||^2 + \mathfrak{H}'' ||\varphi||^2_{\psi,\phi\notin \mho(\varphi)}, \\
&& -\int_{\mathcal{U}}\Big(\mathbfcal{G}(\widehat{\phi},\varphi) + (-1)^{m} \widehat{\mathbfcal{G}}(\phi,\widehat{\varphi})\Big) \mathrm{d}\mu_\bmg \leq \mathfrak{G}||\phi||^2 +\mathfrak{G}'||\psi||^2 +\mathfrak{G}'' ||\varphi||^2_{\psi,\phi\notin \mho(\varphi)},
\end{eqnarray*}
where $||\varphi||^2_{\psi,\phi\notin \mho(\varphi)}$ denotes the sum of the $L^2$-norms of all the irreducible components of $\varphi$ excluding $\phi$ and $\psi$. 
Under this assumption, it  follows that 
\begin{eqnarray}
&& \Big( \mathfrak{W} + \mathfrak{X} +\mathfrak{Y}-\mathfrak{H}'-\mathfrak{G}-1 \Big)||\phi||^2 + \mathfrak{c}\Big(\mathfrak{Z}+\mathfrak{X}+\mathfrak{Y}-\mathfrak{H}-\mathfrak{G}'-1  \Big)||\psi||^2 \nonumber\\
&& \hspace{2cm} \preccurlyeq  ||\mathpzc{f}||^2 + ||\mathpzc{g}||^2 + (\mathfrak{H}'' +\mathfrak{G}'') ||\varphi||^2_{\psi,\phi\notin \mho(\varphi)}.
\label{PhiPsiEstimate}
\end{eqnarray}
If
\begin{eqnarray*}
&& \mathfrak{W} + \mathfrak{X} +\mathfrak{Y}-\mathfrak{H}'-\mathfrak{G} >1, \\
&& \mathfrak{Z}+\mathfrak{X}+\mathfrak{Y}-\mathfrak{H}-\mathfrak{G}'>1,
\end{eqnarray*}
then the inequality \eqref{PhiPsiEstimate} provides control of the norms $||\phi||$ and $||\psi||$ modulo knowledge about
\[
||\mathpzc{f}||, \quad ||\mathpzc{g}||, \quad ||\varphi||_{\psi,\phi\notin \mho(\varphi)}.
\]

\begin{remark}
{\em The $\phi$-$\psi$ system discussed in this section provides a model for the construction of estimates for spinorial fields satisfying very general equations. A particular case will be analysed in the following section.}
\end{remark}

\section{Symmetric spinor fields}
\label{Section:SymmetricSpinorFields}
In this section, we focus on an important particular subcase of the $\phi$-$\psi$ system  ---namely, when 
\[
\psi_{A_3\cdots A_m}=0,
\]
so that one has a symmetric spinor field of valence $m$ ----namely, $\phi_{A_1\cdots A_m}=\phi_{(A_1\cdots A_m)}$. The equation to be considered in this case is given by 
\begin{equation}
    \nabla^{Q}{}_{A'} \phi_{Q A_{2} \cdots A_{m}} + G_{A' A_{2} \cdots A_{m}}{}^{Q_{1} \cdots Q_{m}} \phi_{Q_{1} \cdots Q_{m}} = f_{A' A_{2} \cdots A_{m}},  \label{Symmetric-spinor-fields-equation} 
\end{equation}
While the general strategy for the construction of estimates for this equation is subsumed by the discussion in Section \eqref{Subsection:EstimatesPhiPsi}, in this section, we focus our attention on the additional structures arising in this particular case.

\medskip
In the following, we assume the same geometric setting used in the construction in Subsection \ref{Subsection:DomainOfIntegration}. Contracting $\eqref{Symmetric-spinor-fields-equation}$ with $\tau_{A_{1}}{}^{A'}$ and symmetrising over $\{A_{1}, \cdots A_{m}\}$, the \emph{evolution} system associated with \eqref{Symmetric-spinor-fields-equation} can be written as
\begin{equation*}
    \nabla^{Q}{}_{(A_{1}} \phi_{A_{2} \cdots A_{m}) Q} + G_{(A_{1}  \cdots A_{m})}{}^{Q_{1} \cdots Q_{m}} \phi_{Q_{1} \cdots Q_{m}} = f_{(A_{1} \cdots A_{m})}.
\end{equation*}

\subsection{Structural properties of the principal part}
We begin looking at the key structural properties of the principal part of equation \eqref{Symmetric-spinor-fields-equation} which allow for the construction of estimates. As in previous sections, let $\bmzeta$ denote the real vector field satisfying \eqref{Definition-zeta-spinorial} and define the operator $\mathcal{E}$ given by
\begin{equation*}
    \mathcal{E}(\phi)_{A_{1} \cdots A_{m}} \equiv \zeta_{(A_{1}}{}^{A'} \nabla^{Q}{}_{|A'|} \phi_{A_{2} \cdots A_{m})Q}.
\end{equation*}
From \eqref{Definition-zeta-spinorial}, we can write $\mathcal{E}(\phi)$ as 
\[
\mathcal{E}(\phi)_{A_{1} \cdots A_{m}}  = \varpi^2  \nabla^{Q}{}_{(A_{1}} \phi_{A_{2} \cdots A_{m})Q}.
\]
As in Section \ref{Subsection:EstimatesPhiPsi}, we define a quadratic form $\mathbf{A}(\phi,\phi)$ by
\begin{equation}
    \mathbf{A}(\phi,\phi) \equiv \mathcal{E}(\phi)_{A_{1} \cdots A_{m}} \widehat{\phi}^{A_{1} \cdots A_{m}}+ (-1)^m \widehat{\mathcal{E}(\phi)}_{A_{1} \cdots A_{m}} \phi^{A_{1} \cdots A_{m}}.
    \label{Quadratic-form-symmetric-spinor-field}
\end{equation}
For ease of presentation, we make use of the notation $\mathcal{E}(\phi) \cdot \widehat{\phi}$ to denote the first term on the right-hand side of equation \eqref{Quadratic-form-symmetric-spinor-field}. Rewriting this term in the form of a divergence, one readily finds that  
\begin{eqnarray}
    && \mathcal{E}(\phi) \cdot \widehat{\phi}  = \zeta_{(A_{1}}{}^{A'} \nabla^{Q}{}_{|A'|} \phi_{A_{2} \cdots A_{m})Q} \widehat{\phi}^{A_{1} \cdots A_{m}} \label{First-term-Quadratic-form} \\
    && \phantom{\mathcal{E}(\phi) \cdot \widehat{\phi}} = \zeta_{A_{1}}{}^{A'} \nabla^{Q}{}_{A'} \phi_{A_{2} \cdots A_{m}Q} \widehat{\phi}^{A_{1} \cdots A_{m}} \nonumber \\
    && \phantom{\mathcal{E}(\phi) \cdot \widehat{\phi}} = \nabla^{Q}{}_{A'} \left( \zeta_{A_{1}}{}^{A'} \phi_{A_{2} \cdots A_{m}Q} \widehat{\phi}^{A_{1} \cdots A_{m}} \right)\nonumber  - \phi_{A_{2} \cdots A_{m} Q} \zeta_{A_{1}}{}^{A'} \nabla^{Q}{}_{A'} \widehat{\phi}^{A_{1} \cdots A_{m}} \nonumber \\
    && \phantom{\mathcal{E}(\phi) \cdot \widehat{\phi}=} \hspace{1cm}- \left( \nabla^{Q}{}_{A'} \zeta_{A_{1}}{}^{A'} \right) \phi_{A_{2} \cdots A_{m}Q} \widehat{\phi}^{A_{1} \cdots A_{m}}. \nonumber
\end{eqnarray}
In order to manipulate the second term in the right-hand side of \eqref{Quadratic-form-symmetric-spinor-field}, it is observed that 
\begin{eqnarray*}
    && \widehat{\mathcal{E}(\phi)}_{A_1\cdots A_m} = \tau_{A_{1}}{}^{A'_{1}} \cdots \tau_{A_{m}}{}^{A'_{m}} \zeta_{(A'_{1}}{}^{A} \nabla^{Q'}{}_{|A|} \bar{\phi}_{A'_{2} \cdots A'_{m})Q'} \\
    && \phantom{\widehat{\mathcal{E}(\phi)}_{A_1\cdots A_m} } = (-1)^{m} \tau_{A_{1}}{}^{A'_{1}} \cdots \tau_{A_{m}}{}^{A'_{m}}  \zeta_{(A'_{1}}{}^{A} \nabla^{Q'}{}_{|A|} \left( \tau^{B_{2}}{}_{A'_{2}} \cdots \tau^{B_{m}}{}_{A'_{m}} \tau^{Q}{}_{Q'}\widehat{\phi}_{B_{1} \cdots B_{m} Q} \right) \\
    && \phantom{\widehat{\mathcal{E}(\phi)}_{A_1\cdots A_m} } = (-1)^{m} \varpi^2 \tau_{A_{1}}{}^{A'_{1}} \cdots \tau_{A_{m}}{}^{A'_{m}} \tau_{(A'_{1}}{}^{A} \tau^{B_{2}}{}_{A'_{2}} \cdots \tau^{B_{m}}{}_{A'_{m}} \tau^{Q}{}_{Q'} \nabla^{Q'}{}_{A} \widehat{\phi}_{B_{2} \cdots B_{m} Q} \\
    && \phantom{\widehat{\mathcal{E}(\varphi)}_{A_1\cdots A_m} }\hspace{1cm} + (-1)^{m} \varpi^2 \tau_{A_{1}}{}^{A'_{1}} \cdots \tau_{A_{m}}{}^{A'_{m}}  \tau_{(A'_{1}}{}^{A} \left(\nabla^{Q'}{}_{|A|} \tau^{B_{2}}{}_{A'_{2}} \right) \tau^{B_{3}}{}_{A'_{3}} \cdots \tau^{B_{m}}{}_{A'_{m}} \tau^{Q}{}_{Q'} \\
    &&\hspace{7cm } \times \widehat{\phi}_{B_{2} \cdots B_{m} Q}  \\
    && \phantom{\widehat{\mathcal{E}(\phi)}_{A_1\cdots A_m}}\hspace{1cm}  + \cdots + (-1)^{m} \varpi^{2} \tau_{A_{1}}{}^{A'_{1}} \cdots \tau_{A_{m}}{}^{A'_{m}} \tau_{(A'_{1}}{}^{A} \tau^{B_{2}}{}_{A'_{2}} \cdots \tau^{B_{m-1}}{}_{A'_{m-1}} \nabla^{Q'}{}_{|A|} \tau^{B_{m}}{}_{A'_{m})} \\
    &&\hspace{7cm} \times \tau^{Q}{}_{Q'} \widehat{\phi}_{B_{2} \cdots B_{m} Q} \\
    && \phantom{\widehat{\mathcal{E}(\phi)}_{A_1\cdots A_m}} \hspace{1cm}+ (-1)^{m} \varpi^2 \tau_{A_{1}}{}^{A'_{1}} \cdots \tau_{A_{m}}{}^{A'_{m}}  \tau_{(A'_{1}}{}^{A} \tau^{B_{2}}{}_{A'_{2}} \cdots \tau^{B_{m}}{}_{A'_{m})} \nabla^{Q'}{}_{A} \tau^{Q}{}_{Q'} \widehat{\phi}_{B_{2} \cdots B_{m} Q}.
\end{eqnarray*}
This last expression can be simplified using the identity \eqref{TauTauEpsilon} and the definition of the Weingarten spinor, equation \eqref{Definition-chi}, so as to obtain 
\begin{eqnarray*}
    && \hspace{-1cm}\widehat{\mathcal{E}(\phi)}_{A_1\cdots A_m} = \varpi^{2} \tau^{Q}{}_{Q'} \nabla^{Q'}{}_{(A_{1}} \widehat{\phi}_{A_{2} \cdots A_{m})} - \sqrt{2} \varpi^2 \tau^{Q}{}_{Q'} \tau_{(A_{2}}{}^{A'_{2}} \chi_{A_{1}}{}^{Q'B_{2}}{}_{|A'_{2}|} \widehat{\phi}_{A_{3} \cdots A_{m}) QB_{2}} - \cdots \\
    && \phantom{\widehat{\mathcal{E}(\varphi)}_{A_1\cdots A_m}} - \sqrt{2} \varpi^2 \tau^{Q}{}_{Q'} \tau_{(A_{m}}{}^{A'_{m}} \chi_{A_{1}}{}^{Q' B_{m}}{}_{|A'_{m}|} \widehat{\phi}_{A_{2} \cdots A_{m})QB_{m}} + \sqrt{2} \varpi^2 \chi_{(A_{1}}{}^{Q'Q}{}_{|Q'|} \widehat{\phi}_{A_{2} \cdots A_{m})Q}.
\end{eqnarray*}

Using the above and \eqref{First-term-Quadratic-form} and substituting in equation \eqref{Quadratic-form-symmetric-spinor-field}, we get
\begin{eqnarray}
    && \mathbf{A}(\phi,\phi) = \nabla^{Q}{}_{A'}J_{Q}{}^{A'} - \phi_{A_{2} \cdots A_{m} Q} \zeta_{A_{1}}{}^{A'} \nabla^{Q}{}_{A'} \widehat{\phi}^{A_{1} \cdots A_{m}} - \left( \nabla^{Q}{}_{A'} \zeta_{A_{1}}{}^{A'} \right) \phi_{A_{2} \cdots A_{m}Q} \widehat{\phi}^{A_{1} \cdots A_{m}} \nonumber \qquad \\
    && \hspace{3cm}+ (-1)^{m} \varpi^{2} \phi^{A_{1} \cdots A_{m}} \tau^{Q}{}_{Q'} \nabla^{Q'}{}_{A_{1}} \widehat{\phi}_{A_{2} \cdots A_{m}Q} \nonumber \\
    && \hspace{3cm} - \sqrt{2} (-1)^{m} \varpi^2  \phi^{A_{1} \cdots A_{m}} \tau^{Q}{}_{Q'} \tau_{A_{2}}{}^{A'_{2}} \chi_{A_{1}}{}^{Q'B_{2}}{}_{|A'_{2}|} \widehat{\phi}_{A_{3} \cdots A_{m} QB_{2}} \nonumber \\
    && \hspace{3cm} - \cdots  - \sqrt{2} (-1)^{m} \varpi^2 \phi^{A_{1} \cdots A_{m}} \tau^{Q}{}_{Q'} \tau_{A_{m}}{}^{A'_{m}} \chi_{A_{1}}{}^{Q' B_{m}}{}_{|A'_{m}|} \widehat{\phi}_{A_{2} \cdots A_{m-1} QB_{m}} \nonumber \\
    && \hspace{3cm} + \sqrt{2} (-1)^{m} \varpi^2 \phi^{A_{1} \cdots A_{m}} \chi_{A_{1}}{}^{Q'Q}{}_{|Q'|} \widehat{\phi}_{A_{2} \cdots A_{m}Q} ,  \label{Quadratic-form-expanded}
\end{eqnarray}
where, keeping the notation of Subsection \ref{Subsubsection:TheMainIdentity}, we have defined 
\begin{equation*}
    J_{QA'} \equiv \zeta_{A_{1}A'} \phi_{A_{2} \cdots A_{m}Q} \widehat{\phi}^{A_{1} \cdots A_{m}}.
\end{equation*}

\begin{remark}
    \emph{The symmetrisation parentheses over the indices $\{ A_{1} \cdots A_{m} \}$ were removed in \eqref{Quadratic-form-expanded} since this is guaranteed by the symmetries of $\phi_{A_1\cdots A_m}$.}
\end{remark}

Now, the second term on the right-hand side of equation \eqref{Quadratic-form-expanded} can be manipulated as follows:
\begin{eqnarray*}
    &&  \phi_{A_{2} \cdots A_{m} Q} \zeta_{A_{1}}{}^{A'} \nabla^{Q}{}_{A'} \widehat{\phi}^{A_{1} \cdots A_{m}} =  \varpi^{2} \phi_{A_{2} \cdots A_{m} Q} \tau_{A_{1}}{}^{Q'} \nabla^{Q}{}_{Q'} \widehat{\phi}^{A_{1} \cdots A_{m}} \\
    && \phantom{\phi_{A_{2} \cdots A_{m} Q} \zeta_{A_{1}}{}^{A'} \nabla^{Q}{}_{A'} \widehat{\phi}^{A_{1} \cdots A_{m}}} =  (-1)^{m+1} \varpi^{2} \phi^{A_{2} \cdots A_{m} Q} \tau^{A_{1}Q'} \nabla_{QQ'} \widehat{\phi}_{A_{1} \cdots A_{m}} \\
    && \phantom{\phi_{A_{2} \cdots A_{m} Q} \zeta_{A_{1}}{}^{A'} \nabla^{Q}{}_{A'} \widehat{\phi}^{A_{1} \cdots A_{m}}} = (-1)^{m+2} \varpi^{2} \phi^{A_{1} \cdots A_{m}} \tau^{Q}{}_{Q'} \nabla^{Q'}{}_{A_{1}} \widehat{\phi}_{ A_{2} \cdots A_{m}Q} \\
    && \phantom{\phi_{A_{2} \cdots A_{m} Q} \zeta_{A_{1}}{}^{A'} \nabla^{Q}{}_{A'} \widehat{\phi}^{A_{1} \cdots A_{m}}} = (-1)^{m} \varpi^{2} \phi^{A_{1} \cdots A_{m}} \tau^{Q}{}_{Q'} \nabla^{Q'}{}_{A_{1}} \widehat{\phi}_{ A_{2} \cdots A_{m}Q}.
\end{eqnarray*}
From this, we see that the second term on the right-hand side of \eqref{Quadratic-form-expanded} cancels with the fourth term so as to obtain 
\begin{eqnarray*}
    && \mathbf{A}(\phi,\phi) = \nabla^{Q}{}_{A'} J_{Q}{}^{A'} - \left( \nabla^{Q}{}_{A'} \zeta_{A_{1}}{}^{A'} \right) \phi_{A_{2} \cdots A_{m}Q} \widehat{\phi}^{A_{1} \cdots A_{m}}  \\
    && \hspace{3cm} - \sqrt{2} (-1)^{m} \varpi^2  \phi^{A_{1} \cdots A_{m}} \tau^{Q}{}_{Q'} \tau_{A_{2}}{}^{A'_{2}} \chi_{A_{1}}{}^{Q'B_{2}}{}_{|A'_{2}|} \widehat{\phi}_{A_{3} \cdots A_{m} QB_{2}} \nonumber \\
    && \hspace{3cm} - \cdots  - \sqrt{2} (-1)^{m} \varpi^2 \phi^{A_{1} \cdots A_{m}} \tau^{Q}{}_{Q'} \tau_{A_{m}}{}^{A'_{m}} \chi_{A_{1}}{}^{Q' B_{m}}{}_{|A'_{m}|} \widehat{\phi}_{A_{2} \cdots A_{m-1} QB_{m}} \\
    && \hspace{3cm} + \sqrt{2} (-1)^{m} \varpi^2 \phi^{A_{1} \cdots A_{m}} \chi_{A_{1}}{}^{Q'Q}{}_{|Q'|} \widehat{\phi}_{A_{2} \cdots A_{m}Q}. 
\end{eqnarray*}
Using the expressions in equations \eqref{Definition-chi}, \eqref{Definition-zeta-spinorial} and the fact that 
\[
    \chi_{ABCD} \equiv  \tau_{B}{}^{B'} \tau_{D}{}^{D'} \chi_{AB'CD'},
\]
we can write $\mathbf{A}(\phi,\phi)$ as
\begin{eqnarray*}
    && \mathbf{A}(\phi,\phi) = \nabla^{Q}{}_{A'} J_{Q}{}^{A'} - (2 \varpi \nabla^{Q}{}_{A_{1}} \varpi + \sqrt{2} \varpi^{2} \chi^{Q}{}_{PA_{1}}{}^{P}) \phi_{A_{2} \cdots A_{m}Q} \widehat{\phi}^{A_{1} \cdots A_{m}}  \\
    && \phantom{\mathbf{A}(\phi,\phi)=}+ \sqrt{2} (-1)^{m} \varpi^{2} \phi^{A_{1} \cdots A_{m}} \chi_{A_{1}}{}^{QP}{}_{A_{2}} \widehat{\phi}_{A_{3} \cdots A_{m} QP}  \\
    && \phantom{\mathbf{A}(\phi,\phi)=}+ \cdots   + \sqrt{2} (-1)^{m} \varpi^{2} \phi^{A_{1} \cdots A_{m}} \chi_{A_{1}}{}^{QP}{}_{A_{m}} \widehat{\phi}_{A_{2} \cdots A_{m-1} QP}  \\
    && \phantom{\mathbf{A}(\varphi,\varphi)=} + \sqrt{2} (-1)^{m} \varpi^{2} \phi^{A_{1} \cdots A_{m}} \chi_{A_{1}}{}^{PQ}{}_{P} \widehat{\phi}_{A_{2} \cdots A_{m}Q}. 
\end{eqnarray*}
From here, a straightforward manipulation yields the remarkably compact expression
\begin{eqnarray*}
    && \mathbf{A}(\phi, \phi) = \nabla^{Q}{}_{A'} J_{Q}{}^{A'} - 2 \varpi \nabla^{Q}{}_{A_{1}} \varpi \phi_{A_{2} \cdots A_{m}Q} \widehat{\phi}^{A_{1} \cdots A_{m}} \\
    && \hspace{3cm} + \sqrt{2} (-1)^{m} \varpi^{2} (m-1) \chi_{A_{1}}{}^{QP}{}_{A_{2}} \widehat{\phi}_{A_{3} \cdots A_{m} QP} \phi^{A_{1} \cdots A_{m}}. 
\end{eqnarray*}

\subsubsection{The analogue of the K-current}
Now, given a $\mathcal{U}\subset \mathcal{M}$  with boundary $\partial \mathcal{U}$, it follows from the previous discussion that we can write $\innerbrackets{\mathbf{A}(\phi, \phi)}{1}$ on $\mathcal{U}$ as
\begin{equation}
    \innerbrackets{\mathbf{A}(\phi, \phi)}{1} \approx \int_{\mathcal{U}} \bmK(\phi,\widehat{\phi}) \mathrm{d}\mu_\bmg, 
    \label{Aphiphi-to-Kphiphi}
\end{equation}
where $\bmK(\phi,\phi)$ is a quadratic form given by 
\begin{equation*}
    \bmK(\phi,\phi) \equiv  K^{A_{1} \cdots A_{m} B_{1} \cdots B_{m}} \phi_{A_{1} \cdots A_{m}} \widehat{\phi}_{B_{1} \cdots B_{m}},
\end{equation*}
and $K^{A_1\cdots A_mB_1\cdots B_m}$ is the valence $2m$ spinor given by
\begin{eqnarray*}
    && K^{A_1\cdots A_mB_1\cdots B_m}\equiv  -2 \varpi \left(\nabla^{A_{1}}{}_{Q} \varpi \right) \epsilon^{QB_{1}} \epsilon^{A_{2}B_{2}} \cdots \epsilon^{A_{m}B_{m}}  \\
    && \hspace{4cm}+ \sqrt{2} (-1)^{m} (m-1) \varpi^{2} \epsilon^{PA_{1}} \epsilon^{QA_{2}} \epsilon^{B_{3} A_{3}} \cdots \epsilon^{B_{m} A_{m}} \chi_{P}{}^{B_{1}B_{2}}{}_{Q}.
\end{eqnarray*}

\begin{remark}
{\em The spinor $K^{A_1\cdots A_mB_1\cdots B_m}$ is the analogue of the spinor $K^{AA'BB'}$ (the K-current) appearing in the analysis of the wave equation in Section  \ref{Section:TheScalarWaveEqn} ---see equation \eqref{Definition:KCurrent}. Clearly, the properties of $K^{A_1\cdots A_mB_1\cdots B_m}$ depend entirely on those of the spinor $\zeta^{AA'}$ and its derivatives. }
\end{remark}

\subsection{Construction of estimates}
We now show how the structures described in the previous subsection can be used to construct estimates for the solutions of equation \eqref{Symmetric-spinor-fields-equation}. Starting from \eqref{Symmetric-spinor-fields-equation}, multiplying by $\zeta_{A_{1}}{}^{A'}$ and symmetrising over $\{A_{1}, \cdots A_{m}\}$, one readily finds that 
\begin{equation}
    \mathcal{E}(\phi)_{A_{1} \cdots A_{m}} + \varpi^{2} G_{(A_{1} A_{2} \cdots A_{m})}{}^{Q_{1} \cdots Q_{m}} \phi_{Q_{1} \cdots Q_{m}} = \varpi^{2} f_{(A_{1} \cdots A_{m})}.
\end{equation}
Then, multiply by $\widehat{\phi}^{A_{1} \cdots A_{m}}$ to get
\begin{equation*}
    \mathcal{E}(\phi) \cdot \widehat{\phi} + \varpi^{2} \mathbf{G}(\widehat{\phi},\phi) = \varpi^{2} f \cdot \widehat{\varphi},
\end{equation*}
where, for ease of presentation, we have used the shorthand notation introduced in Section \ref{Subsection:GeneralConstructionEstimates}. Using the above, we can write
\begin{eqnarray*}
 &&   \mathbf{A}(\phi,\phi) = \mathcal{E}(\phi) \cdot \widehat{\phi}  + (-1)^{m} \widehat{\mathcal{E}(\phi)} \cdot \phi \\
 && \phantom{ \mathbf{A}(\phi,\phi)}= \varpi^{2} \left( f \cdot \widehat{\phi} + (-1)^{m} \widehat{f} \cdot \phi - \mathbf{G}(\widehat{\phi},\phi) - (-1)^{m} \widehat{\mathbf{G}}(\phi,\widehat{\phi}) \right), 
\end{eqnarray*}
where 
\begin{eqnarray*}
    \mathcal{E}(\phi) \cdot \widehat{\phi} \equiv \mathcal{E}(\phi)_{A_{1} \cdots A_{m}} \widehat{\phi}^{A_{1} \cdots A_{m}}, && \widehat{\mathcal{E}(\phi)} \cdot \phi \equiv \widehat{\mathcal{E}(\phi)}_{A_{1} \cdots A_{m}} \phi^{A_{1} \cdots A_{m}}, \\
     f \cdot \widehat{\phi} \equiv f_{A_{1} \cdots A_{m}} \widehat{\phi}^{A_{1} \cdots A_{m}}, &&  \widehat{f} \cdot \varphi \equiv \widehat{f}_{A_{1} \cdots A_{m}}  \phi^{A_{1} \cdots A_{m}}, \\
    \mathbf{G}(\widehat{\phi},\phi) \equiv \widehat{\phi}^{A_{1} \cdots A_{m}} G_{A_{1} \cdots A_{m}}{}^{Q_{1} \cdots Q_{m}} \phi_{Q_{1} \cdots Q_{m}}, &&  \widehat{\mathbf{G}}(\phi,\widehat{\phi}) \equiv \phi^{A_{1} \cdots A_{m}} \widehat{G}_{A_{1} \cdots A_{m}}{}^{Q_{1} \cdots Q_{m}} \widehat{\phi}_{Q_{1} \cdots Q_{m}}.
\end{eqnarray*}
Accordingly, we can write $\innerbrackets{\mathbf{A}(\phi, \phi)}{1}$ as
\begin{equation*}
    \innerbrackets{\mathbf{A}(\phi, \phi)}{1} = \int_{\mathcal{U}} \varpi^{2} \left( f \cdot \widehat{\phi} + (-1)^{m} \widehat{f} \cdot \phi - \mathbf{G}(\widehat{\phi},\phi) - (-1)^{m}\widehat{\mathbf{G}}(\phi,\widehat{\phi}) \right) \mathrm{d}\mu_\bmg,
\end{equation*}
so that using equation \eqref{Aphiphi-to-Kphiphi}, we  obtain 
\begin{equation}
    \int_{\mathcal{U}} \bmK(\phi,\widehat{\phi}) \mathrm{d}\mu_\bmg \approx \int_{\mathcal{U}} \varpi^{2} \left( f \cdot \widehat{\phi} + (-1)^{m} \widehat{f} \cdot \phi - \mathbf{G}(\widehat{\phi},\phi) - (-1)^{m} \widehat{\mathbf{G}}(\phi,\widehat{\phi}) \right) \mathrm{d}\mu_\bmg.
    \label{Symmetric-spinor-fields-estimates-eq1}
\end{equation}
The term on the left-hand side is a bilinear form on $\phi_{A_{1} \cdots A_{m}}$ which, under the right assumptions on $\zeta^{AA'}$ can be used to provide control over the $L^{2}$-norm $|| \phi ||$ of $\phi$. In this spirit, in the following \emph{it is assumed that there exists a constant} $\mathfrak{h} > 0$ \emph{such that} 
\begin{equation}
    \bmK(\phi,\widehat{\phi}) \geq \mathfrak{h} |\phi|^{2},
    \label{Control-Kphiphi}
\end{equation}
where 
\begin{equation*}
     |\phi|^{2} \equiv \phi_{A_{1} \cdots A_{m}} \widehat{\phi}^{A_{1} \cdots A_{m}}.
\end{equation*}
Thus, from \eqref{Control-Kphiphi}, we can write
\begin{equation*}
    \int_{\mathcal{U}} \bmK(\phi,\widehat{\phi}) \mathrm{d}\mu_\bmg  \geq \mathfrak{h} || \phi ||^2,
\end{equation*}
so that, in turn, taking into account inequality  \eqref{Symmetric-spinor-fields-estimates-eq1}, we get
\begin{equation}
    \mathfrak{h} || \phi ||^2 \preccurlyeq   \int_{\mathcal{U}} \varpi^{2} \left( f \cdot \widehat{\phi} + (-1)^{m} \widehat{f} \cdot \phi - \mathbf{G}(\widehat{\phi},\phi) - (-1)^{m}\widehat{\mathbf{G}}(\phi,\widehat{\phi}) \right) \mathrm{d}\mu_\bmg.
    \label{Symmetric-spinor-fields-estimates-eq2}
\end{equation}
Now, observe that 
\begin{equation*}
    \int_{\mathcal{U}} \varpi^{2} \left( f \cdot \widehat{\phi} + (-1)^{m} \widehat{f} \cdot \phi \right) \mathrm{d}\mu_\bmg = 2 \text{Re} \innerbrackets{\phi}{g}, \qquad \text{where } \quad g_{A_{1} \cdots A_{m}} \equiv \varpi^{2} f_{A_{1} \cdots A_{m}},
\end{equation*}
Then, using Cauchy–Schwarz and Young's inequalities, we can write
\begin{equation*}
    2 \text{Re} \innerbrackets{\phi}{g} \leq 2 ||g|| \; ||\phi|| \leq ||g||^{2} + ||\phi||^2.
\end{equation*}
If we further assume there exists a positive constant $\mathfrak{z}$ such that 
\begin{equation}
    \int_{\mathcal{U}} -\varpi^{2} \left(   \mathbf{G}(\widehat{\phi},\phi) + (-1)^{m} \widehat{\mathbf{G}}(\phi,\widehat{\phi}) \right) \mathrm{d}\mu_\bmg \leq \mathfrak{z} ||\phi||^{2},
    \label{BoundsG}
\end{equation}
it then follows that inequality  \eqref{Symmetric-spinor-fields-estimates-eq2} implies
\begin{equation*}
    (\mathfrak{h}-\mathfrak{z} -1) || \phi ||^2 \preccurlyeq   ||g||^{2},
\end{equation*}
which provides suitable control over $|| \phi ||$ in terms of $||g||$ if $\mathfrak{h}$ and $\mathfrak{z}$ satisfy
\begin{equation*}
    \mathfrak{h}-\mathfrak{z} -1 >0.
\end{equation*}

\medskip
We summarise the previous discussion in the following:

\begin{proposition}
Assume that there exist constants $\mathfrak{h},\, \mathfrak{z}>0$ satisfying
\[
 \mathfrak{h}-\mathfrak{z} -1 >0
\]
and such that the bounds \eqref{Control-Kphiphi} and \eqref{BoundsG} hold. Then the solutions to equation \eqref{Symmetric-spinor-fields-equation} satisfy the estimate
\[
 (\mathfrak{h}-\mathfrak{z} -1) || \phi ||^2 \preccurlyeq   ||g||^{2}.
\]
\end{proposition}

\begin{remark}
    {\em The key to the success to the above-outlined strategy for the construction of estimates is the choice of a suitable timelike vector field multiplier $\zeta^{AA'}$ and associated weight $\varpi$. }
\end{remark}

\begin{remark}
    {\em Observe that the above method provides integrated estimates ---that is, estimates on the norm over the domain $\mathcal{U}$. More detailed estimates can be obtained by considering the boundary integrals arising in the calculation. In principle, these integrals allow to obtain control over the solution in terms of initial conditions. }
\end{remark}

\section{Conclusions and outlook}
\label{Conclusions}
In this article, we have developed a strategy to study the properties of solutions to a class of linear spinor equations. The model equation \eqref{PrototypeSpinorialEqn} subsumes a number of spinor equations arising in the abstract analysis of solutions to the Einstein equations, their stability and long terms existence ---in particular, in the context of conformal methods; see \cite{CFEBook}. Our strategy is motivated by the positive commutator method to construct estimates used in the study of the Einstein equations through the microlocal analysis approach of Melrose's school of \emph{Geometric Scattering} ---see e.g. \cite{Mel09,HinVas20}. More precisely, in the present work, we have made use of the space spinor formalism to obtain a first order analogue of positive commutator approach for spinorial equations. The natural application of this new method is, in first instance, to construct estimates which allow to control the behaviour of  massless spin-$s$ fields (symmetric spinors) in a neighbourhood of spatial infinity. Ultimately, one could use a similar approach to study the same question for solutions of the conformal Einstein equations. In the latter case, the relevant spinor equation is the Bianchi equation for the rescaled Weyl curvature. Given the coupled nature of this subsystem, the use of the approach described in this article requires a bootstrap argument to facilitate a \emph{conceptual linearisation} of the equations. These challenging problems will be addressed elsewhere.

\section*{Acknowledgements}
Calculations in this project used the computer algebra system Mathematica with the package xAct \cite{xAct}. MM gratefully acknowledges support from Perimeter Institute for Theoretical Physics through the Fields-AIMS-Perimeter fellowship. Research at Perimeter Institute is supported by the Government of Canada through the Department of Innovation, Science and Economic Development and by the Province of Ontario through the Ministry of Colleges and Universities. MM thanks CAMGSD, IST-ID for the support through projects UIDB/04459/2020 and UIDP/04459/2020 during the initial stage of this project. JAVK thanks the hospitality of the Erwin Schr\"odinger Institute for Mathematics and Physics of the University of Vienna as part of the programme \emph{Carrollian Physics and Holography} in April 2024. This research was supported by the EPSRC grant (EP/X012417/1) \emph{Scattering methods for the conformal Einstein equations}.

\appendix
\appendix
\setcounter{equation}{0}  
\renewcommand{\theequation}{\thesection.\arabic{equation}}

\section{Proof of Proposition \ref{Proposition:IrreducibleDecompositions}}
\label{Appendix:ProofIrreducible decompositions}

We provide a proof of Proposition \ref{Proposition:IrreducibleDecompositions} as the techniques of the proof are used several times in the article.

\begin{proof}
    As in the main text, assume that  $\varphi_{ABC\cdots F}$ is of valence $n$ and use the symbol $\sim$ to indicate that two spinors differ by a linear combination of the outer product of $\epsilon$-spinors and spinor of lower valence. The strategy in the proof is to show that 
    \[
    \varphi_{ABC\cdots EF}\sim\varphi_{(ABC\cdots EF)},
    \]
    so that the statement of the proof follows then recursively. To this end it is then observed that 
    \begin{equation}
n\, \varphi_{(ABC\cdots EF)} =\varphi_{A(BC\cdots EF)} + \varphi_{B(AC\cdots EF)} + \varphi_{C(AB\cdots EF)} + \cdots + \varphi_{F(AB\cdots E)},
\label{BreakingSymmetrisation}
    \end{equation}
and consider the difference between the first and the second term, the first and the third term and so on. One then has that 
\begin{eqnarray*}
&& \varphi_{A(BC\cdots EF)} - \varphi_{B(AC\cdots EF)} = - \varphi^Q{}_{(QC\cdots EF)}\epsilon_{AB}, \\
&& \varphi_{A(BC\cdots EF)} -\varphi_{C(AB\cdots EF)} = -\varphi^Q{}_{(QB\cdots EF)}\epsilon_{AC}, \\
&& \hspace{2cm}   \vdots  \\
&& \varphi_{A(BC\cdots EF)} -\varphi_{F(ABC\cdots E)} =- \varphi^Q{}_{(QBC\cdots E)}\epsilon_{AF}. 
\end{eqnarray*}
The above expressions can be used in equation \eqref{BreakingSymmetrisation} to eliminate the the terms
\[
\varphi_{B(AC\cdots EF)}, \qquad \varphi_{C(AB\cdots EF)}, \cdots \qquad \varphi_{F(ABC\cdots E)},
\]
so as to obtain
\[
\varphi_{(ABC\cdots EF)} = \varphi_{A(BC\cdots EF)} + \frac{1}{n}\varphi^Q{}_{(QC\cdots EF)} \epsilon_{AB} + \cdots + \frac{1}{n}\varphi^Q{}_{(QBC\cdots E)} \epsilon_{AF}.
\]
Accordingly, one can write 
\[
\varphi_{(ABC\cdots EF)} \sim \varphi_{A(BC\cdots EF)}.
\]
\medskip
The above procedure can applied, in turn, to each of 
\[
\varphi^Q{}_{(QC\cdots EF)}, \cdots \varphi^Q{}_{(QBC\cdots E)},
\]
so that
\[
\varphi_{(ABC\cdots EF)}\sim \varphi_{A(BC\cdots EF)}\sim \varphi_{AB(C\cdots EF)} \sim \cdots \sim \varphi_{ABC\cdots (EF)} \sim \varphi_{ABC\cdots EF}.
\]
Thus, the result follows.
\end{proof}

\begin{remark}
{\em In particular, the technique used in the proof provides a strategy to remove the symmetrisation brackets over a given subset of indices. This is done at various points in the arguments of Subsection \ref{Subsection:IrreducibleDecompositionEquation}. } 
\end{remark}

\begin{remark}
    {\em It is important to observe that the lower order spinors like $\varphi^Q{}_{QC\cdots F}$ are not the irreducible components of the original spinor as they contain a dummy index in the symmetrisation. The irreducible components discussed in Subsection \ref{Subsection:IrreducibleDecomposition} do not contain dummy indices inside the symmetrisations.}
\end{remark}

\bibliographystyle{unsrt} 
\bibliography{Newgrbib} 

\end{document}